\documentclass[envcountsame,envcountsect]{llncs}
\usepackage{etex}

\usepackage{frameit,amsmath,color,latexsym,graphics,wrapfig,amssymb}

\usepackage{booktabs,times,epsfig,subcaption}

\usepackage{graphicx}

\usepackage{ifthen}
\usepackage{times,epsfig,amsmath,clrscode3e}
\usepackage{frameit,amsmath,color,latexsym,graphics,wrapfig,textcomp,centernot,enumitem,subcaption,stmaryrd,xcolor,colortbl}
\usepackage[utf8]{inputenc}

\usepackage{hyperref}
\usepackage{cleveref}
\usepackage{url}
\usepackage {tikz}
\usetikzlibrary{arrows,automata}
\usetikzlibrary{positioning}
\usepackage{pifont}
\usepackage{multirow}
\usepackage{bussproofs}

\definecolor{Gray}{gray}{0.85}
\usepackage[linesnumbered,ruled,vlined,algo2e]{algorithm2e}

\newcommand{\setvars}[1]{\ensuremath{\bar{#1}}}

\newcommand{\savespace}{\vspace{-2mm}}
\newcommand{\saveone}{\vspace{-1mm}}

\mathchardef\mhyphen="2D

\newcommand{\sepnodeF}[3]{\ensuremath{{#1}{\pto}#2(#3)}}

\newcommand{\sepnode}[3]{\ensuremath{#1{\mapsto}#2(#3)}}

\newcommand{\seppredF}[2]{\ensuremath{#1(#2)}}
\newcommand{\seppred}[2]{\ensuremath{#1(#2)}}

\newcommand{\base}{\ensuremath{{\D^b}}}
\newcommand{\basecomp}[1]{\ensuremath{{\overline{#1}}}}

\newcommand{\report}[1]{ }

\newcommand{\acm}[1]{ }

\newcommand{\hide}[1]{}
\newcommand{\hideie}[1]{}

\newcommand{\nil}{{\code{null}}}

\newcommand{\emp}{\btt{emp}}

\newcommand{\node}{\btt{v}}
\newcommand{\allEnt}{{\ensuremath{\omega}{-}\tt ENT}}

\newcommand{\enode}{\btt{e}}

\newcommand{\evalbase}{\btt{evaluate}}

\newcommand{\pure}{\ensuremath{\pi}}
\newcommand{\subterm}{\ensuremath{\prec}}
\newcommand{\orderp}{\ensuremath{\prec_{\PName}}}
\newcommand{\orderstarp}{\ensuremath{\prec^*_{\PName}}}

\newcommand{\heap}{\ensuremath{\kappa}}

\newcommand{\constr}{\ensuremath{\Phi}}

\newcommand{\entailCyc}[3]{\ensuremath{#1~{\vdash}_{#3}~#2}}
\newcommand{\entailNCyc}[3]{\ensuremath{#1~{\vdash}~#2}}

\newcommand{\myit}[1]{\textit{#1}}

\def\sep{\code{*}}

\def\FV{\myit{FV}}

\def\unfold{\btt{unfold}}

\def\D{\Delta}

\def\int{\code{int}}

\def\true{\code{true}\,}
\def\false{\code{false}\,}
\def\a{a}
\newcommand{\rulename}[1]{\code{\scriptsize #1}}

\newcommand{\code}[1]{{\small {\ensuremath{\tt #1}}}}

\newcommand{\btt}[1]{{\ensuremath{\tt #1}}}

\newcommand{\locnay}[1]{}
\newcommand{\wnnay}[1]{}
\newcommand{\longnay}[1]{}
\newcommand{\nodo}[1]{}

\newtheorem{defn}{Definition}

\def\sep{\ensuremath{*}}
\newcommand{\pto}{{\scriptsize\ensuremath{\mapsto}}}

\def\Flds{\myit{Fields}}
\def\Store{\myit{Heaps}}

\def\Stack{\myit{Stacks}}
\def\Locations{\myit{Loc}}
\def\Val{\myit{Val}}
\def\Var{\myit{Var}}

\def\iffs{\small \btt{ iff~}}
\def\dom{\myit{dom}}

\newcommand{\atom}{\alpha}

\newcommand{\anon}{\ensuremath{\_\,}}

\newcommand{\ent}{\ensuremath{{\vdash}}}
\newcommand{\imply}{\ensuremath{\Rightarrow}}

\def\Dns{\myit{Node}}
\newcommand{\defsym}{\ensuremath{\overset{\text{\scriptsize{def}}}{=}}}

\newcommand{\sheaps}{\ensuremath{h}}
\newcommand{\sstack}{\ensuremath{s}}
\newcommand{\force}{\ensuremath{\models}}

\newcommand{\form}[1]{\ensuremath{#1}}

\newcommand{\PName}{\mathcal{P}}

\newcommand{\toolname}{{\sc{\code{S2S_{Lin}}}}}

\newcommand{\songbird}{{\sc{\code{Songbird}}}}
\newcommand{\harr}{{{{Harrsh}}}}
\newcommand{\spen}{{\sc{\code{Spen}}}}
 
 \newcommand{\valid}{{\sc{\code{valid}}}}
\newcommand{\invalid}{{\sc{\code{invalid}}}}

\newcommand{\sem}[1]{\ensuremath{\llbracket#1\rrbracket}}
\newcommand{\isclose}[1]{{\sc{\code{is\_closed(#1)}}}}
\newcommand{\unknown}{{\sc{\code{unknown}}}}

\newcommand{\cslpluse}{\ensuremath{\code{SHLIDe}}}

\newcommand{\shid}{{\sc{\code{SHID}}}}

\newcommand{\stool}{{\sc{S2ent}}}

\def\qed{\hfill\ensuremath{\square}}

\newcommand{\subst}[2]{\ensuremath{[#1 {/} #2]}}

\newcommand{\utree}[1]{\ensuremath{{\mathcal T}_{#1}}}

\newcommand{\lbent}{\code{link\_back_{e}}}
\newcommand{\lbshid}{\code{link\_back_{\shid}}}

\newcommand{\cyclic}{\code{Cyclist_{SL}}}

\newcommand{\lemstore}{\ensuremath{{\sc L}}}
\newcommand{\backfun}{\ensuremath{{\mathcal C}}}
\newcommand{\forest}{\ensuremath{{\mathcal F}}}
\newcommand{\et}{\enode}

\newcommand{\ctree}[3]{\ensuremath{{\mathcal C}(#1{\rightarrow}#2, #3)}}

\newcommand{\sub}{\ensuremath{\sigma}}

\newcommand{\entProb}{{\tt QF{\_}ENT{-}SL_{LIN}}}

\newcommand{\conferencepaper}{0} 
\newcommand{\rep}[1]{\ifthenelse{\conferencepaper = 0}{#1}{}}
\newcommand{\conf}[1]{\ifthenelse{\conferencepaper = 0}{}{#1}}
\newcommand{\repconf}[2]{\ifthenelse{\conferencepaper = 0}{#1}{#2}}

\usepackage{booktabs} 
\usepackage{subcaption} 

\begin{document}
 \title{An Efficient Cyclic Entailment Procedure in a Fragment of Separation Logic}

\author{Quang Loc Le$^1$ \and Xuan-Bach D. Le$^2$}
\institute{Department of Computer Science, University College London, United Kingdom \and School of Computing and Information Systems, the University of Melbourne, Australia}
\maketitle

\begin{abstract}
An efficient entailment proof system is essential to
 compositional verification
using separation logic. Unfortunately, existing
decision procedures are either inexpressive or inefficient.
For example, Smallfoot is an efficient  procedure
but
only works with hardwired lists and trees.
Other procedures that can support
general inductive predicates
 run exponentially in time as their
proof search requires back-tracking to deal with
disjunction in the consequent.

In this paper,
we present a decision procedure
that can derive cyclic entailment proofs for general inductive predicates in polynomial time. Our
 procedure is efficient and does not require back-tracking;
 it uses normalisation rules that help avoid the introduction of
 disjunction in the consequent.
Moreover, our decidable fragment is sufficiently expressive: It is based on
 compositional
  predicates and can capture a wide range of data structures, including
sorted and nested list segments,
 skip lists with fast forward pointers, and binary search trees.
We have implemented the proposal in a prototype tool
and evaluated it
over challenging problems taken from a recent separation logic competition.
The experimental results confirm the efficiency
of the proposed system.

\end{abstract}


\section{Introduction}
\label{sec:intro}
%
Separation logic \cite{Ishtiaq:POPL01,Reynolds:LICS02}
 has been very successful in automatically reasoning
 about programs that manipulate
pointer structures.
 Separation logic empowers reusability and scalability through compositional reasoning \cite{Cristiano:NFM:15,Calcagno:POPL:2009}.
Moreover,
a compositional verification system relies on a bi-abduction
procedure which is essentially
based on the entailment proof system.
Entailment is defined as:
Given an
antecedent \form{A} and a consequent \form{C} where
\form{A} and \form{C} are
 formulas in separation logic,
  entailment problem is the act of checking
 whether
 \form{A ~{\models}~ C} is valid.
Thus, an efficient
decision procedure for entailments
is the vital ingredient of an automatic verification system
in separation logic.

To enhance the expressiveness of the assertion language, for example, to specify unbounded heaps and
interesting pure properties (e.g., sortedness, parent pointers), separation logic
 is typically combined with user-defined inductive
predicates
 \cite{Chin:CAV:2011,Loc:CAV:2017,Piskac:CAV:2013}.
In this setting,
one key challenge of
an entailment procedure is the ability to support
induction reasoning over the combination of
 heaps and data content. 
The problem of induction is very difficult, especially
for an automated inductive theorem prover,
where the induction rules are not explicitly stated.
In fact, 
this problem is
undecidable \cite{Timos:FOSSACS:2014}.

Developing a sound and complete
 entailment procedure that could be used for compositional reasoning
is not trivial. While it is unknown how model-based systems e.g., \cite{echenim2021unifying,Enea:FMSD:2017,Gu:IJCAR:2016,Iosif:CADE:13,Katelaan:IJCAI:2018,Jens:TACAS:2019},
 could support compositional reasoning,
there was evidence that proof-based
decision procedures, e.g., Smallfoot
\cite{Berdine:APLAS05} and its variant \cite{8882771}, and Cycomp \cite{Makoto:APLAS:2019},
can be extended to solve the bi-abduction problem, which enables compostional reasoning and scalability \cite{Calcagno:POPL:2009,Loc:CAV:2014}. 
In fact, Smallfoot was the center of the biabductive procedure deployed in Infer \cite{Calcagno:POPL:2009},
which achieved great impact in both academia and industry \cite{Distefano:CACM:2019}.
Furthermore,
Smallfoot is very efficient due to its use of ``exclude-the-middle" rule
in which it can avoid the proof search over the disjunction in the consequent.
However, Smallfoot works for
hardwired lists and binary trees only. In contrast,
Cycomp, a recent complete entailment procedure, is a cyclic proof system without
``exclude-the-middle``,
can support general inductive predicates, but has double exponential
time complexity due to the proof search (and back-tracking)
 in the consequent.

In this paper, we introduce  a cyclic proof system with an ``exclude-the-middle"-styled decision procedure for
decidable
 yet expressive inductive predicates. 
Especially, we show that
our procedure runs
  in  polynomial time when
the maximum number of fields of data structures is
 bounded by a constant. 
The decidable fragment, called {\cslpluse},
contains inductive definitions of
 compositional predicates and pure properties.  These predicates
can capture nested list segments,
 skip lists and trees.
The pure properties of small models
can model
a wide range of common data structures e.g.,
a list with fast forward pointers, a nested list being sorted,
 a tree being a binary search tree \cite{Katelaan:IJCAI:2018,McPeak:CAV:2005}.
 This fragment is much more expressive than Smallfoot's fragment
 and is non-overlapping 
 with Cycomp's one \cite{Makoto:APLAS:2019}: there exist some entailments which our system can handle but Cyccomp could not,
 and vice versa.


Our procedure is a variant of the cyclic proof system,
which was first introduced by Brotherston \cite{Brotherston:05,Brotherston:CADE:11} and has become
 one of the main solutions to induction reasoning. 
Intuitively, a cyclic proof is naturally
represented as a tree of statements (entailments in this paper): the leaves are either axioms or nodes which are
linked back to inner nodes, the root of the tree is the theorem to be proven, and
nodes are connected to one or more children by locally sound proof rules. Alternatively,
 a cyclic proof can be viewed as a tree possibly containing
some backlinks (a.k.a. cycles, e.g., ``C, if B, if C'')
such that the proof satisfies some global soundness condition.
This condition ensures that the proof can be viewed as a proof of {\em infinite descent}.
Particularly, for a cyclic entailment proof with inductive definitions,
if every cycle
 contains an unfolding of some inductive predicate, then that
 predicate is infinitely often reduced into a strictly ``smaller" predicate;
 this is impossible as the semantics of inductive definitions
 only allows
 finite steps of unfolding. Hence, that
 proof path with the cycle can be disregarded.

The proposed system advances Brotherston's system in three ways.
First,  the proposed proof search algorithm is specialized
to {\cslpluse} in which it includes ``exclude-the-middle``
rules and excludes
any back-tracking.
The existing proof procedures typically search
for a proof (and back-track)
 over disjunctive cases generated from
 unfolding inductive predicates in the RHS of an entailment.
To avoid such costly searches, we propose a ``exclude-the-middle``-styled normalized rules
in which unfolding of inductive predicates in the RHS always
produces one disjunct. Therefore,
our system is much more efficient than existing systems. 
Second, while a standard Brotherston system is incomplete, our proof search is 
complete in {\cslpluse}: If it is stuck (i.e., it can not apply any inference rules) 
then the root entailment is invalid. 

Lastly, while the global soundness in
 \cite{Brotherston:CADE:11}
must be checked globally and explicitly, 
every backlink generated in {\cslpluse}
 is sound by design. 
%
We note that Cycomp, introduced in \cite{Makoto:APLAS:2019},
was the first work to show completeness of a cyclic proof system. However, in contrast to
ours, it did not discuss the global soundness condition, which is the key idea
attributing to the
soundness of cyclic proofs.


\paragraph{Contributions} Our primary contributions are summarized as follows.
\begin{itemize}
\item
We present a novel decision procedure, called {\toolname}, 
 for the entailment problem
 in separation logic with
 inductive definitions of compositional predicates. 
\item
 We provide a complexity analysis of the procedure.
\item
We have implemented the proposal in a prototype tool, called {\toolname},
 and tested it with the SL-COMP 2022 benchmarks \cite{slcomp:sl:22,Sighireanu:TACAS:19}.
The experimental results
  show that {\toolname} is effective and efficient
when compared with
the state-of-the-art solvers.
\end{itemize}

\paragraph{Organization} The remainder of the paper is organised as follows.
Sect. \ref{sec:spec} describes syntax of
formulas in
fragment {\cslpluse}.
Sect. \ref{sec:overview} presents the basics of an ``exclude-the-middle" proof system and cyclic proofs.
 Sect. \ref{sec.entail} elaborates the result,
the novel cyclic proof system including
an illustrative example.
Sect. \ref{sec.deci} discusses
 the soundness and completeness.
Sect. \ref{sec.impl} presents the implementation
and evaluation.
Sect. \ref{sec.related} discusses related work.
 Finally, Sect. \ref{sec.conc}
concludes the work.
All proofs are available in Appendix.

\section{Decidable Fragment {\cslpluse}}\label{sec:spec}
Subsection \ref{sec.spec} presents syntax of separation logic formulae
and
recursive definitions of linear predicates
and local properties.
Subsection \ref{sec.spec.sem} shows semantics.

 \subsection{Separation Logic Formulas}\label{sec.spec}

Concrete heap models assume a fixed finite collection {\Dns},
a fixed finite collection {\Flds},
a set {\Locations} of locations (heap addresses),
a set of non-addressable values {\Val}, with the requirement that \form{{\Val} {\cap} {\Locations}  {=} \emptyset} (i.e., no pointer arithmetic). \form{\nil} is a special element of {\Val}.
\form{\mathbb{Z}} denotes the set of integers
(\form{\mathbb{Z} {\subseteq} \Val})
and \form{k} denotes integer numbers.
 \form{\Var} an infinite set of variables,
\form{\bar{v}} a sequence of variables.

\paragraph{Syntax}
  Disjunctive formula \form{\constr}, symbolic heaps \form{\D},
spatial formula {\heap}, pure formula \form{\pure}, pointer (dis)equality $\phi$, and (in)equality formula
$\alpha$ are as follows.
 \[
\begin{array}{ll}
\begin{array}{l}
\constr ::= \D ~|~ \constr ~{\vee}~ \constr
\qquad  ~ \D ::= \heap{\wedge}\pure \mid \exists {v}{.}~\heap{\wedge}\pure
 \\
 \heap ::= \emp ~|~
\sepnodeF{x}{c}{f{:}v,..,f{:}v} ~|~
\seppredF{\code{P}}{\setvars{v}} 
~|~\heap {\sep}\heap \\
 \pure ::=  \true \mid
   \atom \mid \phi \mid {\neg} \pure \mid
\exists v{.}~\pure \mid
\pure {\wedge} \pure \\
\end{array}
&\quad
\begin{array}{l}
\phi ::=  v {=}v \mid v{=}\nil  \\
\atom ::=   \a {=} \a \mid \a {\leq} \a \\
 a ::= \!k \mid v
\end{array}
\end{array}
\]
where \form{v {\in} \Var}, \form{c {\in} {\Dns}} and \form{f{\in} {\Flds}}.
Note that we often discard field names \form{f}
 of points-to predicates \form{\sepnodeF{x}{c}{f{:}v,..,f{:}v}} and use the
 short form as \form{\sepnodeF{x}{c}{\setvars{v}}}.
 \form{v_1 {\neq} v_2}
is the short form of \form{\neg (v_1{=}v_2)}.
\form{E} denotes for either a variable or \form{\nil}.
\form{\D\subst{E}{v}} denotes the formula obtained from
\form{\D} by
substituting \form{v} by \form{E}.
 {\em A symbolic heap is referred as a base, denoted as \form{\base}, if it does not
 contain any occurrence of inductive predicates.}

\hide{\noindent{\bf Inductive Definitions}
 Each data structure is defined using the keyword \code{data} and
each inductive predicate is defined by a disjunction \form{\constr}
using the key word \code{pred}.
In each disjunct (called a branch), we require that variables which are not formal parameters must
be existentially quantified.
A predicate is called {\em complete} if its branches are pairwise disjoint.

Given \form{ \seppred{\code{pred~P}}{\setvars{v}}  {\equiv} \constr;},
we use \form{\unfold(\seppred{\code{P}}{\setvars{t}})}
to denote an unfolding of the inductive predicate:
\form{\unfold(\seppred{\code{P}}{\setvars{t}}) {\equiv} \constr[\setvars{t}/\setvars{v}]}. Given the definition \form{\constr} of predicate \form{\seppred{\code{P}}{\setvars{t}}},
a predicate occurrence \form{\seppred{\code{Q}}{\setvars{v}}}  is called a subterm of \form{\seppred{\code{P}}{\setvars{t}}}
if it is a subformula of \form{\constr}. This relation is
 captured by the operation \form{\seppred{\code{Q}}{\setvars{v}} {\subterm} \seppred{\code{P}}{\setvars{t}}}.
The subterm has transitivity property: if \form{\seppred{\code{P_1}}{\setvars{v}_1} {\subterm} \seppred{\code{P_2}}{\setvars{v}_2}} and \form{\seppred{\code{P_2}}{\setvars{v}_2} {\subterm} \seppred{\code{P_3}}{\setvars{v}_3}}, then \form{\seppred{\code{P_1}}{\setvars{v}_1} {\subterm} \seppred{\code{P_3}}{\setvars{v}_3}}.
A parameter \form{v \in \setvars{v}} is called rootable if
\form{v} is allocated in the definition i.e., \form{\sepnodeF{v}{\_}{\_} \in \constr}.}

\paragraph{Inductive Definitions}\label{spec.deci.ent}

 We write
\code{\PName} to denote a set of \form{n} defined predicates 
 \form{\code{\PName}{=}\{\code{P_1},...,\code{P_n}\}} in our system.
Each inductive predicate has following types of parameters:
a pair of root and segment defining segment-based linked points-to heaps,
reference parameters (e.g., parent pointers, fast-forwarding pointers),
transitivity parameters
(e.g., singly-linked lists
where every heap cell contains the same value \form{a})
and pairs of ordering parameters
(e.g., 
trees being binary search trees). 
 An inductive predicate is defined as
\[
\begin{array}{l}
 \seppred{\code{pred~P}}{r{,} F{,} \setvars{B}{,}u{,}sc{,}tg} ~{\equiv}~ \emp {\wedge} r{=}F {\wedge} sc{=}tg \\
~\qquad \vee ~ \exists {X}_{tl}, \setvars{Z},sc'.
\sepnode{r}{c}{X_{tl}{,}\setvars{p}{,}u{,}sc'} ~{\sep}~ \heap' ~{\sep}~ \seppred{\code{P}}{{X}_{tl}{,} F{,} \setvars{B}{,}u{,}sc'{,}tg} \wedge
 r{\neq} F \wedge  sc \diamond sc'
\end{array}
\]
where \form{r} is the root, \form{F} the segment,
 \setvars{B} the borders, \form{u} the parameter for a transitivity property, \form{sc} and \form{tg} source and
 target, respectively, parameters of an
 order
 property, \form{\sepnode{r}{c}{X_{tl}{,}\setvars{p}{,}u{,}sc'} ~{\sep}~ \heap'}
 the matrix of the heaps, and \form{\diamond \in \{=,\geq,\leq\}}. 
(The extension for multiple local properties is straightforward.)
Moreover, this definition
is constrained by the following
three conditions on heap connectivity, establishment, and termination.

\noindent{\bf Condition C1.} In the recursive rule, 
\form{ \setvars{p} = \{\nil\} {\cup} \setvars{Z}}.
This condition implies that 
If two variables points to the same heap,
their content must be the same. 
For instance, 
the following definition of singly-linked lists of even length does not
satisfy this condition.
\saveone\[\saveone
\begin{array}{l}
 \seppred{\code{pred~ell}}{r{,}F} ~{\equiv}~ \emp {\wedge} r{=}F  ~\vee ~ \exists {x_{1}}{,}{X} .
\sepnode{r}{c_1}{x_{1}} {\sep}\sepnode{x_1}{c_1}{X}  {\sep} \seppred{\code{ell}}{X{,}F} {\wedge} r{\neq}F
\end{array}
\]
as \form{n_3} and \form{X} are not field variables of the node pointed-to by \form{r}.

\noindent {\bf Condition C2.}
The matrix heap defines nested and connected list segments as:
\saveone\[\saveone
\form{\heap' {:=} \seppred{\code{Q}}{Z{,}\setvars{U}} \mid \heap'{\sep}\heap' \mid \emp }
\] where
 \form{Z {\in} \setvars{p}} and \form{(\setvars{U} \setminus \setvars{p}) \cap Z = \emptyset}. This condition ensures connectivity (i.e. all
 allocated heaps are connected to the root)
 and establishment (i.e. every existential quantifier  either is allocated or equals to a parameter).

\noindent{\bf Condition C3.} There is no mutual recursion.
We define an order \form{\orderp} on inductive predicates as:
 \form{\code{P} ~{\orderp}~ \code{Q}}
if  at least one occurrence of predicate \code{Q} appears
in the definition of \code{P} and \code{Q}
is called a direct sub-term of \code{P}.
We use \form{{\orderstarp}} to denote the transitive closure
 of \form{\orderp}.

\hide{ 
\begin{itemize}
\item {\bf C3.1} Injectivity: \form{ sc'{=}r } and \form{sc \in \setvars{p}}. 
\item {\bf C3.2} Order: \form{sc{\diamond}sc'{+}k } where  \form{k \in \mathbb{Z}},
\form{\diamond \in \{=,\geq,\leq\}} 
\end{itemize}}

Several definition examples are shown as follows.
\[
\begin{array}{l}
\seppred{\code{pred~ll}}{r{,}F} ~{\equiv}~ \emp {\wedge} r{=}F ~ \vee ~ \exists {X_{tl}}.
 \sepnode{r}{c_1}{X_{tl}} {\sep} \seppred{\code{ll}}{X_{tl},F} {\wedge} r{\neq}F \\
 \seppred{\code{pred~nll}}{r{,}F{,}B} ~{\equiv}~ \emp {\wedge} r{=}F \\
 \quad \vee ~ \exists X_{tl}{,}Z.
 \sepnode{r}{c_3}{X_{tl}{,}Z} {\sep}  \seppred{\code{ll}}{Z,B}{\sep} \seppred{\code{nll}}{X_{tl}{,}F{,}B} {\wedge} r{\neq}F\\
 \seppred{\code{pred~skl1}}{r{,}F} ~{\equiv}~ \emp {\wedge} r{=}F
~\vee ~ \exists {X_{tl}}.
 \sepnode{r}{c_4}{X_{tl}{,}\nil{,}\nil} {\sep}  \seppred{\code{skl1}}{X_{tl},F} {\wedge} r{\neq}F  \\
\seppred{\code{pred~skl2}}{r{,}F} ~{\equiv}~ \emp {\wedge} r{=}F \\
 \quad \vee ~ \exists {X_{tl},Z_1}.
 \sepnode{r}{c_4}{Z_1{,}X_{tl}{,}\nil} {\sep} \seppred{\code{skl1}}{Z_1{,}X_{tl}} {\sep}  \seppred{\code{skl2}}{X_{tl},F} {\wedge} r{\neq}F  \\
 \seppred{\code{pred~skl3}}{r{,}F} ~{\equiv}~ \emp {\wedge} r{=}F \\
 \quad \vee ~ \exists {X_{tl}{,}Z_1{,}Z_2}.
 \sepnode{r}{c_4}{Z_1{,}Z_2{,}X_{tl}} {\sep} \seppred{\code{skl1}}{Z_1{,}Z_2} {\sep}  \seppred{\code{skl2}}{Z_2{,}X_{tl}} {\sep}  \seppred{\code{skl3}}{X_{tl}{,}F} {\wedge} r{\neq}F  \\
\seppred{\code{pred~tree}}{r{,}B} ~{\equiv}~ \emp {\wedge} r{=}B  \\
~\quad \vee ~ \exists {r_l}, {r_r}.
\sepnode{r}{c_{t}}{r_l{,}r_r} {\sep} \seppred{\code{tree}}{r_l{,}B} {\sep} \seppred{\code{tree}}{r_r{,}B} \wedge r{\neq}B
\end{array}
\]
 \code{ll} defines singly-linked lists, \code{nll} defines lists of acyclic lists,
 \code{slk1}, \code{slk2} and \code{slk3} define skip-lists.
Finally, \code{tree} defines binary trees.
\hide{ In intuition, given an inductive predicate
 \form{\seppred{\code{P}}{r{,} F{,} \setvars{B}}} in {\cslpluse}
and \form{\sstack,\sheaps \models \seppred{\code{P}}{r{,} F{,} \setvars{B}}}
then the heap \form{\sheaps} defines a collection of paths
starting from \form{\sstack(r)}
and ending at either \form{\sstack(F)} or \form{\sstack(b)}
 where \form{b \in \setvars{B}}.}
We extend predicate \code{ll} with
transitivity and order parameters to obtain
predicate 
 \code{lla}
and \code{lls}, respectively, as follows.
\[
\begin{array}{l}
\seppred{\code{pred~lla}}{r{,}F{,}a} ~{\equiv}~ \emp {\wedge} r{=}F ~ \vee ~ \exists X_{tl}. \sepnode{r}{c_2}{X_{tl}{,}a} \sep \seppred{\code{lla}}{X_{tl}{,}F{,}a}  {\wedge} r{\neq}F \\
 \seppred{\code{pred~lls}}{r{,}F{,}mi{,}ma} ~{\equiv}~ \emp {\wedge} r{=}F {\wedge}ma{=}mi \\
 \quad \vee ~  \exists X_{tl}{,}mi_1. \sepnode{r}{c_4}{X_{tl}{,}mi_1} \sep \seppred{\code{lls}}{X_{tl}{,}F{,}mi_1{,}ma} {\wedge} r{\neq}F \wedge mi {\leq} mi_1\\
 \end{array}
 \]

\paragraph{Unfolding}
Given \form{\seppred{\code{pred~P}}{\setvars{t}}\equiv \constr} and a formula \form{\seppred{\code{P}}{\setvars{v}}{\sep}\D}, then
unfolding $\seppred{\code{P}}{\setvars{v}}$ means replacing \form{\seppred{\code{P}}{\setvars{v}}} by \form{\constr[\setvars{v}/\setvars{t}]}.
We annotate a number, called unfolding number, for each occurrence of inductive predicates.
Suppose
$\exists \setvars{w}.
\sepnode{r}{c}{\setvars{p}} ~{\sep}~ \seppred{\code{Q}_1}{\setvars{v}_1}{\sep}...{\sep}\seppred{\code{Q}_m}{\setvars{v}_m} ~{\sep}~ \seppred{\code{P}}{\setvars{v}_0} {\wedge} \pure$
be
the recursive rule,
then
in the unfolded formula, if \form{\seppred{\code{P}}{\setvars{v}_0[\setvars{v}/\setvars{t}]}^{k_1}} and
\form{\code{Q_i}(...)^{k_2}} are direct
sub-terms of \form{\seppred{\code{P}}{\setvars{v}}^{k}} like above, then \form{k_1{=}k{+}1} and \form{k_2 = 0}.
When it is unambiguous, we discard
the annotation of the unfolding number for simplicity.

\subsection{Semantics}\label{sec.spec.sem}

The program state is interpreted by a pair \form{(\sstack{,}\sheaps)} where \form{\sstack {\in} {\Stack}}, \form{\sheaps {\in} {\Store} }
and stack ${\Stack}$ and heap ${\Store}$ are defined
as:
\[
\begin{array}{lcl}
{\Store}  & {\defsym} &  {\Locations} {\rightharpoonup_{fin}} ({\Dns} ~{\rightarrow}~
(\Flds ~{\rightarrow}~ \Val \cup \Locations)^m)  \\
{\Stack} & {\defsym} &  {\Var} ~{\rightarrow}~ \Val \cup \Locations
\end{array}
\]
Note that we assume that every data structure contains at most \form{m} fields.
Given  a formula \form{\constr},
its semantics is given by a relation:
\form{\sstack{,}\sheaps ~{\force}~ \constr}
in which the stack \form{\sstack} and the heap \form{\sheaps} satisfy the constraint
 \form{\constr}.
The semantics is shown below
\[
 \begin{array}{lcl}
\form{\sstack},\form{\sheaps} \force \emp &
\iffs & \dom({\sheaps}) {=} \emptyset\\
\form{\sstack},\form{\sheaps} \force \sepnodeF{v}{c}{f_i:v_i} &
\iffs &  {\dom}(\sheaps) {=}\{\sstack(v)\}, \sheaps(\sstack(v)){=}g,
 g(c,f_i){=} \sstack({v_i})\\
\form{\sstack},\form{\sheaps}  \force \seppredF{{P}}{\setvars{v}} & \iffs & \form{(\sheaps,\sstack(\setvars{v}_1),..,\sstack(\setvars{v}_k)) \in \sem{P}} \\
\form{\sstack},\form{\sheaps} \force  \heap_1 \sep \heap_2 &
\iffs & \exists \sheaps_1,\sheaps_2 ~s.t~ \sheaps_1 {\#} \sheaps_2 \mbox{, }
 \sheaps{=}\sheaps_1 {\cdot} \sheaps_2, \mbox{, }
 \sstack,\sheaps_1 \force \heap_1 \mbox{ and } \sstack,\sheaps_2 \force \heap_2\\
\form{\sstack},\form{\sheaps} \force \true &
\iffs & \mbox{always} \\
\form{\sstack},\form{\sheaps} \force \heap{\wedge}\pure &
\iffs &  {\sstack},\form{\sheaps} \force {\heap}  \text{ and } {\sstack} \force {\pure} \\
\form{\sstack},\form{\sheaps} \force \exists {v} {.}\D &
\iffs & \exists {\alpha} . {\sstack}[v \pto {\alpha}],\form{\sheaps} \force {\D} \\
\form{\sstack},\form{\sheaps} \force  \constr_1 \vee \constr_2 &
\iffs & \sstack,\sheaps \force \constr_1 \mbox{ or } \sstack,\sheaps \force \constr_2\\
\end{array}
\]
 \form{dom(g)} is the domain of \form{g}, 
 \form{\sheaps_1 {\#} \sheaps_2} denotes disjoint
 heaps $h_1$ and $h_2$ i.e.,
 \form{{\dom}(\sheaps_1) {\cap} {\dom}(\sheaps_2) {=} \emptyset},
and
  \form{\sheaps_1 {\cdot} \sheaps_2} denotes the
  union of two disjoint heaps.
If \form{\sstack} is a stack, \form{v {\in} \Var},  and
\form{\alpha {\in} \Val {\cup} \Locations}, we write
\form{ {\sstack}[v {\pto} {\alpha}] = {\sstack}}
if \form{v {\in} \dom(\sstack)}, otherwise
\form{ {\sstack}[v {\pto} {\alpha}] = {\sstack} {\cup} \{(v, \alpha)\}}.
 Semantics of non-heap (pure) formulas 
is omitted for simplicity.
The interpretation of an inductive predicate
\form{\seppredF{\code{P}}{\setvars{t}}}
is based on the 
least fixed point semantics \sem{P}.

Entailment
\form{\D \models \D'} holds iff for all \form{\sstack} and \form{\sheaps},
if \form{\sstack,\sheaps \force \D} then \form{\sstack, \sheaps \force \D'}.

\section{Entailment Problem \& Overview}
\label{sec:overview}

Throughout this work, we consider the following problem.
\[
\begin{array}{|ll|}
\hline
 \quad \text{PROBLEM:} & {\entProb}.  \\
 \quad \text{INPUT:} & \form{\D_a \equiv \heap_a{\wedge}\pure_a} \text{ and } \form{\D_c \equiv \heap_c{\wedge}\pure_c}\text{ where } \FV(\D_c) \subseteq \FV(\D_a)\cup\{\nil\}.  \quad\\
 \quad \text{QUESTION:} & \text{Does } \form{\D_a ~\models~ \D_c} \text{ hold? }  \qquad\\
\hline
\end{array}
\]

  An entailment, denoted as \form{\enode}, is syntactically  formalized as:
  \form{\entailNCyc{\D_a}{\D_c}{\heap}}
  where
 \form{\D_a} and \form{\D_c} are quantifier-free formulas whose syntax are defined in the preceding section.

\label{mov.derv.tree}
In Sect. \ref{sec.exmid}, we present the basis of
an exclude-the-middle proof system and our approach
to {$\entProb$}. In
Sect. \ref{sec.cyclic.proof}, we describe the foundation
of cyclic proofs.

\subsection{Exclude-the-middle proof system}\label{sec.exmid}
Given a goal \form{\entailNCyc{\D_a}{\D_c}{\heap}},
an entailment proof system might derive
entailments with disjunction in the right-hand side (RHS).
Such an entailment can be obtained by a proof rule that replaces an inductive predicate by its definition rules.
Authors of Smallfoot \cite{Berdine:APLAS05} introduced a normal form and proof rules to prevent such entailments 
when the predicate are lists or trees. Basically, Smallfoot considers
the following two scenarios.
\begin{itemize}
    \item Case 1 (Exclude-the-middle and Frame): The inductive predicate matches with a points-to predicate in the left hand side (LHS). For instance, the entailment is of the form $\enode_1: \sepnodeF{x}{c}{z} \sep \D ~\ent~ \seppred{\code{ll}}{x,y} \sep \D'$, where \code{ll} is singly-linked lists and $\seppred{\code{ll}}{x,y}$ matches with $\sepnodeF{x}{c}{z}$ as they have the same root $x$. To discharge $\enode_1$, a typical proof system
    might search for a proof through two definition rules of  predicate \code{ll} (i.e., by unfolding $\seppred{\code{ll}}{x,y}$ into two disjuncts):
    One includes the base case with $x=y$ and another contains the recursive case with $x\neq y$. Smallfoot prevents such unfolding by excluding the middle in the LHS:
    It reduces the entaiment into two premises: $\sepnodeF{x}{c}{z} \sep \D \land x=y ~\ent~ \seppred{\code{ll}}{x,y} \sep \D'$ and $\sepnodeF{x}{c}{z} \sep \D \land x\neq y ~\ent~ \seppred{\code{ll}}{x,y} \sep \D'$. The first one
    considers the base case of the list (that is, $\seppred{\code{ll}}{x,x}$) and is equivalent to
    $\sepnodeF{x}{c}{z} \sep \D \land x=y ~\ent~ \D' $.
    And the second premise checks the inductive case of the list
    and is equivalent to $\D \land x\neq y ~\ent~ \seppred{\code{ll}}{x,z} \sep \D'$.
    \item Case 2 (Induction proving via hard-wired Lemma). The inductive predicate matches other inductive predicates in the LHS. For example, the entailment is of the form $\enode_2:\seppred{\code{ll}}{x,z} \sep \D ~\ent~ \seppred{\code{ll}}{x,\nil} \sep \D'$. Smallfoot handle $\enode_2$ by using a proof rule as the consequence of applying
    the following hard-wired lemma $\seppred{\code{ll}}{x,z} \sep \seppred{\code{ll}}{z,\nil} \models \seppred{\code{ll}}{x,\nil}$ and reduces the entailment
    to $\D ~\ent~ \seppred{\code{ll}}{z,\nil} \sep \D'$.
\end{itemize}
In doing so, Smallfoot does not introduce a disjunction in the RHS. However, as it uses specific lemmas in the induction reasoning, it only works for the hardwired lists.

In this paper, we propose {\toolname} as an exclude-the-middle's system
for user-defined predicates, those in {\cslpluse}. In stead of using hardwired lemmas, we apply cyclic proofs for
induction reasoning. For instance, to discharge the entailment $\enode_2$ above,  {\toolname} first unfolds
$\seppred{\code{ll}}{x,z}$ in the LHS and obtains two premises:
\begin{itemize}
    \item $\enode_{21}: (\emp \land x=z) \sep \D  ~\ent~ \seppred{\code{ll}}{x,\nil} \sep \D'$; and
    \item $\enode_{22}: (\sepnodeF{x}{c}{y} \sep \seppred{\code{ll}}{y,z} \land x \neq z) \sep \D  ~\ent~ \seppred{\code{ll}}{x,\nil} \sep \D'$
\end{itemize}
While it reduces $\enode_{21}$
 to $\D[z/x] ~\ent~ \seppred{\code{ll}}{z,\nil} \sep \D'[z/x]$, for $\enode_{22}$, it further applies the frame rule as in {\bf Case 1} above and obtains $ \seppred{\code{ll}}{y,z} \sep \D \land x \neq z ~\ent~ \seppred{\code{ll}}{y,\nil} \sep \D'$. Then, it makes a backlink between the latter and $\enode_2$ and closes this path. By doing so, it does not introduce disjunctions in the RHS and can handle
 user-defined predicates.

\subsection{Cyclic proofs}\label{sec.cyclic.proof}
 Central to our work is a procedure that  construct a
 cyclic proof for an entailment.
Given an entailment \form{\D~ \ent ~\D'},
if our system can derive a cyclic proof, then \form{\D \models \D'} holds.
If, instead, it is stuck without a proof, then \form{\D \models \D'} is not valid.

The procedure includes proof rules, each of which
  is of the form:
\[
\AxiomC{
$
 \enode_1   \quad ... \quad \enode_n
$
}
\LeftLabel{\scriptsize \code{PR_0}}
\RightLabel{\scriptsize \code{cond}}
\UnaryInfC{$
 \enode
$}
\DisplayProof
\]
where entailment \form{\enode} (called the conclusion) is reduced to
entailments \form{\enode_1}, ..,\form{\enode_n}
(called the premises)
through inference rule \code{PR_0}
given that the {\em side condition} \code{cond} holds.

A cyclic proof is  a proof tree \form{\utree{i}} which is
 a tuple \form{(\sc{V}, \sc{E}, \backfun)} where
\begin{itemize}
    \item
 $V$ is a finite set of nodes representing entailments derived during the proof search;
    \item A directed edges \form{(\et, \code{PR}, \et') \in \sc{E}} (where \form{\et'} is a child of \form{\et}) means
 that the premise \form{\et'} is derived from the conclusion
 \form{\et} via 
 inference rule  \code{PR}. For instance,
suppose that the
  rule \code{PR_0} above has been applied, 
then
the following \form{n} edges are generated:
 \form{(\et, \code{PR_0}, \et_1)}, .., \form{(\et, \code{PR_0}, \et_n)};
    \item and $\backfun$ is a partial relation which captures back-links in
 the proof tree. If
  \form{\ctree{\et_c}{\et_b}{\sub}} holds, then
 \form{\et_b} is linked back to its ancestor \form{\et_c} through the substitution \form{\sub}
(where \form{\et_b} is referred as a {\em bud}
and \form{\et_c} is referred as a {\em companion}).
In particular, \form{\et_c} is of the form: \form{\D ~\ent~ \D'}
and \form{\et_b} is of the form: \form{\D_1 {\wedge}\pure ~\ent~ \D_1'}
where \form{\D \equiv \D_1\sub} and \form{\D' \equiv \D'_1\sub}.
\end{itemize}
A leaf node 
is marked as
  closed if it is evaluated as
 valid (i.e. the node is applied with an axiom) or 
 invalid (i.e. no rule can apply), or it is linked back. Otherwise, it is marked as open.
A proof tree is {\em invalid} if it contains at least one invalid leaf node.
It is a {\em pre-proof} if all its leaf nodes are either valid
or linked back.
A pre-proof is a cyclic proof if
a  global soundness
condition is imposed in the tree. 
 Intuitively,
this soundness condition requires that for every \form{\ctree{\enode_c}{\enode_b}{\sub}},
there exist inductive predicates \form{\seppred{\code{P}}{\setvars{t_1}}} in \form{\enode_c} and \form{\seppred{\code{Q}}{\setvars{t_2}}} in \form{\enode_b} such that \form{\seppred{\code{Q}}{\setvars{t_2}}}
is a subterm of  \form{\seppred{\code{P}}{\setvars{t_1}}}.

\begin{defn}[Trace]
\noindent Let \form{\utree{i}} be a pre-proof of \form{\D_{a} ~{\ent}~ \D_{c}} and $({\D_{a_i} ~{\ent}~ \D_{c_i} })_{i{\geq}0}$ be a path of $\utree{i}$.
A trace following
$({\D_{a_i} {\ent} \D_{c_i} })_{i{\geq}0}$ is a sequence $(\alpha_i)_{i{\geq}0}$ such that each $\alpha_i$  (for all
$i{\geq}0$) is a subformula of \form{{\D_{a_i}}} containing predicate $\code{P}(\setvars{t})^u$, and either:
\begin{itemize}
\item $\alpha_{i{+}1}$ is the subformula occurrence in $\D_{a_{i+1}}$ corresponding to $\alpha_{i}$ in $\D_{a_{i}}$.
\item or ${\D_{a_i} ~{\ent}~ \D_{c_i} }$ is the conclusion
of a left-unfolding rule, $\alpha_{i} \equiv \code{P}(\setvars{t})^u$ is unfolded,
and 
$\alpha_{i+1}$ is a subformula in $\D_{a_{i+1}}$
and is the 
 definition rule of $\code{P}(\setvars{x})^{u}[\setvars{t}/\setvars{x}]$.
 In this case,
 \form{i}
is said to be a {\em progressing point} of the trace.
\end{itemize}
\end{defn}

\begin{defn}[Cyclic proof]\label{cyclic}
A pre-proof \form{\utree{i}} of \form{\D_{a} ~{\ent}~ \D_{c}} is a cyclic proof if, for every infinite path $(\form{\D_{a_i} {\ent} \D_{c_i}})_{i{\geq}0}$ of $\utree{i}$,
there is a tail of the path $p{=}(\form{\D_{a_i} ~{\ent}~ \D_{c_i}})_{i{\geq}n}$
 such that there is a trace following $p$ which has
infinitely progressing points.
\end{defn}

Suppose that all proof rules are (locally) sound
(i.e., if the premises are valid then the conclusion is valid), the following Theorem shows the {\em global soundness}.
\begin{theorem}[Soundness \cite{Brotherston:CADE:11}]
If there is a cyclic proof of \form{\D_a ~\ent~ \D_c}, then \form{\D_a \models \D_c}.
\end{theorem}
The proof is by contraction and can be found in \cite{Brotherston:CADE:11}. Intuitively, if we can derive a cyclic proof for \form{\D_a ~\ent~ \D_c} and \form{\D_a \not\models \D_c}, then the inductive predicate at the progress points can be unfolded infinite often.
This contradicts with the least semantics of the predicate.

\section{Cyclic Entailment Procedure}
\label{sec.entail}
In this section, we present our main proposal, the entailment procedure
{\allEnt} with the proposed inference rules
(subsection \ref{ent.search.deci}),
and an illustrative example in subsection \ref{ent.example}.

\subsection{Proof Search}\label{ent.search.deci}
\begin{wrapfigure}{r}{0.6\textwidth} \savespace \savespace 
  \begin{minipage}{0.50\textwidth}\savespace \savespace 
\[
\begin{array}{l}
\hline
 {\allEnt} \\
\hline
 \code{\bf input}{:}~{\form{\enode_0}}
 \qquad \qquad
 \code{\bf output}{:}~{\code{\valid} ~or~ \code{invalid}} \\
\begin{array}{rl}
  1{:}& \form{i {\leftarrow} 0};
   \form{\utree{i} {\leftarrow} \enode_0}; \\
 2{:} & \code{\bf while} ~{ \form{\true} }~ \code{\bf do} \\
 3{:} & \quad   \form{(\code{res},\enode_i, PR_i) {\leftarrow} \isclose{\utree{i}}} ; \\
 4{:} &  \quad \code{\bf if~}{\code{res}{=}\valid}~ \code{\bf then} ~ \code{\bf return}~\valid;\\
 5{:}&  \quad \code{\bf if~}{\code{res}{=}\invalid} ~\code{\bf then} \code{\bf ~ return}~\code{\invalid};\\
 6{:} &\quad  \code{\bf if~}\code{\lbent}(\utree{i}, \enode_i) = \false~ \code{\bf then} \\
 7{:}&\qquad  \form{\utree{i+1} \leftarrow \code{apply}(\utree{i}, \enode_i, PR_i)} ; \\
8{:}       & \quad \form{i {\leftarrow} i{+}1};\\
9{:}& \code{\bf end}
\end{array} \\
\hline
\end{array}
\] \savespace
\caption{Proof tree construction procedure}\label{algo.proof.search}
\end{minipage} \savespace \savespace \savespace
\end{wrapfigure}

The proof search algorithm
{\allEnt} is presented in Fig. \ref{algo.proof.search}.
{\allEnt} takes \form{\enode_0} as input, produces
cyclic proofs and based on that decides
whether the input
    is
\form{\valid} or \form{\invalid}.
Initially, for every \form{\seppred{\code{P}}{\setvars{v}}^k \in \enode_0},
\form{k} is reset to \form{0} and
\utree{0} only has \form{\enode_0}
as an open leaf, the root.
The overall
 idea of
{\allEnt} is to iteratively reduce \form{\utree{0}} into a sequence of cyclic proof trees
\utree{i}, $i\geq 0$. 
On line 3, through procedure \isclose{\utree{i}},  {\allEnt} chooses an {\em open} leaf node $\enode_i$ and a proof rule
$PR_i$
to apply.
If \isclose{\utree{i}} returns {\valid} (that is, every leaf is applied to an axiom rule or involved in a backlink),
{\allEnt} returns \form{\valid} on line 4.
If it returns {\invalid},  then {\allEnt} returns {\invalid} (one line 5). Otherwise, it tries to link ${\enode_i}$
back to an internal node (on line 6). If this attempt
fails, it applies the rule (line 7).

Note that for each leaf, {\code{is\_closed}} attempts rules in the following order: normalization rules,
 axiom rules,
 and reduction rules.
A rule $PR_i$ is chosen if its conclusion can be unified with the leaf, through some substitution $\sub$.
 Then,  on line 7, for each premise of $PR_i$, procedure \code{apply} creates a new open node and connects the node to $\enode_i$ via a new edge.
If $PR_i$ is an axiom, 
procedure \code{apply} marks $\enode_i$ as closed and returns.

\paragraph{Procedure \isclose{\utree{i}}} This procedure examines
the following three cases.
\begin{enumerate}
\item First, if all leaf nodes are marked closed and none of them is {\invalid}
then
\code{is\_closed} returns  \form{\valid}. 
\item Secondly, \code{is\_closed} returns \form{\invalid} if there exists an open leaf node
\form{\enode_i:~\D~\ent~\D'} in NF such that one of the four following conditions holds:
\begin{enumerate}
\item \form{\enode_i} could not be applied by
any inference rule. 
\item there exists a predicate \form{op_1(E) \in \D} such that
\form{op_2(E) \notin \D'} and
one of the following conditions holds:
\begin{itemize}
\item either \form{ \seppred{\code{P}}{E'{,} E{,}...}}
or \form{\sepnode{E'}{c}{E{,}..}} are in both sides
\item both \form{\seppred{\code{P}}{E'{,} E{,}...} \not\in \D}
and \form{\sepnode{E'}{c}{E{,}..} \not\in \D}
\end{itemize}
\item there exists a predicate \form{op_1(E) {\in} \D'}  such that \form{G(op_1(E)) {\in} \D} and
\form{op_2(E) {\notin} \D}.
\item there exist \form{\sepnodeF{x}{c_1}{\setvars{v}_1} \in \D}, \form{\sepnodeF{x}{c_2}{\setvars{v}_2} \in \D'}
such that \form{c_1 \not\equiv c_2} or \form{\setvars{v}_1 {\not}{\equiv} \setvars{v}_2}.
\end{enumerate}
\item Lastly, there exists an open leaf node \form{\enode_i} that could be applied
by an inference rule (e.g. $PR_i$), \code{is\_closed} returns the triple (\form{\unknown}, \form{\enode_i}, $PR_i$). 
\end{enumerate}


In the rest, we discuss the proof rules 
and the auxiliary procedures in detail.

\hide{
In particular, 
  \code{\evalbase} 
is used to discharge
a solvable leaf node.
An entailment is solvable (and called a base entailment) if
both its LHS
and RHS
are {\em base} formulae.
If it is evaluated as {\valid}, it is marked as closed.
If the base formula contains unknown predicates,
\code{\evalbase} may generate relational assumptions
to constrain the pure properties in the LHS of the input entailment.
If the leaf node is not a base formula,
it uses procedure \code{abs}  check whether
the LHS of the node is unsatisfiable. If this is the case, it returns
{\valid}.
Next, {\allEnt} attempts to establish back-links over the remaining leaf
 nodes and marks the nodes as closed accordingly.
After that, at line 5, {\allEnt}
checks whether all leaf nodes are marked {\em closed}.
If it is the case, it
 returns \form{\valid} with the set of
relational assumptions.
If the tree contains a {\invalid} node, it stores the current tree
as a witness of invalidity (lines 9 and 11) prior to loading and examining
 another tree on the waiting list \form{\forest_{w}} (line 10).
Otherwise, it chooses an {\em open} leaf node to expand the proof
using inference rules.
When {\allEnt} applies an inference rule
to a leaf node \form{\enode} and reduces
this node into \form{n} new nodes \form{\enode_1},...,\form{\enode_n},
there are two scenarios.
If the rule is right unfolding \code{RU} (i.e.,
unfolding an
inductive predicate occurrence in RHS of a node \form{\enode}),
then it generates
up to \form{n} trees based on
the current tree  where each tree is augmented with one new node
\form{\enode_j} (\form{j\in \{1...n\}})
and one new edge \form{(\enode, RU_j, \enode_j)}.
If the rule applied is not right unfolding,
  it augments the tree with \form{n} new nodes
and \form{n} new edges \form{(\enode, R_1, \enode_1)},...,
\form{(\enode, R_n, \enode_n)}.
}

 \paragraph{Normalization} 
 \begin{figure}[tb]
\begin{center}
\[
\AxiomC{$\begin{array}{l}
\entailNCyc{\D\subst{E}{x}}{\D'\subst{E}{x}}{\heap\subst{E}{x}}
\end{array}$}
\LeftLabel{\rulename{Subst}}
\UnaryInfC{$\begin{array}{l}
\entailNCyc{\D{\wedge}x{=}E}{\D'}{\heap}
\end{array}$}
\DisplayProof
\quad
\AxiomC{$\begin{array}{l}
    \entailNCyc{\D{\wedge}E_1{=}E_2}{\D'}{\heap}\\
\entailNCyc{\D{\wedge}E_1{\neq}E_2}{\D'}{\heap}
  \end{array}$}
\RightLabel{\scriptsize $\begin{array}{c}
    E_1{=}E_2,E_1{\neq}E_2 {\not\in} \pure ~\& \\
    \FV({E_1,E_2}) \subseteq 
     (\FV(\D) {\cup} \FV(\D'))^S
  \end{array}$}
\LeftLabel{\scriptsize ExM}
\UnaryInfC{$\begin{array}{l}
\entailNCyc{\D}{\D'}{\heap}
\end{array}$}
\DisplayProof
\]
\[
\AxiomC{\entailNCyc{\D}{\D'}{\heap}}
\LeftLabel{\scriptsize =L}
\UnaryInfC{$\begin{array}{l}
\entailNCyc{\D{\wedge}E{=}E}{\D'}{\heap}
\end{array}$}
\DisplayProof
\qquad
\AxiomC{$\begin{array}{l}
\entailNCyc{(\heap{\wedge}\pure)\subst{tg}{sc}}{\D'\subst{tg}{sc}}{\heap\subst{tg}{sc}}
  \end{array}$}
\LeftLabel{\scriptsize LBase}
\UnaryInfC{$\begin{array}{l}
\entailNCyc{ \seppredF{\code{P}}{E{,}E{,}\setvars{B}{,}u{,}sc{,}tg}{\sep} \heap{\wedge}\pure}{\D'}{\heap}
\end{array}$}
\DisplayProof
\]
\[
\AxiomC{$\begin{array}{l}
\entailNCyc{ op(E){\sep}\heap{\wedge}\pure{\wedge}G(op(E)){\wedge}E{\neq}\nil}{\D'}{\heap}
  \end{array}$}
\RightLabel{\scriptsize $E{\neq}\nil {\notin} \pure$}
\LeftLabel{\scriptsize ${\neq}\nil$}
\UnaryInfC{$\begin{array}{l}
\entailNCyc{ op(E){\sep}\heap{\wedge}\pure{\wedge}G(op(E))}{\D'}{\heap}
\end{array}$}
\DisplayProof
\]
\[
\AxiomC{$\begin{array}{l}
\entailNCyc{op_1(E_1){\sep}op_2(E_2){\sep}\heap{\wedge}\pure{\wedge}E_1{\neq}E_2}{\D'}{\heap}
  \end{array}$}
\RightLabel{\scriptsize $ E_1{\neq}E_2 {\not\in}\pure
  \text{ and } \form{G(op_1(E_1))}, \form{G(op_2(E_2)) \in \pure}$}
\LeftLabel{\scriptsize ${\neq}\sep$}
\UnaryInfC{$\begin{array}{l}
\entailNCyc{op_1(E_1){\sep}op_2(E_2){\sep}\heap{\wedge}\pure}{\D'}{\heap}
\end{array}$}
\DisplayProof
\]
\caption{Normalization rules}
\label{fig.ent.comp.base}
\end{center} \savespace \savespace  \savespace
\end{figure}

An entailment is in the normal form (NF) if
 its LHS is in NF.
We write
\form{op(E)} to denote for either \form{\sepnodeF{E}{c}{\setvars{v}}} or
\form{\seppred{\code{P}}{E{,} F{,} \setvars{B}{,}\setvars{v}}}.
 Furthermore,
the guard \form{G(op(E))} is defined by:
\form{G(\sepnodeF{E}{c}{\setvars{v}}) \defsym \true} 
and
\form{G(\seppred{\code{P}}{E{,} F{,} \setvars{B}{,}\setvars{v}}) \defsym E{\neq}F}. 

\begin{defn}[Normal Form]\label{defn.nf} A formula \form{\heap{\wedge} \phi {\wedge} \a}
is in normal form if:
\[
\begin{array}{clcl}
 1.& \form{op(E) \in \heap} \text{ implies } \form{G(op(E)) \in \phi} & \qquad
 4.& \form{E_1{=}E_2 \not\in \phi} \\
2.& \form{op(E) \in \heap} \text{ implies } \form{E{\neq}\nil \in \phi} & \qquad
5.& \form{E{\neq}E  \not\in \phi} \\
3.&  \form{op_1(E_1)\sep op_2(E_2) \in \heap} \text{ implies }
\form{E_1{\neq}E_2\in \phi} & \qquad
6.& \form{\a} \text{ is satisfiable}
\end{array}
 \]
\end{defn}
If \form{\D} is in NF and for any \form{\sstack, \sheaps \models \D}, then
 \form{\dom(\sheaps)} is uniquely defined by \form{\sstack}.

The normalisation rules are presented in Fig. \ref{fig.ent.comp.base}.
{Basically, {\allEnt} applies these rules to
a leaf exhaustively and transforms it into NF before other.}
Given an inductive predicate \form{\seppredF{\code{P}}{E,F,...}},
rule \rulename{ExM} excludes the middle
by doing case analysis for the predicate
between base-case (i.e., \form{E{=}F}) and recursive-case 
(i.e., \form{E{\neq}F}).
The normalization rule \form{{\neq}{\nil}} follows the following facts:
$\sepnode{E}{c}{\anon}  \imply E{\neq}\nil $
and $\seppred{\code{P}}{E{,} F{,} \anon} {\wedge}E{\neq}F  \imply  E{\neq}\nil$.
 Similarly, rule \form{{\neq}{\sep}} follows the following facts:
$x{\pto}\anon {\sep} \seppred{\code{P}}{y{,} F{,} \anon} {\wedge}y{\neq}F  \imply  x{\neq}y$, 
$x{\pto}\anon {\sep} y{\pto}\anon  \imply x{\neq}y$, and
$\seppred{\code{P_i}}{x{,} F_1{,} \anon} {\sep} \seppred{\code{P_j}}{y{,} F_2{,} \anon} {\wedge}x{\neq}F_1 {\wedge}y{\neq}F_2  \imply x{\neq}y$.

\paragraph{Axiom and Reduction} \label{ent.comp.base}

\begin{figure}[tb]
\begin{center}
    \[
\AxiomC{}
\LeftLabel{\scriptsize Id}
\UnaryInfC{$\begin{array}{l}
 \D \wedge \pure ~\ent~ {\D }
\end{array}$}
\DisplayProof
\quad
 \AxiomC{}
 \LeftLabel{\scriptsize Emp}
\UnaryInfC{$\begin{array}{l}
\entailNCyc{\emp{\wedge}\pure}{\emp{\wedge}\true}{\heap}
\end{array}$}
\DisplayProof
\quad
\AxiomC{}
\LeftLabel{\scriptsize Inconsistency}
\RightLabel{$
\begin{array}{l}
\pure ~{\models}~ \false \\
\end{array} $}
\UnaryInfC{$\begin{array}{l}
 \heap {\wedge} {\pure}  ~\ent~ {\D }
\end{array}$}
\DisplayProof
\]
\[
\AxiomC{\entailNCyc{\D}{\D'}{\heap}}
\LeftLabel{\scriptsize =R}
\UnaryInfC{$\begin{array}{l}
\entailNCyc{\D}{\D'{\wedge}E{=}E}{\heap}
\end{array}$}
\DisplayProof
\quad
\AxiomC{\entailNCyc{\D{\wedge}\pure}{\D'}{\heap}}
\LeftLabel{\scriptsize Hypothesis}
\RightLabel{\scriptsize $\pure \models \pure' $}
\UnaryInfC{$\begin{array}{l}
\entailNCyc{ \D{\wedge}\pure}{\D'{\wedge}\pure'}{\heap}
\end{array}$}
\DisplayProof
\quad
\AxiomC{$\begin{array}{l}
\entailNCyc{\D}{\D'\wedge tg{=}sc}{\heap}
  \end{array}$}
\LeftLabel{\scriptsize RBase}
\UnaryInfC{$\begin{array}{l}
\entailNCyc{\D}{ \seppredF{\code{P}}{E{,}E{,}\setvars{B}{,}u{,}sc{,}tg}{\sep} \D'}{\heap}
\end{array}$}
\DisplayProof
\]
\[
\AxiomC{
$\begin{array}{c}
    \entailNCyc{\heap_1{\wedge}{\pure}}{\heap_2} {\heap{\sep}\heap'} \quad
    \entailNCyc{\heap{\wedge}\pure}{\heap'{\land} \pure'} {\heap{\sep}\heap'}
\end{array}$
}
\LeftLabel{\scriptsize ${\sep}$}
\RightLabel{$\begin{array}{c}
\code{roots}(\heap_1)\cap\code{roots}(\heap)=\emptyset ~\&~
\FV(\heap_2) {\subseteq} \FV(\heap_1{\wedge}\pure){\cup}\{\nil\}\\
\&~\FV(\heap') {\subseteq} \FV(\heap{\wedge}\pure){\cup}\{\nil\}
\end{array}$}
\UnaryInfC{
$\entailNCyc{\heap_1 {\sep} \heap{\wedge}\pure}{\heap_2 {\sep} \heap' \land \pure'}{\heap}$
}
\DisplayProof
\]
\[
\AxiomC{
$
\begin{array}{l}
  {\seppredF{\code{Q_1}}{E_1{,}B}^{\textcolor{blue}{0}}{\sep}\seppredF{\code{Q_2}}{E_2{,}X}^{\textcolor{blue}{0}}{\sep}\seppredF{\code{P}}{X{,}F{,}\setvars{B}{,}u{,}sc'{,}tg}^{\textcolor{blue}{k}}   {\sep}\D_1{\wedge}} x{\neq}F_3{\wedge}\pure_0 \\ \quad   ~\ent~{ \seppredF{\code{Q}}{x{,}F_3{,}\setvars{B}{,}u{,}sc{,}tg_2} {\sep} \heap_2{\wedge}\pure_2}  \\
\end{array}
$
}
\RightLabel{\scriptsize  $\sepnodeF{x}{c}{\_}{\not\in}\heap_2$
}
\LeftLabel{\scriptsize Frame}
\UnaryInfC{$\begin{array}{l}
    \entailNCyc{\seppredF{\code{P}}{x{,}F{,}\setvars{B}{,}u{,}sc{,}tg}^{\textcolor{blue}{k}}{\sep}\D_1{\wedge}x{\neq}F_3
      }{\sepnodeF{x}{c}{X{,}E_1{,}E_2{,}u{,}sc'} {\sep} \heap_2{\wedge}\pure_2}{\heap}
\end{array}$}
\DisplayProof
\]
\[
\AxiomC{$\begin{array}{l}
    { \sepnodeF{x}{c}{X{,}E_1{,}E_2{,}u{,}sc'} {\sep}\heap_1{\wedge}\pure_1{\wedge}x{\neq}F} \\
 \quad \ent~ { \sepnodeF{x}{c}{X{,}E_1{,}E_2{,}u{,}sc'}{\sep}\seppredF{\code{Q_1}}{E_1{,}B}{\sep}\seppredF{\code{Q_2}}{E_2{,}X}{\sep}\seppredF{\code{P}}{X{,}F{,}\setvars{B}{,}u{,}sc'{,}tg}{\sep}\heap_2 {\wedge}\pure_2{\wedge}\pure_0}
  \end{array}$}
\RightLabel{\scriptsize $\begin{array}{c} \dagger
 \end{array}$}
\LeftLabel{\scriptsize RInd}
\UnaryInfC{$\begin{array}{l}
\entailNCyc{\sepnodeF{x}{c}{X{,}E_1{,}E_2{,}u{,}sc'} {\sep}\heap_1{\wedge}\pure_1{\wedge}x{\neq}F}{ \seppredF{\code{P}}{x{,}F{,}\setvars{B}{,}u{,}sc{,}tg}{\sep} \heap_2{\wedge}\pure_2}{\heap}
\end{array}$}
\DisplayProof
\]
\[
\AxiomC{
$
\begin{array}{l}
  { \sepnodeF{x}{c}{X{,}E_1{,}E_2{,}u{,}sc'}{\sep}\seppredF{\code{Q_1}}{E_1{,}B}^{\textcolor{blue}{0}}{\sep}\seppredF{\code{Q_2}}{E_2{,}X}^{\textcolor{blue}{0}}{\sep}\seppredF{\code{P}}{X{,}F{,}\setvars{B}{,}u{,}sc'{,}tg}^{\textcolor{blue}{k{+}1}}   {\sep}\D_1{\wedge}} x{\neq}F_3{\wedge}\pure_0 \\ \quad  ~\ent~{ \seppredF{\code{Q}}{x{,}F_3{,}\setvars{B}{,}u{,}sc{,}tg_2} {\sep} \heap_2{\wedge}\pure_2}  \\
\end{array}
$
}
\RightLabel{\scriptsize $\sharp$ 
}
\LeftLabel{\scriptsize LInd}
\UnaryInfC{$\begin{array}{l}
    \entailNCyc{\seppredF{\code{P}}{x{,}F{,}\setvars{B}{,}u{,}sc{,}tg}^{\textcolor{blue}{k}}{\sep}\D_1{\wedge}x{\neq}F_3
      }{\seppredF{\code{Q}}{x{,}F_3{,}\setvars{B}{,}u{,}sc{,}tg_2} {\sep} \heap_2{\wedge}\pure_2}{\heap}
\end{array}$}
\DisplayProof
\]
\caption{Reduction rules
 (where $\sharp{:}~\seppredF{\code{P}}{x{,}F{,}\setvars{B}{,}u{,}sc{,}tg}{\not\in}\heap_2$, $\dagger{:}~ \sepnodeF{x}{c}{X{,}E_1{,}E_2{,}u{,}sc'}  {\not\in} \heap_2$)
}
\label{fig.ent.comp.red}
\end{center}\savespace \savespace 
\end{figure}

Axiom rules include \rulename{Emp},
 \rulename{Inconsistency} and \rulename{Id}
 presented in Fig. \ref{fig.ent.comp.red}.
If each of these rules is applied into a leaf node, the node is
evaluated as
\form{\valid} and marked as closed.
\label{ent.comp.search}
The remaining ones in Fig. \ref{fig.ent.comp.red} are reduction rules. 

To simplify the presentation, the unfoldings
in rules \rulename{Frame}, \rulename{RInd}, and \rulename{LInd}
are applied with the following definition of inductive predicates:
\saveone\[\saveone
\begin{array}{l}
  \seppredF{\code{P}}{x{,}F{,}\setvars{B}{,}u{,}sc{,}tg} \equiv \emp {\wedge}x{=} F {\wedge}sc{=}tg \\
\quad \vee~ \exists X{,} sc'{,}d_1{,} d_2 . \sepnodeF{x}{c}{X{,}d_1{,}d_2{,}u{,}sc}{\sep}\seppredF{\code{Q_1}}{d_1{,}B}{\sep}\seppredF{\code{Q_2}}{d_2{,}X}{\sep}\seppredF{\code{P}}{X{,}F{,}\setvars{B}{,}u{,}sc'{,}tg}{\wedge}\pure_0
\end{array}
\]
where \form{B {\in} \setvars{B}}, the matrix \form{\heap'} contains two nested predicates
\form{Q_1} and \form{Q_2}, and the heap cell \form{c \in \Dns}
is defined as
\form{\code{data}~ c \{c ~ next;c_{1}~down_1;c_{2}~down_2;\tau_s~scdata;\tau_u ~udata\}}
where \form{c_{1},c_{2} {\in} \Dns},
\form{down_1}, \form{down_2} fields are for the nested
predicates in the matrix heaps,
\form{udata} field is for the transitivity data, and \form{scdata} field are for 
ordering data.
The formalism of these rules for general form of the matrix heaps \form{\heap'} is presented in
\repconf{App. \ref{app.ent.red}.}{\cite{Loc:Lin:TR:2019}.}

 \rulename{{=}R} and \rulename{Hypothesis}
eliminate  pure constraints in the RHS.
In rule \rulename{{\sep}}, 
\form{\code{roots}(\heap)} is defined inductively as:
\form{\code{roots}(\emp){\equiv}\{\}}, \form{\code{roots}(r\pto\anon){\equiv}\{r\}}, \form{\code{roots}(P(r,F,..)){\equiv}\{r\}}
and \form{\code{roots}(\heap_1{\sep}\heap_2)\equiv\code{roots}(\heap_1)\cup\code{roots}(\heap_2)}.
This rule 
is applied in three ways. First, it is applied into an entailment which is of the form
\form{\heap {\wedge} \pure ~\ent~ \heap {\wedge} \pure'}. It
matches and discards
the identified heap predicates between the two sides so as to generate
a premise with empty heaps. As a result, this premise may be applied with the axiom
rule \rulename{EMP}.
Secondly, it is applied into an entailment whose LHS is a base formula
e.g., \form{\sepnodeF{x_1}{c_1}{\setvars{v}_1}{\sep}...{\sep}\sepnodeF{x_n}{c_n}{\setvars{v}_n} {\wedge} \pure~\ent~ \heap' {\wedge} \pure'}.
For each points-to predicate \form{\sepnodeF{x_i}{c_i}{\setvars{v}_i} {\in} \heap'},
{\allEnt} searches for one points-to predicate \form{\sepnodeF{x_j}{c_j}{\setvars{v}_j}}
in the LHS such that \form{\sepnodeF{x_j}{c_j}{\setvars{v}_j} \equiv \sepnodeF{x_i}{c_i}{\setvars{v}_i}}.
Likewise, for each occurrence of inductive predicates \form{\seppred{\code{P}}{r{,} F{,} \setvars{B}{,}u{,}sc{,}tg}}
 in the RHS, {\allEnt} searches for a points-to predicate \form{r\pto\anon}
in the LHS 
such that rule \rulename{RInd}
could be applied.
If any of these searches fails, {\allEnt} decides the conclusion as {\invalid}.
Lastly, it is applied into an entailment that is of the form \form{\D_1 \sep \D ~\ent~ \D_2 \sep \D'} where 
either \form{\D_1 ~\ent~\D_2} or \form{\D ~\ent~ \D'} could be linked back into an internal node.

Rule \rulename{LInd} unfolds the inductive predicates
in the LHS.
We notice that every LHS of entailments in
this rule also captures the unfolding numbers
for subterm relationship and generates the progressing point in the cyclic proofs afterward.
These numbers are essential for our system to construct cyclic proofs.
This rule is applied in a {\em depth-first}
 manner i.e., if there are more than one
occurrences of inductive predicates in the LHS that could be applied by this rule,
the one with the greatest unfolding number is chosen.
We emphasize that the last five rules
still work well when the predicate in the RHS
 contains only
 a subset of the local properties wrt. the predicate in the LHS.

\paragraph{Back-Link Generation}\label{ent.comp.cyclic}

Procedure {\lbent} generates a back-link as follows.
In a pre-proof, given a path containing a back-link, say \form{\enode_1,\enode_2,..,\enode_m} where
\form{\enode_1} is a companion and \form{\enode_m} a bud, then
\form{\enode_1} is in NF and of the following form:

\begin{itemize}
\item \form{\enode_1{\equiv}{\seppredF{\code{P}}{x{,}F{,}\setvars{B}{,}u{,}sc{,}tg}^{k}{\sep}\heap{\wedge}\pure{\wedge}x{\neq}F{\wedge}x{\neq}\nil} 
~ \ent~{\seppredF{\code{Q}}{x{,}F_2{,}\setvars{B}{,}u{,}sc{,}tg_2}{\sep}\heap'{\wedge}\pure'}}.
\item \form{\enode_2} is obtained from applying \rulename{LInd} into \form{\enode_1}. \form{\enode_2} is of the form:
\[
\begin{array}{l}
\form{{\sepnodeF{x}{c}{X,\setvars{p},{,}u{,}sc} {\sep}\heap' {\sep} \seppredF{\code{P}}{X{,}F{,}\setvars{B}{,}u{,}sc'{,}tg}^{k{+}1}{\sep}\heap{\wedge}\pure{\wedge}x{\neq}F{\wedge}x{\neq}\nil{\wedge}\pure_1} \\
\quad \ent~{\seppredF{\code{Q}}{x{,}F_2{,}\setvars{B}{,}u{,}sc{,}tg_2}{\sep}\heap'{\wedge}\pure'}}
\end{array}
\]
We remark that
\form{sc \diamond sc' \in \pure_1}
 and if \form{k \geq 1} then \form{sc_i \diamond sc \in \pure}
\item \form{\enode_3}, .., \form{\enode_{m{-}4}} are obtained from applications of normalization rules
in order to normalize the LHS of \form{\enode_2} due to the presence of \form{\heap'}.
We note that as the roots of inductive predicates in \form{\heap'} are fresh variables,
the applications of the normalization rules above do not affect the RHS of \form{\enode_2}.
That means RHS of \form{\enode_3}, .., \form{\enode_{m{-}4}} are the same with the RHS of \form{\enode_2}.
As a result,
\form{\enode_{m{-}4}} is of the form:
\[
\begin{array}{l}
\form{{\sepnodeF{x}{c}{X,\setvars{p},{,}u{,}sc} {\sep}\heap_1'' {\sep} \seppredF{\code{P}}{X{,}F{,}\setvars{B}{,}u{,}sc'{,}tg}^{k{+}1}{\sep}\heap{\wedge}\pure{\wedge}x{\neq}F{\wedge}x{\neq}\nil{\wedge}\pure_1{\wedge}\pure_2} \\
\quad \ent~{\seppredF{\code{Q}}{x{,}F_2{,}\setvars{B}{,}u{,}sc{,}tg_2}{\sep}\heap'{\wedge}\pure'}}
\end{array}
\]
where \form{\heap_1''} may be \form{\emp} and \form{\pure_2} is a conjunction of disequalities coming from \rulename{ExM}.
\item \form{\enode_{m{-}3}} is obtained from  application of \rulename{ExM} over \form{x} and \form{F_2}
and of the form:
\[
\begin{array}{l}
\form{{\sepnodeF{x}{c}{X,\setvars{p},{,}u{,}sc} {\sep}\heap_1'' {\sep} \seppredF{\code{P}}{X{,}F{,}\setvars{B}{,}u{,}sc'{,}tg}^{k{+}1}{\sep}\heap{\wedge}\pure{\wedge}x{\neq}F{\wedge}x{\neq}\nil{\wedge}\pure_1{\wedge}\pure_2} \\
\quad {\wedge}x{\neq}F_2~\ent~{\seppredF{\code{Q}}{x{,}F_2{,}\setvars{B}{,}u{,}sc{,}tg_2}{\sep}\heap'{\wedge}\pure'}}
\end{array}
\]
(For the case \form{x{=}F_2}, the rule \rulename{ExM} is kept applying until either \form{F\equiv F_2},
 that is two sides are reaching the end of the same heap segment, or it is stuck.)
\item  \form{\enode_{m{-}2}} is obtained from  application of \rulename{RInd} and is of the form:
\saveone\[\saveone
\begin{array}{l}
\form{{\sepnodeF{x}{c}{X,\setvars{p},{,}u{,}sc} {\sep}\heap_1'' {\sep} \seppredF{\code{P}}{X{,}F{,}\setvars{B}{,}u{,}sc'{,}tg}^{k{+}1}{\sep}\heap{\wedge}\pure{\wedge}x{\neq}F{\wedge}x{\neq}\nil{\wedge}\pure_1{\wedge}\pure_2} \\
\quad {\wedge}x{\neq}F_2~\ent~{\sepnodeF{x}{c}{X{,}\setvars{p}{,}u{,}sc}{\sep}\heap_2'' {\sep}\seppredF{\code{Q}}{X{,}F_2{,}\setvars{B}{,}u{,}sc'{,}tg_2}{\sep}\heap'{\wedge}\pure'}{\wedge}\pure_2'}
\end{array}
\]
\item  \form{\enode_{m{-}1}} is obtained from  application of \rulename{Hypothesis} to eliminate \form{\pure_2'}
 (otherwise, it is stuck) and is of the form:
\saveone\[\saveone
\begin{array}{l}
\form{{\sepnodeF{x}{c}{X,\setvars{p},{,}u{,}sc} {\sep}\heap_1'' {\sep} \seppredF{\code{P}}{X{,}F{,}\setvars{B}{,}u{,}sc'{,}tg}^{k{+}1}{\sep}\heap{\wedge}\pure{\wedge}x{\neq}F{\wedge}x{\neq}\nil{\wedge}\pure_1{\wedge}\pure_2} \\
\quad {\wedge}x{\neq}F_2~\ent~{\sepnodeF{x}{c}{X{,}\setvars{p}{,}u{,}sc}{\sep}\heap_2'' {\sep}\seppredF{\code{Q}}{X{,}F_2{,}\setvars{B}{,}u{,}sc'{,}tg_2}{\sep}\heap'{\wedge}\pure'}}
\end{array}
\]
\item \form{\enode_{m}} is obtained from  application of \rulename{\sep} and is of the form:
\saveone\[\saveone
\begin{array}{l}
\form{ \seppredF{\code{P}}{X{,}F{,}\setvars{B}{,}u{,}sc'{,}tg}^{k{+}1}{\sep}\heap{\wedge}\pure{\wedge}x{\neq}F{\wedge}x{\neq}\nil{\wedge}\pure_1{\wedge}\pure_2 {\wedge}x{\neq}F_2} \\
\quad \ent~\form{\seppredF{\code{Q}}{X{,}F_2{,}\setvars{B}{,}u{,}sc'{,}tg_2}{\sep}\heap'{\wedge}\pure'}
\end{array}
\]
\end{itemize}
When \form{k \geq 1} it is always possible to
link \form{\enode_m} back to \form{\enode_1} through 
the substitution is \form{\sub{\equiv}[x/X,sc/sc']} after weakening some pure
constraints in its LHS.
\hide{A routine of rule applications to generate a pre-proof in our system typically is:
\form{(\rulename{EXM}, \rulename{LInd}, \rulename{{\neq}\nil}, \rulename{RInd}, (\rulename{HYPOTHESIS})^*,
 (\rulename{EXM})^*, (\rulename{{\neq}\sep})^*)^+,\sep}
where \form{R^*} (resp. \form{R^+}) means the sequence of rules \form{R} can be applied
 zero (resp. one) or more times.
Here, rule \form{\code{\scriptsize HYPOTHESIS}} is applied to eliminate
pure formulas capturing the local properties of the occurrence of inductive predicates in the RHS.
As a result, a pair of companion and bud in the system of spatial-only definitions are as follows.
\[
\begin{array}{rl}
\text{bud:} \qquad & \qquad
\form{ \seppred{\code{P_1}}{X_1{,} F_1{,} \setvars{B}_1}^{l_1} {\sep} ...{\sep}
\seppred{\code{P_m}}{X_m{,} F_m{,} \setvars{B}_m}^{l_m} {\sep} \D_a{\wedge}\pure_1 \\
& \qquad \qquad \ent~ {\seppred{\code{P_1}}{r_1{,} F_{2_1}{,} \setvars{B}_1} {\sep} ...{\sep}
\seppred{\code{P_m}}{r_m{,} F_{2_m}{,} \setvars{B}_m}  {\sep} \D_c}}\\
\text{companion:} & \qquad
\form{ \seppred{\code{P_1}}{x_1{,} F_1{,} \setvars{B}_1}^{k_1} {\sep} ...{\sep}
\seppred{\code{P_m}}{x_m{,} F_m{,} \setvars{B}_m}^{k_m} {\sep} \D_a \\
&  \qquad\qquad \ent~{\seppred{\code{P_1}}{t_1{,} F_{2_1}{,} \setvars{B}_1} {\sep} ...{\sep}
\seppred{\code{P_m}}{t_m{,} F_{2_m}{,} \setvars{B}_m} {\sep} \D_c}}
\end{array}
\]
where \form{\heap} is a conjunction of either points-to or inductive predicates,
\form{k_1,...,k_m {\geq}0},
\form{l_1{-}k_1}, ..., \form{l_m{-}k_m} are positive numbers. 
Furthermore, the substitution is \form{\sub{\equiv}[x_1/X_1,...,x_m/X_m]}.
For those inductive predicates with parameters specifying local properties, the substitution are extended with
the substitution of the source local
parameters (e.g., \form{[sc/sc']}).
}
\hide{In the following, we briefly describe a cyclic proof technique
enhancing the proposal in \cite{Brotherston:05}.

\begin{defn}[Pre-proof]
A {\em pre-proof} of entailment \form{\D_{a} {\ent_{\lemstore}} \D_{c}}
is a pair ($\utree{i}$, $\mathcal L$) where $\utree{i}$ is a derivation tree
and $\mathcal L$ is a back-link function
 such that:
the root of $\utree{i}$  is \form{\D_{a} {\ent_{\lemstore}} \D_{c}};
for every edge from \form{E_i} to \form{E_j} in  $\utree{i}$ , \form{E_i} is
 a conclusion of an
inference rule with a premise $E_{j}$.
There is a back-link between \form{E_c} and \form{E_l}
if there exists
   $\mathcal L$(\form{E_l}) = \form{E_c} (i.e., \form{E_{c}=E_{l}\,\theta} with some substitution $\theta$)
;
and for every leaf \form{E_l}, \form{E_l} is an axiom rule (without conclusion). 
\end{defn}
\noindent If $\mathcal L$(\form{E_l}) {=} \form{E_c}, \form{E_{l}}
 (resp. \form{E_{c}}) is referred as a bud (resp. companion).

\begin{defn}[Trace]
\noindent Let ($\utree{i}$, $\mathcal L$) be a pre-proof of \form{\D_{a} {\ent_{\lemstore}} \D_{c}}; $({\D_{a_i} {\ent_{\lemstore_i}} \D_{c_i} })_{i{\geq}0}$ be a path of $\utree{i}$.
A trace following
$({\D_{a_i} {\ent_{\lemstore_i}} \D_{c_i} })_{i{\geq}0}$ is a sequence $(\alpha_i)_{i{\geq}0}$ such that each $\alpha_i$  (for all
$i{\geq}0$) is an instance of the predicate $\code{P}(\setvars{t})$ in the formula ${\D_{a_i}}$, and either:
\begin{itemize}
\item $\alpha_{i{+}1}$ is the sub-formula containing an instance of $\code{P}(\setvars{t})$
  in $\D_{a_{i+1}}$;
\item or ${\D_{a_i} {\ent_{\lemstore_i}} \D_{c_i} }$ is the conclusion of an unfolding rule,
$\alpha_i$ is an instance predicate $\code{P}(\setvars{t})$ in $\D_{a_i}$ and
$\alpha_{i+1}$ is a sub-formula
\form{\D}[\setvars{t}/\setvars{v}] which is
a definition rule of the inductive predicate
$\code{P}(\setvars{v})$. i is a progressing point of the trace.
\end{itemize}
\end{defn}

To ensure that a pre-proof is sound,
a global
{\em soundness condition} must be imposed to guarantee well-foundedness.

\begin{defn}[Cyclic proof]\label{cyclic}
A pre-proof ($\utree{i}$, $\mathcal L$) of \form{\D_{a} {\ent_{\lemstore}} \D_{c}} is a cyclic proof if, for every infinite path $(\form{\D_{a_i} {\ent_{\lemstore_i}} \D_{c_i}})_{i{\geq}0}$ of $\utree{i}$,
there is a tail of the path $p{=}(\form{\D_{a_i} {\ent_{\lemstore_i}} \D_{c_i}})_{i{\geq}n}$
 such that there is a trace following $p$ which has
infinitely progressing points.
\end{defn}
}

\subsection{Illustrative Example}\label{ent.example}

We illustrate our system through the following example:
\[
\begin{array}{l}
 \enode_0{:}~ \seppredF{\code{lls}}{x{,}\nil{,}mi{,}ma}^0
                \wedge x{\neq}\nil
  ~ \ent ~
  \seppredF{\code{llb}}{x{,}\nil{,}mi}  \\
\end{array}
\]
where the sorted linked-list \code{lls} 
(\form{mi} is the minimum value and \form{ma} is the maximum value)
is defined in Sect. \ref{spec.deci.ent} 
  and  \code{llb} defines
 singly-linked lists whose values are greater than
 or equal to a constant number. 
Particularly, predicate \code{llb} is defined as follows.
\[
\begin{array}{l}
\seppred{\code{pred~llb}}{r{,}F{,}b} ~{\equiv}~ \emp {\wedge} r{=}F \\
 \quad \vee ~  \exists X_{tl}{,}d. \sepnode{r}{c_4}{X_{tl}{,}d} \sep \seppred{\code{llb}}{X_{tl}{,}F{,}b} {\wedge} r{\neq}F \wedge
 b {\leq} d\\
\end{array}
\]
Since the LHS is stronger than the RHS, this entailment is valid.
Our system could generate the cyclic proof (shown in Fig. \ref{mov.ctree}) to prove the validity of \form{\enode_0}.
In the following, we present step-by-step
to show how the proof was created.
Firstly, \form{\enode_{0}}, which is in NF, is applied with rule \rulename{LInd} to
unfold predicate \form{\seppredF{\code{lls}}{x{,}\nil{,}mi{,}ma}^0} and obtain
\form{\enode_{1}} as:
\[
\begin{array}{ll}
\enode_{1}{:}& \sepnodeF{x}{c_4}{X,m'} \sep
  \seppredF{\code{lls}}{X{,}\nil{,}m'{,ma}}^{\textcolor{red}{1}}
                \wedge x{\neq}\nil \wedge mi{\leq}m' 
 ~ \ent ~
  \seppredF{\code{llb}}{x{,}\nil{,}mi}
\end{array}
\]
We remark that the unfolding number of the recursive predicate \code{lls} in the LHS is increased by \form{1}.
Next, our system normalizes \form{\enode_{1}} by applying rule \rulename{ExM} into \form{X} and \form{\nil}
to generate two children \form{\enode_{2}} and \form{\enode_{3}} as follows.
\[
\begin{array}{ll}
  \enode_{2}{:}& \sepnodeF{x}{c_4}{X,m'} \sep
  \seppredF{\code{lls}}{X{,}\nil{,}m'{,}ma}^1
                \wedge x{\neq}\nil \wedge mi{\leq}m' 
\wedge X{=}\nil \\
& \quad \ent ~
  \seppredF{\code{llb}}{x{,}\nil{,}mi} \\
  \enode_{3}{:}& \sepnodeF{x}{c_4}{X,m'} \sep
  \seppredF{\code{lla}}{X{,}\nil{,}m'{,}ma}^1
                \wedge x{\neq}\nil \wedge mi{\leq}m' 
                \wedge X{\neq}\nil
 \\ &
  \quad \ent ~
  \seppredF{\code{llb}}{x{,}\nil{,}mi}
\end{array}
\]
For the left child, it applies normalization rules to obtain 
\begin{figure}[t!]
\begin{center}
\begin{tikzpicture}[node distance=18mm,level 1/.style={sibling distance=22mm},
      level 2/.style={sibling distance=7mm},
                        level distance=22pt, draw]
  \tikzstyle{every state}=[draw,text=black]

\node (E1)                    {\textcolor{blue}{$\enode_0$}};
   \node         (F1) [below =2mm  of E1] {\textcolor{black}{$\enode_{1}$}};
   \node         (B2) [below left=1mm and 4mm of F1] {$\enode_{2}$};
   \node         (C2) [below right=1mm and 4mm of F1] {\textcolor{black}{$\enode_{3}$}};
   \node         (D2) [below left=2mm and 1mm of B2] {\textcolor{black}{$\enode_{4}$}};
   \node         (E2) [below left=2mm and 2mm   of D2] {\textcolor{black}{$\enode_{5}$}};
   \node         (F2) [below left=2mm and 2mm  of E2] {\textcolor{black}{$\enode_{6}$}};
   \node         (G2) [below =2mm  of F2] {\textcolor{black}{$\enode_{7}$}};
   \node         (H2) [below =2mm  of G2] {\textcolor{black}{$\enode_{8}$}};
   \node         (G3) [below right=4mm and 1mm   of C2] {\textcolor{black}{$\enode_{9}$}};
    \node         (H3) [below right=4mm and 1mm   of G3] {\textcolor{black}{$\enode_{10}$}};
\node         (I1) [below left=1mm and 4mm of H3] {$\enode_{11}$};
   \node       (I2) [below right=1mm and 4mm of H3] {\textcolor{blue}{$\enode_{12}$}};

  \path     (E1) edge             node [left] {{\scriptsize LInd}} (F1)
            (F1) edge             node [pos=0.1,left] {{\scriptsize ExM}} (B2)
            edge              node[pos=0.1,right]  { {\scriptsize ExM}} (C2)
            (B2) edge             node [left] {{\scriptsize Subst}} (D2)
            (D2) edge             node [left] {{\scriptsize LBase}} (E2)
            (E2) edge             node [left] {{\scriptsize RInd}} (F2)
            (F2) edge             node [right] {{\scriptsize Hypothesis{$+$}RBase}} (G2)
            (G2) edge             node [left] {{\scriptsize $\sep$}} (H2)
            (C2) edge             node [left] {{\scriptsize ${\neq}\sep${+}RInd}} (G3)
            (G3) edge             node [left] {{\scriptsize Hypothesis}} (H3)
            (H3) edge             node [pos=0.1,left] {{\scriptsize \form{\sep}}} (I1)
            edge              node[pos=0.1,right]  { {\scriptsize \form{\sep}}} (I2)
            (I2)  edge [->,bend right=45,dotted]  node[right] {\scriptsize \form{[x/X,mi/m']}} (E1)
;
\end{tikzpicture}
\end{center}\savespace
\caption{Cyclic Proof of \form{\seppredF{\code{lls}}{x{,}\nil{,}mi{,ma}}^0
                {\wedge}x{\neq}\nil
  ~ \ent ~
  \seppredF{\code{llb}}{x{,}\nil{,}mi}}.}\label{mov.ctree}
  \savespace 
\end{figure}

\form{\enode_4}  (substitute \form{X} by \form{\nil})
and then
\form{\enode_5}, by \rulename{LBase} to unfold \form{\seppredF{\code{lls}}{\nil{,}\nil{,}m'{,}ma}^1} to the base case,
  as:
\[
\begin{array}{ll}
\enode_{4}{:}& \sepnodeF{x}{c_4}{\nil,m'}\sep
  \seppredF{\code{lls}}{\nil{,}\nil{,}m'{,}ma}^1
                \wedge x{\neq}\nil \wedge mi{\leq}m' 
  ~ \ent ~
  \seppredF{\code{llb}}{x{,}\nil{,}mi}  \\
\enode_{5}{:}& \sepnodeF{x}{c_4}{\nil,ma}
                \wedge x{\neq}\nil \wedge mi{\leq}ma  
  ~ \ent ~
  \seppredF{\code{llb}}{x{,}\nil{,}mi} 
\end{array}
\]
Now, \form{\enode_5} is in NF. {\toolname} applies \rulename{RInd}
and then \rulename{RBase} to 
\code{llb}
in the RHS as:
\[
\begin{array}{ll}
\enode_{6}{:}& \sepnodeF{x}{c_4}{\nil,ma}
                \wedge x{\neq}\nil \wedge mi{\leq}ma  
  \\ &
  \quad \ent ~ \sepnodeF{x}{c_4}{\nil,ma}\sep
  \seppredF{\code{llb}}{\nil{,}\nil{,}mi} \wedge mi{\leq}ma\\
\enode_{6'}{:}& \sepnodeF{x}{c_4}{\nil,ma}
                \wedge x{\neq}\nil \wedge mi{\leq}ma 
  ~ \ent ~ \sepnodeF{x}{c_4}{\nil,ma} {\wedge}mi{\leq}ma
\end{array}
\]
After that, as \form{mi{\leq}ma \imply mi{\leq}ma}, \form{\enode_{6'}} is applied with \rulename{Hypothesis}
to 
obtain \form{\enode_{7}}.
\[
\enode_{7}{:}~ \sepnodeF{x}{c_4}{\nil,ma}
                \wedge x{\neq}\nil \wedge mi{\leq}ma 
  ~ \ent ~ \sepnodeF{x}{c_4}{\nil,ma}
\]
As the LHS of \form{\enode_{7}} is in NF and a base formula,
 it is sound and complete to apply rule
\form{\sep} to have \form{\enode_{8}} as:
\form{ \emp
                \wedge x{\neq}\nil \wedge mi{\leq}ma
  ~ \ent ~ \emp}.
By \rulename{Emp}, \form{\enode_{8}} is decided as \form{\valid}.
For the right branch of the proof, \form{\enode_{3}} is applied with rule \rulename{{\neq}{\sep}} and then
\rulename{RInd} to obtain \form{\enode_{9}}:
\[
\begin{array}{ll}
\enode_{9}{:}& \sepnodeF{x}{c_4}{X,m'}{\sep}
  \seppredF{\code{lls}}{X{,}\nil{,}m'{,}ma}^1
                \wedge x{\neq}\nil  \wedge mi{\leq}m'
 \wedge X{\neq}\nil \wedge x{\neq}X
  \\ &
  \quad \ent ~ \sepnodeF{x}{c_4}{X,m'}{\sep}
  \seppredF{\code{llb}}{X{,}\nil{,}mi} {\wedge}mi{\leq}m'
\end{array}
\]
After that, \form{\enode_{9}} is applied with \rulename{Hypothesis} to eliminate the pure constraint
in the RHS:
\[
\begin{array}{ll}
     \enode_{10}{:}& \sepnodeF{x}{c_4}{X,m'}{\sep}
  \seppredF{\code{lls}}{X{,}\nil{,}m'{,}ma}^1
                \wedge x{\neq}\nil  \wedge mi{\leq}m' \wedge X{\neq}\nil \wedge x{\neq}X
  \\ &
  \quad \ent ~ \sepnodeF{x}{c_4}{X,m'}{\sep}
  \seppredF{\code{llb}}{X{,}\nil{,}mi}
\end{array}
\]
\form{\enode_{10}} is then applied with \rulename{\sep} to obtain
\form{\enode_{11}} and \form{\enode_{12}} as follows.
\[
\begin{array}{ll}
     \enode_{11}{:}& \sepnodeF{x}{c_4}{X,m'}~ \ent ~ \sepnodeF{x}{c_4}{X,m'}  \\
\enode_{12}{:}& 
  \seppredF{\code{lls}}{X{,}\nil{,}m'{,}ma}^1
                \wedge x{\neq}\nil \wedge mi{\leq}m' \wedge X{\neq}\nil \wedge x{\neq}X
  ~ \ent ~
  \seppredF{\code{llb}}{X{,}\nil{,}mi}
\end{array}
\]
\form{\enode_{11}} is valid by \rulename{Id}.
\form{\enode_{12}} is successfully linked back to \form{\enode_{0}}  to form
a pre-proof as
\[
\form{(\seppredF{\code{lls}}{X{,}\nil{,}m'{,}ma}^1
                {\wedge}X{\neq}\nil)[x/X,mi/m']
  ~\ent ~   \seppredF{\code{llb}}{X{,}\nil{,}mi}[x/X,mi/m'] }
\]
is identical to \form{\enode_0}. Since \form{\seppredF{\code{lls}}{X{,}\nil{,}m'{,}ma}^1} in \form{\enode_{12}}
is the subterm of \form{\seppredF{\code{lls}}{x{,}\nil{,}mi{,}ma}^0}
in \form{\enode_{0}}, 
our system decided that \form{\enode_{0}} is valid with the cyclic proof
presented in Fig. \ref{mov.ctree}.

\section{Soundness, Completeness, and Complexity}
\label{sec.deci}

 We describe the soundness, termination, and completeness of {\allEnt}. 
First, we need to show the invariant about the quantifier-free entailments of our system.
\begin{corollary}\label{cor.free.inv}
  Every entailment 
  derived from {\allEnt} is quantifier-free. 
\end{corollary}

The following lemma shows the soundness of the proof rules.
\begin{lemma}[Soundness]\label{ent.local.sound}
For each proof rule,
 if all premises are valid, then the conclusion is valid.
\end{lemma}

As every backlink
generated contains at least one pair of inductive predicate occurences in
 a subterm relationship,
the global soundness condition holds in our system.

\begin{lemma}[Global Soundness]\label{ent.global.sound}
A pre-proof derived
is indeed a cyclic proof.
\end{lemma}

The termination relies on the number of premises/entailments generated by rule \form{\sep}.
As the number of inductive symbols and their arities are finite,
there is a finite number of equivalent classes of these entailments in which
any two entailments in the same class are equivalent under some substitution and linked back together.
Therefore, the number of premises generated by rule \form{\sep}
is finite
considering the generation of backlinks.

\begin{lemma}\label{lem.term.shape}
{\allEnt} terminates.
\end{lemma}

In the following, we show the complexity analysis.
First, we show
that every occurrence of inductive predicates in the LHS is unfolded at most two times. 
\begin{lemma}\label{lem.term.comp.shape}
Given any entailment \form{\entailNCyc{\seppredF{\code{P}}{\setvars{v}}^k \sep \D_a}{\D_c}{\heap}}, then \form{0 \leq k \leq 2}.
\end{lemma}

Let n be the maximum number of predicates
(both inductive predicates and points-to predicates) among 
the LHS of the input and the definitions in \form{\PName}, and
\form{m} be the maximum number of fields of data structures.
Then, the complexity is defined as follows.

\begin{proposition}[Complexity]\label{prop.complex}
${\entProb}$ is $\mathcal{O}(n \times 2^{m} + n^3)$.
\end{proposition}
As such, if \form{m} is bounded by a constant, the complexity becomes polynomial in time.

 Our completeness proofs are shown in two steps. First, we show the proofs for an entailment
 whose LHS is a base formula.
 Second, we show the correctness when the LHS contains inductive predicates.
In the following, we first define the base formulas of the LHS derived by {\allEnt} from occurrences of inductive predicates.
Based on that, we define bad models to capture counter-model of invalid entailments.

\begin{defn}[{\cslpluse} Base]\label{defn.composition.base} Given \form{\heap}, define \form{\basecomp{\heap}} as follows. 
\[
\begin{array}{l}
\basecomp{\seppred{\code{P}}{E{,}F{,}\setvars{B}{,}u{,}sc{,}tg}} ~{\defsym}~ \sepnodeF{E}{c}{F{,}E_1{,}E_2{,}u{,}tg} \sep \basecomp{\seppredF{\code{Q_1}}{E_1{,}B}} {\sep} \basecomp{\seppredF{\code{Q_2}}{E_2{,}F}}  {\wedge} \pure_0   \\
\basecomp{\sepnodeF{E}{c}{\setvars{v}}} ~{\defsym}~ \sepnodeF{E}{c}{\setvars{v}}
\qquad \quad
\basecomp{\emp} ~{\defsym}~ \emp \qquad \quad
\basecomp{\heap_1 {\sep} \heap_2} ~{\defsym}~ \basecomp{\heap_1} {\sep} \basecomp{\heap_2}
\end{array}
\]
\end{defn}
The definition  for general predicates
with arbitrary matrix heaps
is presented in \repconf{App. \ref{app.ent.red}.}{\cite{Loc:Lin:TR:2019}.}
As  \form{\PName} does not include mutual recursion (Condition {\bf C3}),
 the definition above terminates in a finite number of steps.
In a pre-proof,
these {\cslpluse} base formulas 
of the LHS of an entailment are obtained
once every inductive predicate has been unfolded once.

\begin{lemma}
If \form{\heap \wedge \pure} is in NF then \form{\basecomp{\heap} \wedge \pure} is in NF,
and \form{\basecomp{\heap}\wedge \pure ~\ent~ \heap} is valid.
\end{lemma}
In other words, \form{\basecomp{\heap}\wedge \pure} is an under-approximation
of \form{{\heap}\wedge \pure}; invalidity of  \form{\basecomp{\heap}\wedge \pure ~\ent~ \D'}
implies invalidity of  \form{{\heap}\wedge \pure ~\ent~ \D'}.

\begin{defn}[Bad Model] The bad model for \form{\basecomp{\heap} \wedge \phi \wedge \a}
in NF is obtained by assigning
\begin{itemize}
\item a distinct non-$\nil$ value to each variable in
\form{\FV(\basecomp{\heap} \wedge \phi)}; and
\item a value to each variable in \form{\FV(\a)}
such that \form{\a} is satisfiable.
\end{itemize}
\end{defn}

\begin{lemma}
\label{rule.complete}
\begin{enumerate}
\item For every proof rule except rule \rulename{\sep},
all premises are valid only if the conclusion is valid.
\item For rule \rulename{\sep} where the conclusion is of
the form \form{{\base} ~\ent~ \heap'},
all premises are valid only
if the conclusion is valid and 
    \form{\base} is in NF.
\end{enumerate}
\end{lemma}

\noindent The following lemma states the correctness of the procedure \code{is\_closed}
for cases 2(b-d).

\begin{lemma}[Stuck Invalidity]\label{stuck}
Given \form{\heap{\wedge}\pure~\ent~\D'} in NF, it is {\invalid} if
procedure \code{is\_closed} returns {\invalid} for
 cases 2(b-d).
\end{lemma}
A bad model of the \form{\basecomp{\heap}{\wedge}\pure} is a counter-model.
Cases 2b) and 2c) show that the heaps of bad models are not connected and thus
accordingly to conditions {\bf C1} and {\bf C2}, any model of the LHS could not be
a model of the RHS. Case 2d) shows that heaps of the two sides could not be
matched.
Now, we show the correctness of Case 2(a) of procedure \code{is\_closed}
and invalidity is preserved during the proof search in {\allEnt}.
\begin{proposition}[Invalidity Preservation]\label{complete}
 If {\allEnt} is stuck, the input is invalid.
\end{proposition}

\begin{theorem}\label{thm.complete}
\form{\entProb} is decidable.
\end{theorem}

\section{Implementation and Evaluation}\label{sec.impl}
We implement {\toolname} using OCaml.
This implementation is an instantiation of a general framework
for cyclic proofs.
To discharge satisfiability for a separation logic formula,
we utilize the cyclic proof systems to derive
bases for inductive predicates in the decidable fragment
shown in \cite{Le:VMCAI:2021}. For those formulas
beyond this fragment, we use the solver presented in \cite{Loc:CAV:2016,Loc:CAV:2017}.
 We also develop a built-in solver for discharging equalities.

We evaluated {\toolname} to show that
i) it
 can discharge 
 problems in {\cslpluse} effectively;
 and ii) its performance is compatible to
 the state-of-the-art solvers.

\paragraph{Experiment settings}
We have evaluated
 {\toolname}
 on entailment problems taken from SL-COMP 2022 \cite{slcomp:sl:22},
a competition of separation logic solvers.
We take the problems in
two divisions of the SL-COMP 2022, {\em qf\_shls\_entl} and {\em qf\_shlid\_entl},
and one new division {\em qf\_shlid2\_entl}.
All these problems semantically belongs to our decidable fragment and their syntax are written in SMT 2.6 format \cite{Sighireanu:TACAS:19}. 
\begin{itemize}
 \item 
Division {\em qf\_shls\_entl} includes 296 entailment problems,
\form{122} {\invalid} problems and \form{174} {\valid} problems,
with only singly linked lists.
They were randomly generated
by the authors in \cite{pldi:PerezR11}.

\item Division {\em qf\_shlid\_entl} contains 60 entailment problems
which were mostly handcrafted by the authors in \cite{Enea:FMSD:2017}.
They include
 singly-linked lists,  doubly-linked lists, lists of singly-linked lists
 or skip lists. Furthermore, the system of inductive predicates must satisfy the following
 condition: For two different predicates \form{\code{P}}, \form{\code{Q}} in the system
 of definitions, 
 either \form{\code{P} \orderstarp \code{Q}} or \form{\code{Q} \orderstarp \code{P}}.
 \item In the third division, we introduce new benchmarks, with 27 problems,
 that are beyond the problems in the previous two divisions. In particular,
 in every system
 of predicate definitions,
 there exist  two predicates \form{\code{P}}, \form{\code{Q}}  such that they are semantically equivalent.
We have submitted this division to the Github
 repository of SL-COMP.
\end{itemize} 

 To evaluate {\toolname}'s performance,
we compared it with the state-of-the-art tools such as {\cyclic} \cite{Brotherston:CADE:11},
{\spen} \cite{Enea:FMSD:2017},
 {\songbird} \cite{Ta:FM:2016}, SLS \cite{Ta:POPL:2018}
and {\harr} \cite{Jens:TACAS:2019}. We did not include Cycomp \cite{Makoto:APLAS:2019},
 as these benchmarks are beyond its decidable fragment.
 Note that {\cyclic}, {\songbird} and SLS are not complete; for non-valid problems,
 while {\cyclic} returns {\unknown},
 {\songbird}, and SLS use some heuristic to guess the outcome.
For each division, we report the number of correct outputs ({\invalid}, {\valid}) and the time (in minutes and seconds) taken by
each tool.
Note that we use the status ({\invalid}, {\valid}) annotated with each problem in the SL-COMP benchmark as the ground truth.
If an output is the same with the status, we classify it as correct; otherwise,
it is marked as incorrect.
We also note that in these experiments we used the competition pre-processing tool \cite{Sighireanu:TACAS:19}
to transform the SMT 2.6 format into the corresponding formats of the tools before running them.
All experiments were performed on a machine with
 Intel Core i7-6700 CPU 3.4Gh and 8GB RAM. The CPU timeout is 600 seconds.

\label{impl.entl}

\begin{table}[tb]
\begin{center}
\caption{Experimental results} \label{tbl:expr.entl}
\begin{tabular}[t]{|c|c  | c | c |c  | c | c |c  | c | c |}
\hline
 Tool            &   \multicolumn{3}{c|}{{\em qf\_shls\_entl}} &   \multicolumn{3}{c|}{{\em qf\_shlid\_entl}} &   \multicolumn{3}{c|}{{\em qf\_shlid2\_entl}}  \\
 \hline
 & {\small \invalid} &  {\small \valid} & Time & {\small \invalid} &  {\small \valid} & Time & {\small \invalid} &  {\small \valid} & Time \\
    & {\small (122)} &  {\small (174)} & (296)  & {\small (24)} &  {\small (36)} &  (60) & {\small (14)} &  {\small (13)} & (27)  \\
\hline
SLS & 12 & 174 & 507m42s & 2 & 35 & 133m28s & 0 & 11 & 97m54s \\
\rowcolor{Gray} \spen & 122 & 174 & 10.78s & 14 & 13 & 3.44s & 8 & 2 & 1.69s\\
\cyclic & 0 & 58 & 1520m5s & 0 & 24 & 360m38s & 0 & 3 & 240m3s \\
\rowcolor{Gray} \harr & 39 & 116 & 425m19s & 18 & 27 & 53m56s & 8 & 7 & 156m45s \\
 \songbird & 12 & 174 & 237m25s & 2 & 35 & 40m38s & 0 & 12 & 47m11s \\
\rowcolor{Gray}\toolname & 122  & 174 & 6.22s & 24  & 36 & 0.96s & 14 & 13 & 1.20s \\
\hline
 \end{tabular}
\end{center}\savespace \savespace
\end{table}

\paragraph{Experiment results}
The experimental results are reported in Table \ref{tbl:expr.entl}.
In this table, the first column presents
the names of the tools.
The next three columns show the results of the first division
 including
the number of correct {\invalid} outputs,
the number of correct {\valid} outputs and
the time taken (where {\em m} for minutes and {\em s} for seconds), respectively.
In the third row,
the number between each pair of brackets {\em (...)} shows the number of problems
 in the corresponding column.
Similarly, the next two groups of six columns describe the results of the second
and third divisions, respectively.

In general, the experimental results show that {\toolname} is the one
(and only one) that could produce all the
correct results.
Other solvers either produced wrong results or could discharge a fraction of the experiments.
Moreover, {\toolname} took a short time for the experiments (8.38 seconds compared
to 15.91 seconds for {\spen}, 324 minutes for {\songbird}, 635 minutes
for {\harr}, 739 minutes for SLS and 2120 minutes for {\cyclic}).
While SLS returned 14 false negatives, {\spen} reported 20 false positives.  {\cyclic}, {\songbird} and {\harr} did not
produce any wrong result.
Of 569 tests,
while {\cyclic} could handle 85 tests (15\%), {\harr} could handle 215 tests (38\%)
and
{\songbird} could decide 235 tests (41.3\%).
 In total of 223 {\valid} tests,
while {\cyclic} could handle 85 problems (38\%), 
{\songbird} could decide 222 problems (99.5\%).

Now we examine the results for each division in details.
For {\em qf\_shls\_entl},
{\spen} returned all correct, 
{\songbird} 186,
{\harr} 155, and
{\cyclic} 58.
If we set the timeout to 2400 seconds, both
{\songbird} and {\harr} produced all the correct results.
For division {\em qf\_shlid\_entl} includes
 \form{24} {\invalid} problems and \form{36} {\valid} problems.
While {\songbird} produced
37 problems correctly,
%
{\cyclic} produced 24 correct results.
{\spen} reported 27 correct results
and 13 false positives
(\code{skl2{-}vc\{01-04\}}
\code{skl3{-}vc01}, \code{skl3{-}vc\{03-10\}}).
For
the last division {\em qf\_shlid2\_entl} includes 14 {\invalid} test problems and 13 {\valid} test problems. While
{\songbird} decided only 12 problems correctly,
{\cyclic} produced 3 correct outcomes.
{\spen} reported 10 correct results. However, it produced
7 false positives (\code{ls{-}mul{-}vc\{01-03\}},
 \code{ls{-}mul{-}vc05}, \code{nll{-}mul{-}vc\{01-03\}}).

Since our experiments provide break-down results of the two divisions of
SL-COMP competition, we hope that they provide an initial understanding
of the SL-COMP benchmarks and tools. Consequently, this might reduce
the effort to prepare experiments over these benchmarks
 to evaluate new SL solvers.
%
%
Finally, one might point out that {\toolname} performed well because the entailments in the experiments
are within its scope. We do not totally disagree with this argument, but would like to emphasize
that tools do not always work well on favorable benchmarks. For example, {\spen} introduced wrong results on {\em qf\_shlid\_entl},
 and {\harr} did not handle {\em qf\_shlid\_entl} and {\em qf\_shlid2\_entl} well
 although these problems are in their decidable fragments.
 We believe that engineering design and effort play an important role along side with
 theory development.

\section{Related Work} \label{sec.related}

{\toolname} is a variant of the cyclic proof systems
\cite{Brotherston:05,Brotherston:CADE:11,Brotherston:APLAS:12,Le:APLAS:2018} and
\cite{Makoto:APLAS:2019}. 
Unlike existing cyclic proof systems, the soundness of {\toolname} is local,
and the proof search is not back-tracking.
The work presented in \cite{Makoto:APLAS:2019} shows the completeness of the cyclic proof system.
Its main contribution is the introduction of rule \form{\sep} for
those entailments with disjunction in the RHS obtained from predicate unfolding.
 In contrast to \cite{Makoto:APLAS:2019}, our work includes
 normalization to soundly and completely avoid disjunction in the RHS during unfolding.
Our work also presents how to obtain the global soundness condition for cyclic proofs.
Moreover, our decidable fragment {\cslpluse} is non-overlapping to
the cone predicates introduced in \cite{Makoto:APLAS:2019}.
Furthermore, due to the empty heap in the base cases, the
matching rule in \cite{Makoto:APLAS:2019} cannot be applied to the predicates in {\cslpluse}.

Our work relates to the inductive theorem provers
introduced in \cite{Chu:PLDI:2015}, \cite{Ta:FM:2016} and Smallfoot \cite{Berdine:APLAS05}.
While \cite{Chu:PLDI:2015} is based on structural induction, \cite{Ta:FM:2016} is based on mathematical
induction.
  Smallfoot \cite{Berdine:APLAS05}
proposed a decision procedure for a fragment with linked lists
and trees (and without arithmetic).
To handle inductive entailments,
this system made use of a fixed compositional rule as
 consequences of induction reasoning.
This technique was further explored by the authors
in \cite{pldi:PerezR11}.
Compared with  Smallfoot,
our proof system replaces the compositional rule by the combination of
rule
\rulename{LInd} and the back-link construction.
In doing so, our system could support induction reasoning on
a much more expressive fragment of inductive predicates.

Our proposal also relates to works that use lemmas as consequences of induction reasoning
\cite{Berdine:APLAS05,EneaSW:ATVA:15,Loc:TACAS:2018,Ta:POPL:2018}.
These works in \cite{EneaSW:ATVA:15,Loc:CAV:2014,Loc:TACAS:2018,Ta:POPL:2018}
 automatically generate lemmas
for some classes of inductive predicates.
S2 \cite{Loc:CAV:2014} generated lemmas to normalize (such as split, equivalence)
the shapes of the synthesized data structures.
\cite{EneaSW:ATVA:15} proposed to generate several
 sets of lemmas not only for compositional predicates, but also for different predicates (e.g., completion lemmas, stronger lemmas and static parameter contraction lemmas).
To prove an entailment,
SLS \cite{Ta:POPL:2018} aims to infer general lemmas.
Similarly, {\stool} \cite{Loc:TACAS:2018} solves a more generic problem, frame inference, using cyclic proofs
and lemma synthesis.
It first infers shape-based residual frame in the LHS
and then synthesizes the pure constraints over the two sides.
It would be a future work to integrate the pure constraint synthesis
into {\toolname} to
support non-local pure properties.

{\toolname} relates to model-based decision procedures
that reduce the entailment problem in separation logic to
a well-studied problem in other domains. For instance,
in \cite{Chen:CONCUR:2017,Cook:CONCUR:2011,Gu:IJCAR:2016}
the entailment problem including singly-linked lists and their invariants
is reduced to the problem of inclusion checking in a graph
theory.
The authors in \cite{Iosif:CADE:13} reduced the entailment problem
to the satisfiability problem in second-order monadic logic.
 This reduction could
 handle an expressive fragment of spatial-based predicates, called bounded-tree width.
Recently, the work presented in \cite{Jens:TACAS:2019} show a model-based decision procedure
for a subfragment of the bounded-tree width.
Furthermore, while the work in
\cite{Enea:FMSD:2017,Radu:ATVA:2014}
reduced the entailment problem
to the tree automata inclusion checking problem,
\cite{Jansen:ESOP:2017} presented 
an idea to reduce the problem to the heap automata inclusion checking problem.
Moreover, while the procedure in \cite{Enea:FMSD:2017} supported well
compositional predicates (single and double links),
the procedure in \cite{Radu:ATVA:2014} could handle
predicates satisfying local properties (e.g., trees with parent pointers).
Our decidable fragment subsumes the one described in \cite{Berdine:APLAS05,Cook:CONCUR:2011,Enea:FMSD:2017} but is incompatible to the ones presented in \cite{Chen:CONCUR:2017,Gu:IJCAR:2016,Iosif:CADE:13,Radu:ATVA:2014}.
Works in
\cite{Navarro:APLAS:2013} and \cite{Piskac:CAV:2013,Piskac:CAV:2014} reduced the entailment problem in separation
logic into the satisfiability problem in SMT. 
While GRASShoper \cite{Piskac:CAV:2013,Piskac:CAV:2014} could handle transitive closure pure properties,
{\toolname} is capable of supporting local ones.
Unlike GRASShoper, which reduces entailment into SMT problems,
 {\toolname} reduces an entailment to admissible entailments and detects repetitions via cyclic proofs.

Our work relates to decidable fragments and
 complexity results of the entailment problem in separation logic with inductive
predicates. The entailment is 2-EXPTIME 
in cone predicates \cite{Makoto:APLAS:2019}, the bounded tree width predicates and beyond
 \cite{Iosif:CADE:13,echenim2021unifying},
and EXPTIME in a sub-fragment of cone predicates
 \cite{Radu:ATVA:2014}. In the other class,
entailment is in polynomial time
 for singly-linked lists
\cite{Cook:CONCUR:2011}, semantically linear inductive predicates \cite{Enea:FMSD:2017},
 and its extensions with arithmetic \cite{Gu:IJCAR:2016}
(but becomes EXPTIME when the lists are extended with
double links \cite{Chen:CONCUR:2017}).
Our fragment (with nested lists, trees and arithmetic properties) is roughly
 in the ``middle'' of the two classes above
where the entailment is EXPTIME and becomes polynomial
under the upper bound restriction.


\section{Conclusion}\label{sec.conc}
We have presented a novel decision procedure  for
the quantifier-free entailment problem in separation logic
combining with
inductive definitions of compositional predicates
 and pure properties.
Our proposal is the first complete cyclic proof system for the problem in separation logic
without back-tracking.
We have implemented the proposal in {\toolname} and evaluated it
over the set of nontrivial entailments taken from the SL-COMP competition.
The experimental results show that
our proposal is both effective and efficient when compared
against the state-of-the-art solvers.

For future work, we plan to combine this proposal
with the cyclic frame inference procedure presented in
\cite{Loc:TACAS:2018}
for a bi-abductive procedure. This is a basic step
to obtain
a compositional shape analysis beyond the lists and trees.
Another work is to formally prove that our system
is as strong as Smallfoot in the decidable fragment with lists and trees \cite{Berdine:APLAS05}: Given an entailment, if Smallfoot can produce a proof,
so is {\toolname}.

\bibliography{refs}
 \bibliographystyle{plain}

  \repconf{
\newpage
 \appendix

\section{Reduction Rules for Compositional Predicates in General Form}\label{app.ent.red}
\begin{figure}[tb]
\begin{center}
\[
\AxiomC{$\begin{array}{c}
    \sub {=} \circ \{\setvars{v}_i {/} \setvars{p}_i \mid \setvars{p}_i \in \setvars{w} \wedge \setvars{p}_i {\neq}\nil \} \\
    { \sepnodeF{x}{c}{\setvars{v}} {\sep}\heap_1{\wedge}\pure_1{\wedge}x{\neq}F}
 ~ \ent~ { (\exists (\setvars{w}{\setminus}\setvars{p}) . \sepnodeF{x}{c}{\setvars{p}}{\sep}\heap'{\sep}\seppredF{\code{P}}{w{,}F{,}\setvars{B}{,}u{,}sc'{,}tg}{\wedge}\pure_0)\sub{\sep}\heap_2 {\wedge}\pure_2} 
  \end{array}$}
\RightLabel{\scriptsize $\begin{array}{c} \dagger
 \end{array}$}
\LeftLabel{\scriptsize RInd}
\UnaryInfC{$\begin{array}{l}
\entailNCyc{  \sepnodeF{x}{c}{\setvars{v}} {\sep}\heap_1{\wedge}\pure_1{\wedge}x{\neq}F}{ \seppredF{\code{P}}{x{,}F{,}\setvars{B}{,}u{,}sc{,}tg}{\sep} \heap_2{\wedge}\pure_2}{\heap}
\end{array}$}
\DisplayProof
\]
\[
\AxiomC{
$
\begin{array}{l}
  { (\exists (\setvars{w}{\setminus}\setvars{p}) . \sepnodeF{x}{c}{\setvars{p}}{\sep}\heap'{\sep}\seppredF{\code{P}}{w{,}F{,}\setvars{B}{,}u{,}sc'{,}tg}^{k{+}1}{\wedge}\pure_0){\sep}\heap_1 {\wedge}\pure_1{\wedge}x{\neq}F_3} \\ \qquad \ent~{ \seppredF{\code{Q}}{x{,}F_3{,}\setvars{B}{,}u{,}sc{,}tg_2} {\sep} \heap_2{\wedge}\pure_2}  \\
\end{array}
$
}
\RightLabel{\scriptsize $\seppredF{\code{P}}{x{,}F{,}\setvars{B}{,}u{,}sc{,}tg}{\not\in}\heap_2$}
\LeftLabel{\scriptsize LInd}
\UnaryInfC{$\begin{array}{l}
    \entailNCyc{\seppredF{\code{P}}{x{,}F{,}\setvars{B}{,}u{,}sc{,}tg}^k{\sep}\heap_1{\wedge}\pure_1{\wedge}x{\neq}F_3
      }{\seppredF{\code{Q}}{x{,}F_3{,}\setvars{B}{,}u{,}sc{,}tg_2} {\sep} \heap_2{\wedge}\pure_2}{\heap}
\end{array}$}
\DisplayProof
\]
\caption{Reduction Ruleswhere  $\dagger{:}~ \sepnodeF{x}{c}{\setvars{v}}  {\not\in} \heap_2$
}
\label{fig.gen.ent.comp.red}
\end{center}\savespace
\end{figure}
In Figure \ref{fig.gen.ent.comp.red},
we present rules \rulename{RInd1} and \rulename{LInd}
for the following definitions of compositional predicates:
\[
  \seppredF{\code{P}}{x{,}F{,}\setvars{B}{,}u{,}sc{,}tg} \equiv \emp {\wedge}x{=} F {\wedge}sc{=}tg ~\vee~ \exists \setvars{w} . \sepnodeF{x}{c}{\setvars{p}}{\sep}\heap'{\sep}\seppredF{\code{P}}{w{,}F{,}\setvars{B}{,}u{,}sc'{,}tg}{\wedge}\pure_0;
\]
where $\setvars{w}$ are fresh variables.

To define {\cslpluse} base in a general form,
we further
 assume every heap cells \form{c_i \in \Dns} used in definitions
of compositional predicates  \form{\seppred{\code{P}}{r{,} F{,} \setvars{B}{,}u{,}sc{,}tg}}
are defined in the form of
\form{\code{data}~ c_i \{c_i ~ next;c_{i_1}~down_1;...;c_{i_j}~down_j;\tau_u ~udata;\tau_s~scdata\}}
where \form{c_{i_1},..,c_{i_j} \in \Dns},
\form{down_1},..., \form{down_j} fields are for the nested structures in the matrix heaps,
\form{udata} field is for the transitivity data, and \form{scdata} field are for 
ordering data. Then, {\cslpluse} base of an occurrence
of the compositional predicates is defined as:
\[
\basecomp{\seppred{\code{P}}{E{,}F{,}\setvars{B}{,}u{,}sc{,}tg}} ~{\defsym}~ \sepnodeF{E}{c}{F{,}\setvars{d},tg{,}u}\subst{\setvars{v}}{\setvars{d}}{\wedge} \pure_0\subst{tg}{scd} \sep \basecomp{\heap'(\subst{\setvars{v}}{\setvars{d}}\circ[tg/scd])}
\]

\section{Proof of Corollary \ref{cor.free.inv}}
\begin{proof}
We need to show that the premises in rule \rulename{RInd}
and rule \rulename{LInd} are quantifier-free.
The condition {\bf C1} in section \ref{spec.deci.ent} ensures that
\form{\setvars{w} \subseteq \setvars{p}}. Hence,
\form{\setvars{w} \setminus \setvars{p} \equiv \emptyset}.
Thus, the RHS of the premise in \rulename{RInd}
and the LHS of the premise in \rulename{LInd} are quantifier-free.
\end{proof}

\section{Proof of Soundness}
We show the correctness of the soundness
 of the proof system.

\subsection{Soundness of proof rules: Lemma \ref{ent.local.sound}}

For each rule, we show that if all the premises hold, so is the conclusion

\paragraph{Rule \rulename{Subst}}. First, we consider the
case \form{E} is a variable.
Suppose \form{\D[v/x] ~\ent~ \D'[v/x]}.
That is for any {\sstack, \sheaps}, if \form{\sstack, \sheaps \models \D[v/x]} then \form{\sstack, \sheaps \models \D'[v/x]}.
As \form{x \not\in \FV(\D[v/x])}, we
 extend the domain
of stack with \form{x} as:
\form{\sstack' = \sstack[x {\pto} {\sstack(v)}]}.
As so, \form{\sstack', \sheaps \models \D \wedge v=x}
and \form{\sstack', \sheaps \models \D'}.
Therefore \form{\D \wedge v=x ~\models~ \D'} holds.

The case \form{E} is \form{\nil} is similar.

\paragraph{Rule \rulename{ExM}}
For simplicity, we assume that \form{E_1} and \form{E_2}
are both variables.
Suppose \form{\D \wedge v_1=v_2 ~\ent~ \D'}
and \form{\D\wedge v_1\neq v_2 ~\ent~ \D'}.

Suppose \form{\sstack, \sheaps \models \D}.

\begin{itemize}
\item Case 1: if \form{\sstack(v_1) = \sstack(v_2)} then
\form{\sstack, \sheaps \models \D \wedge v_1=v_2}.
As \form{\D \wedge v_1=v_2 ~\ent~ \D'},
 \form{\sstack, \sheaps \models \D'}.
\item Case 1: if \form{\sstack(v_1) \neq \sstack(v_2)} then
\form{\sstack, \sheaps \models \D \wedge v_1\neq v_2}.
As \form{\D \wedge v_1\neq v_2 ~\ent~ \D'},
 \form{\sstack, \sheaps \models \D'}.
\end{itemize}
\paragraph{Rule \rulename{{=}L}, rule \rulename{{=}R}, and
rule
\rulename{Hypothesis}} Trivial.

\paragraph{Rule \rulename{LBase} and rule \rulename{RBase}}
based on the fact that given a compositional predicate
\form{\seppredF{\code{P}}{E{,}F{,}\setvars{B}{,}u{,}sc{,}tg}} where \form{F} is a dangling pointer,
then \form{\form{\seppredF{\code{P}}{E{,}E{,}\setvars{B}{,}u{,}sc{,}tg}}}
implies the base rule with \form{\emp} heap predicate.

\paragraph{Rule \rulename{{\neq}\nil}} 
Follows semantics of points-to predicate
 where \form{\nil \not\in \Locations}.

\paragraph{Rule \rulename{{\neq}\sep}}
Follows semantics of the spatial conjunction \form{\sep}.

\paragraph{Rule \rulename{\sep}}

Suppose \form{\heap_1\wedge \pure \models \heap_2\wedge \pure' }
and \form{\heap\wedge \pure \models \heap'\wedge \pure' }.

For any \form{\sstack, \sheaps_1 \models \heap_1\wedge \pure},
\form{\sstack, \sheaps_1 \models \heap_2\wedge \pure'}.
And any \form{\sstack, \sheaps_2 \models \heap \wedge \pure},
\form{\sstack, \sheaps_2 \models \heap' \wedge \pure}.
as \form{\code{roots}(\heap_1)\cap\code{roots}(\heap)=\emptyset}
\form{\dom(\sheaps_1) \cap \dom(\sheaps_2) = \emptyset}.
Hence \form{\sstack, \sheaps_1\dot\sheaps_2 \models \heap_1\sep \heap_2 \wedge \pure} (a).
Similarly, \form{\sstack, \sheaps_1\dot\sheaps_2 \models \heap_2\sep \heap' \wedge \pure'} (b).

From (a), (b), \form{\heap_1\sep \heap_2 \wedge \pure \models  \heap_2\sep \heap' \wedge \pure'}.

\paragraph{Rule \rulename{LInd} and \rulename{RInd}}.
Based on the least semantics of the inductive predicates
and the base case could not happen due to constraint
\form{x\neq F} in \rulename{RInd}
(respectively \form{x\neq F_3} in \rulename{LInd}).

\subsection{Global Soundness: Lemma \ref{ent.global.sound}}
\begin{figure}[t!]
\centering
\begin{tikzpicture}[node distance=18mm,level 1/.style={sibling distance=22mm},
      level 2/.style={sibling distance=7mm},
                        level distance=22pt, draw]
  \tikzstyle{every state}=[draw,text=black]

\node (A)                    {{$\D_0$}};
  \node         (B) [below left=4mm and 8mm of A] {$\D_{1}$};
  \node         (C) [below right=4mm and 8mm of A] {\textcolor{blue}{$\D_{2}^\clubsuit$}};
\node         (C4) [below right=3mm and 1mm of C] {....};
\node         (C1) [below right=3mm and 1mm of C4] {\textcolor{red}{$\D_{3}$}};
  \node    (C3) [below left=3mm and 1mm of C] {{$\D_{7}$}};
 \node         (C12) [below right=3mm and 1mm of C1] {{...}};
 \node         (C2) [below right=3mm and 1mm of C12] {{$\D_{4}$}};
 \node         (D1) [below left=3mm and 1mm of C2] {...};
  \node         (D) [below left=3mm and 1mm of D1] {\textcolor{red}{$\D_{5}$}};
  \node         (E1) [below right=3mm and 1mm of C2] {...};
\node         (E) [below right=3mm and 1mm of E1] {\textcolor{blue}{$\D_{6}^\clubsuit$}};

  \path (A) edge              node {} (B)
            edge              node {} (C)
        (C) edge              node {} (C4)
            edge              node {} (C3)
        (C4) edge    [dotted]          node {} (C1)
        (C1) edge              node {} (C12)
        (C12) edge    [dotted]          node {} (C2)
        (C2) edge              node {} (D1)
            edge              node {} (E1)
        (D1) edge   [dotted]           node {} (D)
       (E1) edge   [dotted]           node {} (E)
        (E) edge [->,bend right=60,dotted]  node {} (C)
            (D) edge [->,bend left=60,dotted]  node {} (C1)
;
\end{tikzpicture}
\caption{An Example of Non-Disjoint Back-Links.}
\label{fig.cex.tree}
\end{figure}
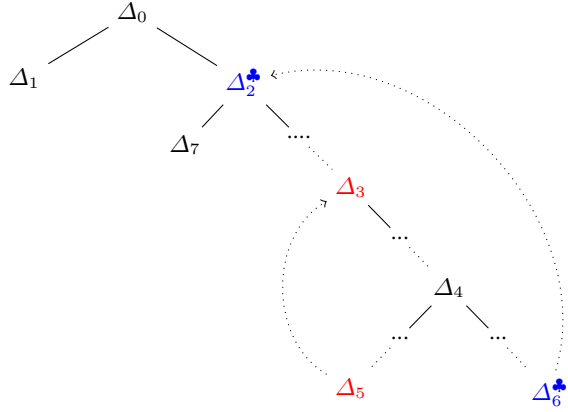

As our system always generates back-links with progressing points
(via rule \rulename{LInd}), there are infinitely progressing
points in any infinite trace.

We now
show that all cycles are pairwise disjoint (such that in the path
between a companion and a bud of every back-link, no rule can ever ``delete'' an inductive predicate
formula on which the soundness relies).
We prove by contradiction.

In intuition, the soundness replies on a pair of
inductive predicates in a sub-term relationship.
 Given an inductive predicate \form{\seppred{\code{P}}{E{,} F{,} \setvars{B}{,}\setvars{v}}}
 only rule \rulename{ExM} is able to generate
the constraint \form{E{=}F} such that
\form{\seppred{\code{P}}{E{,} F{,} \setvars{B}{,}\setvars{v}}} can be transformed into  \form{\seppred{\code{P}}{F{,} F{,} \setvars{B}{,}\setvars{v}'}} via rule \rulename{Subst}
and finally eliminated by rule \rulename{LBase}.
We now show that every companion node of a back-link involving a bud that is in the branch
\form{E{\neq}F} of rule \rulename{ExM}
is below the node including
 the applications of rule \rulename{ExM}.

Assume that our system generates back-links with
 non-disjoint cycles. Two cycles are non-disjoint only
when their both companion nodes are above at least one branch (applications of rule \rulename{ExM}).
 The non-disjoint cycles is similar to the one as shown in the proof tree
 in Fig. \ref{fig.cex.tree} where there is no branch in the path between \form{\D_2} and \form{\D_3}.
(The proof for the case where \form{\D_5} is linked with \form{\D_2} and \form{\D_6} is linked with \form{\D_3} is similar.
We discuss this proof below.)

We prove the contradiction by case analysis on the pair of
variables \form{E_1} and \form{E_2}
applied with rule \rulename{ExM} at node \form{\D_4}.
As \rulename{ExM} is applied to introduce \form{G(op(E))}
for every \form{op(E)} in the LHS of entailments.
We proceed case analysis on \form{op(E)}.

\begin{enumerate}
\item Case 1. \form{op(E) \equiv E{\pto}\anon} and \rulename{ExM} at node \form{\D_4} does case split \form{E{=}\nil} and
\form{E{\neq}\nil} to obtain two children. Assume that
the left child (on the path from \form{\D_4} to \form{\D_5})
 is \form{\D_4 {\wedge} E{=}\nil}. After substitution, LHS of
 this node
is reduced to \form{\D_4' \sep \nil{\pto}\anon}
which is equivalent to \form{\false}. Thus, the back-link from
\form{\D_5} to \form{\D_3} could not established.
\item Case 2. \form{op(E) \equiv \seppred{\code{P}}{E{,} F{,} \setvars{B}{,}\setvars{v}}} and \rulename{ExM} at node \form{\D_4} does case split \form{E{=}F} and
\form{E{\neq}F} to obtain two children.
 Assume that
the left child (on the path from \form{\D_4} to \form{\D_5})
 is \form{\D_4 {\wedge} E{=}F}.
After substitution, LHS of this node
is reduced to \form{\D_4' \sep \seppred{\code{P}}{F{,} F{,} \setvars{B}{,}\setvars{v}'}}.
 In turn, this entailment is applied with normalization rule
\rulename{LBase} to eliminate \form{\seppred{\code{P}}{F{,} F{,} \setvars{B}{,}\setvars{v}'}}.
Next, we consider two sub-cases of the inductive predicate in any
application of rule \rulename{LInd} applied into a node
 between \form{\D_3} and \form{\D_5}.
\begin{enumerate}
\item the predicate applied is  \form{\seppred{\code{P}}{E'{,} F{,} \setvars{B}{,}\setvars{v}}} .  We note that in the recursive rule
of definitions of compositional predicates
\[
\form{\exists X{,} sc'{,}d_1{,} d_2 . \sepnodeF{x}{c}{X{,}d_1{,}d_2{,}u{,}sc}{\sep}\seppredF{\code{Q_1}}{d_1{,}B}{\sep}\seppredF{\code{Q_2}}{d_2{,}X}{\sep}\seppredF{\code{P}}{X{,}F{,}\setvars{B}{,}u{,}sc'{,}tg}{\wedge}\pure_0}
\]
all  nested predicates \form{Q_1}, \form{Q_2}
are syntactically different to \form{P}.
 \form{\D_5} is missing
one occurrence of predicate \form{P}. Hence, it could not be linked
back to \form{\D_3}.
\item the predicate applied is  \form{\seppred{\code{Q}}{E'{,}F'{,}F{,}\setvars{v}'}} such that
\form{\seppredF{\code{P}}{U{,}F{,}\setvars{B}{,}u'{,}sc'{,}tg'}} is a nested predicate in the definition
of \form{Q}. However, \form{u'}, \form{tg'} are fresh variables
and in any back-links they are never substituted to become \form{u} and \form{tg},
respectively. Hence, \form{\D_5} could not be linked
back to \form{\D_3}.
\end{enumerate}
\item Case 3.  \form{op(E) \equiv \seppred{\code{tree}}{E{,} \setvars{B}{,}\setvars{v}}} and \rulename{ExM} at node \form{\D_4} does case split \form{E{=}\nil} and
\form{E{\neq}\nil} to obtain two children.
As \rulename{LInd} only applies for compositional predicates,
\form{\seppred{\code{tree}}{E{,} \setvars{B}{,}\setvars{v}}} could not be
a fresh formula. It has been normalised in
 \form{\D_3} already.
This case could not be occurred.
\end{enumerate}

The proof for the case where \form{\D_5} is linked with \form{\D_2} and \form{\D_6} is linked with \form{\D_3} is similar.
The main difference is that we need
to show that predicate \form{\seppred{\code{P}}{E{,} F{,} \setvars{B}{,}\setvars{v}}} is a sub-formula
of \form{\D_2}  in the proof of {\bf Case 2} like above. That means \form{\seppred{\code{P}}{E{,} F{,} \setvars{B}{,}\setvars{v}}}
has not been eliminated by rule \rulename{ExM} in the path between \form{\D_2}
\form{\D_3}. This is straightforward as
no branch exists in the path between \form{\D_2} and \form{\D_3}.

\section{Proofs of Termination}\label{app.term}
\subsection{Proof of Lemma \ref{lem.term.shape}}
\begin{proof}
Termination of our system is based on the size of an entailment which
is defined as:
\begin{defn}[Size]
The size of an entailment \form{\enode{:}\entailNCyc{ \heap_a{\wedge}\phi_a{\wedge}\a_a}{\heap_c{\wedge}\phi_c{\wedge}\a_c}{\heap}} is a triple of:
\begin{enumerate}
\item \form{N_p{-}n_p} where \form{N_p} is the maximal number of both points-to predicates and occurrences of inductive predicates that
the RHS of any entailments derived (by {\allEnt}) from {\enode}
may contain,
and \form{n_p} is the total number of both points-to predicates and occurrences of inductive predicates in \form{\heap_c}.
\item \form{N_e{-}n_e} where \form{N_e} is the maximal number of both disequalities and non-trivial equalities that the LHS
of any entailments derived (by {\allEnt}) from {\enode}
may contain, and \form{n_e} is the number of both disequalities and non-trial equalities in \form{\phi_a}.
\item the sum of the length of \form{\heap_a{\wedge}\phi_a{\wedge}\a_a ~\ent~\heap_c{\wedge}\phi_c{\wedge}\a_c},
where length is defined in the obvious way taking all
simple formulas to have length 1.
\item \form{N_a}: the number of constraints on arithmetic properties generated by the recursive
rules of inductive definition.
\end{enumerate}
\end{defn}

If \form{N_p} and \form{N_e} are bounded, applying any rules except \rulename{LInd}
 makes
progress since the size of each premise of
any rule application is lexicographically less than the size of the conclusion.
\form{N_p} and \form{N_a} rely on the number of applications of rule \rulename{LInd}.
\form{N_e} depends on the number applications of rule \rulename{ExM}. In turn, the application of \rulename{ExM} relies on the number
of spatial variables. Thus, \form{N_e} also relies on the number applications of rule \rulename{LInd}.
To show the termination, we show that the number applications of rule \rulename{LInd}
is bounded.
In consequence, this bound is achieved if the number of applications of \rulename{\sep} is finite.
As the number of inductive symbols as their arities are finite,
rule \form{\sep} indeed generates a finite number of equivalent classes of entailments in which two entailments
in the same class are equivalent after some substitution.
Thus, all entailments in the same class are linked back together
through a finite number of steps.
\end{proof}
\subsection{Proof of Lemma \ref{lem.term.comp.shape}}
Suppose we have an entailment \form{\entailNCyc{\seppredF{\code{P}}{E{,}F{,}\setvars{B}{,}\setvars{v}}^k \sep \heap {\wedge} \pure}{\heap'{\wedge}\pure'}{\heap}}.
If \form{\pure \not{\models}\pure'} then exhaustively applying rule \rulename{ExM}
our system decides it as {\invalid} through the base cases like \form{E{=}F}.

If \form{\pure {\models}\pure'}, then our system applies rule 
\rulename{Hypothesis}
to obtain \form{\entailNCyc{\seppredF{\code{P}}{E{,}F{,}\setvars{B}{,}\setvars{v}}^k \sep \heap {\wedge} \pure}{\heap'}{\heap}}. Hence, in the following proof, we only consider the later form of
the entailment in conclusion of rule \rulename{LInd}.

Without loss of generality, we assume \form{\PName} includes \form{4} predicates definitions:
\code{P_1},  \code{P_2},
\code{Q_1} and \code{Q_2} where
\code{P_1}{\orderp}\code{P_2} (that is the recursive branch
of predicate definition \code{P_2} contains one and only one occurrence of 
predicate \code{P_1} and \code{P_1} is self-recursive),
\code{Q_1}{\orderp}\code{Q_2} (that is the recursive branch
of predicate definition \code{Q_2} contains one and only one occurrence of 
predicate \code{Q_1}  and \code{Q_1} is self-recursive), \form{\code{P_1}{\not\orderp}\code{Q_1}}, \form{\code{P_1}{\not\orderp}\code{Q_2}},
\form{\code{P_2}{\not\orderp}\code{Q_1}}, and \form{\code{P_2}{\not\orderp}\code{Q_2}}.
For instance, the definitions of these predicates could be as follows.
\[
\begin{array}{l}
 \seppred{\code{pred~P_1}}{r{,} F{,} u} ~{\equiv}~ \emp {\wedge} r{=}F  \\
~\quad \vee ~ \exists {X}{,}sc'.
\sepnode{r}{c_1}{X{,}\anon{,}u{,}\anon{,}\anon} ~{\sep}~ \seppred{\code{P_1}}{{X}{,} F{,} u} {\wedge}
r{\neq} F {\wedge} \a_1
\\
\seppred{\code{pred~P_2}}{r{,} F{,}{B_1}{,} u{,}sc{,}tg} ~{\equiv}~ \emp {\wedge} r{=}F {\wedge} sc{=}tg \\
~\quad \vee ~ \exists {X}{,}d{,}u'{,}sc'.
\sepnode{r}{c_1}{X{,}d{,}u{,}u'{,}sc'} {\sep}\seppred{\code{P_1}}{d{,} B{,} u'}{\sep}~ \seppred{\code{P_2}}{{X}{,} F{,}{B_1}{,} u{,}sc'{,}tg} {\wedge}
r{\neq} F {\wedge} \a_2
\\
\seppred{\code{pred~Q_1}}{r{,} F{,} u} ~{\equiv}~ \emp {\wedge} r{=}F \\
~\quad \vee ~ \exists {X}{,}sc'.
\sepnode{r}{c_2}{X{,}\anon{,}u{,}\anon{,}\anon} ~{\sep}~ \seppred{\code{Q_1}}{{X}{,} F{,} u} {\wedge}
r{\neq} F {\wedge} \a_3
\\
\seppred{\code{pred~Q_2}}{r{,} F{,}{B_1}{,} u{,}sc{,}tg} ~{\equiv}~ \emp {\wedge} r{=}F {\wedge} sc{=}tg \\
~\quad \vee ~ \exists {X}{,}d{,}u'{,}sc'.
\sepnode{r}{c_2}{X{,}d{,}u{,}u'{,}sc'} {\sep}\seppred{\code{Q_1}}{d{,} B{,} u'}{\sep}~ \seppred{\code{Q_2}}{{X}{,} F{,}{B_1}{,} u{,}sc'{,}tg} {\wedge}
r{\neq} F {\wedge} \a_4
\end{array}
\]

We notice that in the definitions of \code{P_2} and \code{Q_2}, we assume that
in the recursive rule
\form{sc'} is a variable of a field of the root points-to predicate.
In general, it may be a parameter of \code{P_2} and \code{Q_2} as well.

If the input entailment is in NF and of the form:
\form{\seppredF{\code{P}}{x{,}F{,}\setvars{B}{,}u{,}sc{,}tg}^0 \sep \D \ent \D'}
and there does not exist an occurrence of inductive predicate \form{\seppredF{\code{Q}}{x{,}F_2{,}\setvars{B}{,}...} \in \D'} then this entailment satisfy the case 2c in Sect. \ref{ent.search.deci}
and is classified as {\invalid} immediately. Thus, the Lemma holds.
In the rest,
to prove this Lemma, we only need to consider the application of rule \rulename{LInd}
where the entailment is in NF and of the form in the conclusion of \rulename{LInd} as:
\[\enode_0{:}~
\entailNCyc{\seppredF{\code{P}}{x{,}F{,}{B}{,}u{,}sc{,}tg}^0{\sep}\heap{\wedge}\pure{\wedge}x{\neq}F{\wedge}x{\neq}F_2}{\seppredF{\code{Q}}{x{,}F_2{,}{B}{,}u{,}sc{,}tg_2} {\sep} \heap'}{\heap}
\]
{\em \bf Furthermore, it is safe to assume that \form{\seppredF{\code{P}}{x{,}F{,}\setvars{B}{,}u{,}sc{,}tg}^0}
is the only one with the smallest unfolding number (i.e., \form{1}) in the LHS of \form{\enode_0} could be applied with rule \rulename{LInd}.}
We prove it by the structural induction on the number of occurrences of inductive
predicates in the LHS of the input entailment.
We do case splits.

\subsection{Case 1: \code{P} and \code{Q} have the same definition.}
We consider two cases where the definition contains nested structures or not.

\subsubsection{Case 1.1} For the simpliest scenario, we assume both definitions of \code{P}
and \code{Q} are self-recursive and do not contain nested structures i.e.,
\form{\code{P} \equiv \code{Q} \equiv \code{P_1}}.
Then, \form{\enode_0} becomes:
\[\enode_{1_1}{:}~
\entailNCyc{\seppredF{\code{P_1}}{x{,}F{,}u}^0{\sep}\heap{\wedge}\pure{\wedge}x{\neq}F{\wedge}x{\neq}F_2}{\seppredF{\code{P_1}}{x{,}F_2{,}u} {\sep} \heap'}{\heap}
\]
After applied with rule
\rulename{LInd},  our system generates a premise as follows.
\[\begin{array}{ll}
\enode_{{1_1}_1}{:}&
\sepnode{x}{c_1}{X{,}d{,}u{,}sc'{,}uprm} {\sep} \seppredF{\code{P_1}}{X{,}F{,}u}^1{\sep}\heap{\wedge}\pure{\wedge}x{\neq}F{\wedge}x{\neq}F_2 {\wedge} \a_{1_1} \\
    & \quad \ent~{\seppredF{\code{P_1}}{x{,}F_2{,}u} {\sep} \heap'}
\end{array}
\]
where \form{X}, \form{d}, \form{sc'} and \form{uprm} are two fresh variables
and \form{\a_{1_1}} is the arithmetical contraint obtained by
subsituting actual/formal paramters into the constrtaint
\form{a_{1}} of the recursive rule of the definition of \code{P_1}.
Next,
entailment \form{\enode_{{1_1}_1}} is normalized by applying rule
\rulename{{\neq}\nil} to
obtain:
\[\begin{array}{ll}
\enode_{{1_1}_2}{:}&
\sepnode{x}{c_1}{X{,}d{,}u{,}sc'{,}uprm} {\sep} \seppredF{\code{P_1}}{X{,}F{,}u}^1{\sep}\heap{\wedge}\pure{\wedge}x{\neq}F{\wedge}x{\neq}F_2 {\wedge} \a_{1_1} {\wedge} x{\neq}\nil\\
    & \quad \ent~{\seppredF{\code{P_1}}{x{,}F_2{,}u} {\sep} \heap'}
\end{array}
\]

Now, \form{\enode_{{1_1}_2}} is applied with rule
\rulename{RInd} to obtain:
\[\begin{array}{ll}
\enode_{{1_1}_3}{:}&
\sepnode{x}{c_1}{X{,}d{,}u{,}sc'{,}uprm} {\sep} \seppredF{\code{P_1}}{X{,}F{,}u}^1{\sep}\heap{\wedge}\pure{\wedge}x{\neq}F{\wedge}x{\neq}F_2 {\wedge} \a_{1_1}{\wedge} x{\neq}\nil\\
        & \ent~ \sepnode{x}{c_1}{X{,}d{,}u{,}sc'{,}uprm} {\sep} {\seppredF{\code{P_1}}{x{,}F_2{,}u} {\sep} \heap' {\wedge} \a_{1_1}}
\end{array}
\]
We note that for completeness applications of rule \rulename{\sep} are
always performed after all other rules. As so,
next, \rulename{Hypothesis} is applied to eliminate the arithmetical constraint in the RHS
to obtain:
\[\begin{array}{ll}
\enode_{{1_1}_4}{:}&
\sepnode{x}{c_1}{X{,}d{,}u{,}sc'{,}uprm} {\sep} \seppredF{\code{P_1}}{X{,}F{,}u}^1{\sep}\heap{\wedge}\pure{\wedge}x{\neq}F{\wedge}x{\neq}F_2 {\wedge} \a_{1_1}\\
        & \quad \ent~ \sepnode{x}{c_1}{X{,}d{,}u{,}sc'{,}uprm} {\sep} {\seppredF{\code{P_1}}{x{,}F_2{,}u} {\sep} \heap'}
\end{array}
\]
Our system now applies rules \rulename{ExM}
and \rulename{{{\neq}\sep}} to normalize the LHS where application of the latter rule generates two premises.
\[\begin{array}{ll}
\enode_{{1_1}_5}{:}&
\sepnode{x}{c_1}{X{,}d{,}u{,}sc'{,}uprm} {\sep} \seppredF{\code{P_1}}{X{,}F{,}u}^1{\sep}\heap{\wedge}\pure{\wedge}x{\neq}F{\wedge}x{\neq}F_2 {\wedge} \a_{1_1} {\wedge} x{\neq}\nil {\wedge} X{=}F \\
& \quad \ent~ \sepnode{x}{c_1}{X{,}d{,}u{,}sc'{,}uprm} {\sep} {\seppredF{\code{P_1}}{x{,}F_2{,}u} {\sep} \heap'} \\
\enode_{{1_1}_6}{:}&
\sepnode{x}{c_1}{X{,}d{,}u{,}sc'{,}uprm} {\sep} \seppredF{\code{P_1}}{X{,}F{,}u}^1{\sep}\heap{\wedge}\pure{\wedge}x{\neq}F{\wedge}x{\neq}F_2 {\wedge} \a_{1_1} {\wedge} x{\neq}\nil {\wedge} X{\neq}F  \\
        & \quad {\wedge} x{\neq}X~ \ent~ \sepnode{x}{c_1}{X{,}d{,}u{,}sc'{,}uprm} {\sep} {\seppredF{\code{P_1}}{x{,}F_2{,}u} {\sep} \heap'}
\end{array}
\]
\begin{enumerate}
\item For the premise \form{\enode_{{1_1}_5}}, our system applies rules \rulename{Subst}
  and \rulename{{L}=}
to eliminate \form{X{=}F} and obtain:
\[\begin{array}{ll}
\enode_{{1_1}_{5_1}}{:}&
\sepnode{x}{c_1}{F{,}d{,}u{,}sc'{,}uprm} {\sep} \seppredF{\code{P_1}}{F{,}F{,}u}^1{\sep}\heap{\wedge}\pure{\wedge}x{\neq}F{\wedge}x{\neq}F_2 {\wedge} \a_{1_1} {\wedge} x{\neq}\nil \\
& \quad \ent~ \sepnode{x}{c_1}{F{,}d{,}u{,}sc'{,}uprm} {\sep} {\seppredF{\code{P_1}}{x{,}F_2{,}u} {\sep} \heap'} \\
\end{array}
\]
Next, rule \rulename{LBase} is applied to discard the inductive predicate in
the LHS and obtain the following premise:
\[\begin{array}{ll}
\enode_{{1_1}_{5_2}}{:}&
\sepnode{x}{c_1}{F{,}d{,}u{,}sc'{,}uprm} {\sep}\heap{\wedge}\pure{\wedge}x{\neq}F{\wedge}x{\neq}F_2 {\wedge} \a_{1_1} {\wedge} x{\neq}\nil \\
& \quad \ent~ \sepnode{x}{c_1}{X{,}d{,}u{,}sc'{,}uprm} {\sep} {\seppredF{\code{P_1}}{x{,}F_2{,}u} {\sep} \heap'} \\
\end{array}
\]
As the number of inductive predicates in the LHS of \form{\enode_{{1_1}_{5_2}}} is reduced, 
by induction, this Lemma holds.
\item \label{proof.term.two.unfold} For the premise \form{\enode_{{1_1}_6}}, we have two cases.
\begin{enumerate}
\item If \form{\FV(\pure) \cap \FV(\a_{1_1}) = \emptyset}. Our system links  \form{\enode_{{1_1}_6}} back to
  \form{\enode_{1_1}} as follows. First, it weakens (a.k.a discards) two matched
  points-to predicates in the two sides and the following pure constraints
  in LHS: \form{x{\neq}F}, \form{\a_{1_1}}, \form{x{\neq}\nil}, and \form{x{\neq}X}.
  After that, it substitutes the remaining entailment with \form{\sub{=}\{x/X\}}
  to obtain the identical entailment with \form{\enode_{1_1}}. We notice that
  as \form{X} is a fresh variable, it does not apprear in \form{\heap {\wedge}\pure} and \form{\heap'}. Then, the Lemma holds for this case.
\item \form{\FV(\pure) \cap \FV(\a_{1_1}) \neq \emptyset}. As
the substitution \form{[sc/sc']} could not be applied, our system
could not link  \form{\enode_{{1_1}_6}} back to
  \form{\enode_{1_1}}. It applies the same the proof search as applied
for \form{\enode_{1_1}} to unfold  \form{\enode_{{1_1}_6'}}. As this time,
  \form{\enode_{{1_1}_6'}} contains respective \form{\a_{1_1}'} and \form{sc''} and where \form{\a_{1_1} = \a_{1_1}'[sc'/sc'']}.
Now, our system could link \form{\enode_{{1_1}_6'}} back to \form{\enode_{{1_1}_6}}.
\end{enumerate}
  \end{enumerate}

\subsubsection{Case 1.2} For a more general case, we assume \form{\code{P} \equiv \code{Q} \equiv \code{P_2}}.
Then, \form{\enode_0} becomes:
\[\enode_{1_2}{:}~
\entailNCyc{\seppredF{\code{P_2}}{x{,}F{,}{B}{,}u{,}sc{,}tg}^0{\sep}\heap{\wedge}\pure{\wedge}x{\neq}F{\wedge}x{\neq}F_2}{\seppredF{\code{P_2}}{x{,}F_2{,}{B}{,}u{,}sc{,}tg_2} {\sep} \heap'}{\heap}
\]
The first four steps are similar to {\bf Case 1.1}. After applied with rule
\rulename{LInd},  our system generates a premise as follows.
\[\begin{array}{ll}
\enode_{{1_2}_1}{:}&
\sepnode{x}{c_1}{X{,}d{,}u{,}u'{,}sc'} {\sep} \seppredF{\code{P_1}}{d{,}B{,}u'}^1{\sep} \seppredF{\code{P_2}}{X{,}F{,}{B}{,}u{,}sc'{,}tg}^1{\sep}\heap{\wedge}\pure{\wedge}x{\neq}F{\wedge}x{\neq}F_2 {\wedge} \a_{2_1}\\
        & \quad \ent~{\seppredF{\code{P_2}}{x{,}F_2{,}{B}{,}u{,}sc{,}tg_2} {\sep} \heap'}
\end{array}
\]
where \form{\a_{2_1}} is obtaied by substituting actual/formal paramters
into the arithmetical contraint \form{a_{2}} of the recursive rule of the definition of \code{P_2}.
Next,
this entailment is normalized by applying rule \rulename{{\neq}\nil} to obtain:
\[\begin{array}{ll}
\enode_{{1_2}_2}{:}&
\sepnode{x}{c_1}{X{,}d{,}u{,}u'{,}sc'} {\sep} \seppredF{\code{P_1}}{d{,}B{,}u'}^1{\sep} \seppredF{\code{P_2}}{X{,}F{,}{B}{,}u{,}sc'{,}tg}^1{\sep}\heap{\wedge}\pure{\wedge}x{\neq}F{\wedge}x{\neq}F_2 {\wedge} \a_{2_1}\\
        & \quad  {\wedge} x{\neq}\nil~ \ent~{\seppredF{\code{P_2}}{x{,}F_2{,}{B}{,}u{,}sc{,}tg_2} {\sep} \heap'}
\end{array}
\]
Next, \form{\enode_{{1_2}_2}} is
applied with rule \rulename{RInd} to obtain:
\[\begin{array}{ll}
\enode_{{1_2}_3}{:}&
\sepnode{x}{c_1}{X{,}d{,}u{,}u'{,}sc'} {\sep} \seppredF{\code{P_1}}{d{,}B{,}u'}^1{\sep} \seppredF{\code{P_2}}{X{,}F{,}{B}{,}u{,}sc'{,}tg}^1{\sep}\heap{\wedge}\pure{\wedge}x{\neq}F{\wedge}x{\neq}F_2{\wedge} \a_{2_1} \\
        & \quad  {\wedge} x{\neq}\nil~ \ent~ \sepnode{x}{c_1}{X{,}d{,}u{,}u'{,}sc'} {\sep} \seppredF{\code{P_1}}{d{,}B{,}u'}^0{\sep}{\seppredF{\code{P_2}}{X{,}F_2{,}{B}{,}u{,}sc'{,}tg_2} {\sep} \heap' {\wedge} \a_{2_1}}
\end{array}
\]
We note that rule \rulename{\sep} is always applied after all other rules. As so,
next, \rulename{Hypothesis} is applied to eliminate the arithmetical constraint in the RHS
to obtain:
\[\begin{array}{ll}
\enode_{{1_2}_4}{:}&
\sepnode{x}{c_1}{X{,}d{,}u{,}u'{,}sc'} {\sep} \seppredF{\code{P_1}}{d{,}B{,}u'}^1{\sep} \seppredF{\code{P_2}}{X{,}F{,}{B}{,}u{,}sc'{,}tg}^1{\sep}\heap{\wedge}\pure{\wedge}x{\neq}F{\wedge}x{\neq}F_2 {\wedge} \a_{2_1}\\
        & \quad  {\wedge} x{\neq}\nil~ \ent~ \sepnode{x}{c_1}{X{,}d{,}u{,}u'{,}sc'} {\sep} \seppredF{\code{P_1}}{d{,}B{,}u'}^0{\sep}{\seppredF{\code{P_2}}{X{,}F_2{,}{B}{,}u{,}sc'{,}tg_2} {\sep} \heap' }
\end{array}
\]
Our system now applies rules  \rulename{ExM} and
\rulename{{\neq}{\sep}} to normalize the LHS. Particularly,  applying rule
\rulename{ExM} for \form{d} and \form{B} generates two premises:
\[\begin{array}{ll}
\enode_{{1_2}_5}{:}&
\sepnode{x}{c_1}{X{,}d{,}u{,}u'{,}sc'} {\sep} \seppredF{\code{P_1}}{d{,}B{,}u'}^1{\sep} \seppredF{\code{P_2}}{X{,}F{,}{B}{,}u{,}sc'{,}tg}^1{\sep}\heap{\wedge}\pure{\wedge}x{\neq}F{\wedge}x{\neq}F_2 {\wedge} \a_{2_1}\\
& \quad  {\wedge} x{\neq}\nil{\wedge}d{=}B~ \ent~ \sepnode{x}{c_1}{X{,}d{,}u{,}u'{,}sc'} {\sep} \seppredF{\code{P_1}}{d{,}B{,}u'}^0{\sep}{\seppredF{\code{P_2}}{X{,}F_2{,}{B}{,}u{,}sc'{,}tg_2} {\sep} \heap' } \\
\enode_{{1_2}_6}{:}&
\sepnode{x}{c_1}{X{,}d{,}u{,}u'{,}sc'} {\sep} \seppredF{\code{P_1}}{d{,}B{,}u'}^1{\sep} \seppredF{\code{P_2}}{X{,}F{,}{B}{,}u{,}sc'{,}tg}^1{\sep}\heap{\wedge}\pure{\wedge}x{\neq}F{\wedge}x{\neq}F_2 {\wedge} \a_{2_1}\\
& \quad  {\wedge} x{\neq}\nil{\wedge}d{\neq}B~ \ent~ \sepnode{x}{c_1}{X{,}d{,}u{,}u'{,}sc'} {\sep} \seppredF{\code{P_1}}{d{,}B{,}u'}^0{\sep}{\seppredF{\code{P_2}}{X{,}F_2{,}{B}{,}u{,}sc'{,}tg_2} {\sep} \heap' } \\
\end{array}
\]
\begin{enumerate}
\item For the first premise \form{\enode_{{1_2}_5}}, our system first applies rules \rulename{Subst} to obtain and \rulename{{L}=}:
  \[\begin{array}{ll}
\enode_{{1_2}_{5_1}}{:}&
\sepnode{x}{c_1}{X{,}B{,}u{,}u'{,}sc'} {\sep} \seppredF{\code{P_1}}{B{,}B{,}u'}^1{\sep} \seppredF{\code{P_2}}{X{,}F{,}{B}{,}u{,}sc'{,}tg}^1{\sep}\heap{\wedge}\pure{\wedge}x{\neq}F{\wedge}x{\neq}F_2 {\wedge} \a_{2_1}\\
& \quad  {\wedge} x{\neq}\nil~ \ent~ \sepnode{x}{c_1}{X{,}B{,}u{,}u'{,}sc'} {\sep} \seppredF{\code{P_1}}{B{,}B{,}u'}^0{\sep}{\seppredF{\code{P_2}}{X{,}F_2{,}{B}{,}u{,}sc'{,}tg_2} {\sep} \heap' }
\end{array}
  \]
  After that, it applies rules \rulename{LBase} and
 \rulename{RBase} to eliminate
  inductive predicates \code{P_1} in the LHS and RHS, respectively. Afterward, the premise is
  obtained as:
   \[\begin{array}{ll}
\enode_{{1_2}_{5_3}}{:}&
\sepnode{x}{c_1}{X{,}B{,}u{,}u'{,}sc'} {\sep} \seppredF{\code{P_2}}{X{,}F{,}{B}{,}u{,}sc'{,}tg}^1{\sep}\heap{\wedge}\pure{\wedge}x{\neq}F{\wedge}x{\neq}F_2 {\wedge} \a_{2_1}\\
& \quad  {\wedge} x{\neq}\nil~ \ent~ \sepnode{x}{c_1}{X{,}B{,}u{,}u'{,}sc'} {\sep}{\seppredF{\code{P_2}}{X{,}F_2{,}{B}{,}u{,}sc'{,}tg_2} {\sep} \heap' }
\end{array}
  \]
  Now, it generates a back-link between \form{\enode_{{1_2}_{5_3}}} and \form{\enode_{{1_2}}}.
  Hence, the Lamma holds.
\item For the second premise \form{\enode_{{1_2}_6}}, the system applies rule
  \rulename{{\neq}\sep} and then rule \rulename{ExM} to obtain two following
  premises:
  \[\begin{array}{ll}
\enode_{{1_2}_7}{:}&
\sepnode{x}{c_1}{X{,}d{,}u{,}u'{,}sc'} {\sep} \seppredF{\code{P_1}}{d{,}B{,}u'}^1{\sep} \seppredF{\code{P_2}}{X{,}F{,}{B}{,}u{,}sc'{,}tg}^1{\sep}\heap{\wedge}\pure{\wedge}x{\neq}F \\
& ~  {\wedge}x{\neq}F_2 {\wedge} \a_{2_1} {\wedge} x{\neq}\nil{\wedge}d{\neq}B {\wedge}x{\neq}d {\wedge}X{=}F\\
& \quad \ent~ \sepnode{x}{c_1}{X{,}d{,}u{,}u'{,}sc'} {\sep} \seppredF{\code{P_1}}{d{,}B{,}u'}{\sep}{\seppredF{\code{P_2}}{X{,}F_2{,}{B}{,}u{,}sc'{,}tg_2} {\sep} \heap' } \\
\enode_{{1_2}_8}{:}&
\sepnode{x}{c_1}{X{,}d{,}u{,}u'{,}sc'} {\sep} \seppredF{\code{P_1}}{d{,}B{,}u'}^1{\sep} \seppredF{\code{P_2}}{X{,}F{,}{B}{,}u{,}sc'{,}tg}^1{\sep}\heap{\wedge}\pure{\wedge}x{\neq}F \\
& ~  {\wedge}x{\neq}F_2{\wedge} \a_{2_1} {\wedge} x{\neq}\nil{\wedge}d{\neq}B {\wedge}x{\neq}d {\wedge}X{\neq}F \\
& \quad \ent~ \sepnode{x}{c_1}{X{,}d{,}u{,}u'{,}sc'} {\sep} \seppredF{\code{P_1}}{d{,}B{,}u'}{\sep}{\seppredF{\code{P_2}}{X{,}F_2{,}{B}{,}u{,}sc'{,}tg_2} {\sep} \heap' } 
\end{array}
  \]
  \begin{enumerate}
  \item For the premise \form{\enode_{{1_2}_7}}, our system applies rules \rulename{Subst}
    and \rulename{{L}=}
    first and then  rule \rulename{LBase} to eliminate
  inductive predicates \code{P_2} in the LHS. Afterward, the premise is
  obtained as:
   \[\begin{array}{ll}
\enode_{{1_2}_{7_3}}{:}&
\sepnode{x}{c_1}{F{,}d{,}u{,}u'{,}sc'} {\sep} \seppredF{\code{P_1}}{d{,}B{,}u'}^1{\sep}\heap{\wedge}\pure{\wedge}x{\neq}F \\
& ~  {\wedge}x{\neq}F_2 {\wedge} \a_{2_1} {\wedge} x{\neq}\nil{\wedge}d{\neq}B {\wedge}x{\neq}d \\
& \quad \ent~ \sepnode{x}{c_1}{F{,}d{,}u{,}u'{,}sc'} {\sep} \seppredF{\code{P_1}}{d{,}B{,}u'}{\sep}{\seppredF{\code{P_2}}{F{,}F_2{,}{B}{,}u{,}sc'{,}tg_2} {\sep} \heap' } \\
\end{array}
   \]
    Similarly to {\bf Case 1.1}, the Lemma holds for \form{\enode_{{1_2}_{7_3}}}.
  \item For the premise \form{\enode_{{1_2}_8}}, our system  processes similarly to 
 \ref{proof.term.two.unfold} in {\bf Case 1.1}: the predicate in the LHS is unfolded at most
two times. Hence, the Lemma holds.
   \end{enumerate}
\end{enumerate}

\subsection{Case 2: \code{P} and \code{Q} have different definitions and they are syntactically dependent.}
We consider two sub-cases. In the first case, we assume \form{\code{P}{\orderp}\code{Q}}.
In the second case, we assume \form{\code{Q}{\orderp}\code{P}}.

\subsubsection{Case 2.1: \form{\code{P}{\orderp}\code{Q}}}
For a general case, we assume \form{\code{P} \equiv \code{P_1'}} and \form{\code{Q} \equiv \code{P_2'}}
where \code{P_1'} is defined similarly to \code{P_1} except it contains an additional local property (Otherwise, the proof for \form{\code{P} \equiv \code{P_1}} and \form{\code{Q} \equiv \code{P_2}} is quite trivial.).
\[
\begin{array}{l}
 \seppred{\code{pred~P'_1}}{r{,} F{,} u{,}sc{,}tg} ~{\equiv}~ \emp {\wedge} r{=}F {\wedge} sc{=}tg  \\
~\qquad \vee ~ \exists {X}{,}sc'.
\sepnode{r}{c_1}{X{,}\anon{,}u{,}\anon{,}sc'} ~{\sep}~ \seppred{\code{P'_1}}{{X}{,} F{,} u{,}sc'{,}tg} {\wedge}
r{\neq} F {\wedge} \a_1
\\
\seppred{\code{pred~P'_2}}{r{,} F{,}{B_1}{,} u{,}sc{,}tg} ~{\equiv}~ \emp {\wedge} r{=}F {\wedge} sc{=}tg \\
~\qquad \vee ~ \exists {X}{,}d{,}u'{,}sc'.
\sepnode{r}{c_1}{X{,}d{,}u{,}u'{,}sc'} {\sep}\seppred{\code{P'_1}}{d{,} B{,} u'{,}sc{,}tg}{\sep}~ \seppred{\code{P'_2}}{{X}{,} F{,}{B_1}{,} u{,}sc'{,}tg} {\wedge}
r{\neq} F {\wedge} \a_2
\end{array}
\]
Then, \form{\enode_0} becomes:
\[\enode_{2_1}{:}~
\entailNCyc{\seppredF{\code{P'_1}}{x{,}F{,}{B}{,}u{,}sc{,}tg}^0{\sep}\heap{\wedge}\pure{\wedge}x{\neq}F{\wedge}x{\neq}F_2}{\seppredF{\code{P'_2}}{x{,}F_2{,}{B}{,}u{,}sc{,}tg_2} {\sep} \heap'}{\heap}
\]
After applying three rules \rulename{LInd},
 \rulename{{\neq}\nil}
and \rulename{RInd} in sequence (and similarly to {\bf Case 1.1} and {\bf Case 1.2} above),
our system generates the following premise.
\[
\begin{array}{ll}
  \enode_{{2_1}_3}{:}&
\sepnode{x}{c_1}{X{,}d{,}u{,}u'{,}sc'} {\sep} \seppredF{\code{P'_1}}{X{,}F{,}{B}{,}u{,}sc'{,}tg}^1{\sep}\heap{\wedge}\pure{\wedge}x{\neq}F{\wedge}x{\neq}F_2 {\wedge}\a_{1_1}\\
&       \quad \ent~{\sepnode{x}{c_1}{X{,}d{,}u{,}u'{,}sc'} {\sep}\seppred{\code{P'_1}}{d{,} B{,} u'{,}sc{,}tg} {\sep} \seppredF{\code{P_2}}{X{,}F_2{,}{B}{,}u{,}sc'{,}tg_2} {\sep} \heap'{\wedge}\a_{2_1}}
\end{array}
\]
where \form{X}, \form{d}, \form{sc'} and \form{u'} are two fresh variables.
Our system applies rule \rulename{ExM} for \form{d} and \form{B} to generate
the following two premises.
\[
\begin{array}{ll}
  \enode_{{2_1}_4}{:}&
\sepnode{x}{c_1}{X{,}d{,}u{,}u'{,}sc'} {\sep} \seppredF{\code{P'_1}}{X{,}F{,}{B}{,}u{,}sc'{,}tg}^1{\sep}\heap{\wedge}\pure{\wedge}x{\neq}F{\wedge}x{\neq}F_2 {\wedge}d{=}B {\wedge}\a_{1_1}\\
&       \quad \ent~{\sepnode{x}{c_1}{X{,}d{,}u{,}u'{,}sc'} {\sep}\seppred{\code{P'_1}}{d{,} B{,} u'{,}sc{,}tg} {\sep} \seppredF{\code{P_2}}{X{,}F_2{,}{B}{,}u{,}sc'{,}tg_2} {\sep} \heap'{\wedge}\a_{2_1}}\\
\enode_{{2_1}_5}{:}&
\sepnode{x}{c_1}{X{,}d{,}u{,}u'{,}sc'} {\sep} \seppredF{\code{P'_1}}{X{,}F{,}{B}{,}u{,}sc'{,}tg}^1{\sep}\heap{\wedge}\pure{\wedge}x{\neq}F{\wedge}x{\neq}F_2 {\wedge}d{\neq}B {\wedge}\a_{1_1}\\
&       \quad \ent~{\sepnode{x}{c_1}{X{,}d{,}u{,}u'{,}sc'} {\sep}\seppred{\code{P'_1}}{d{,} B{,} u'{,}sc{,}tg} {\sep} \seppredF{\code{P_2}}{X{,}F_2{,}{B}{,}u{,}sc'{,}tg_2} {\sep} \heap'{\wedge}\a_{2_1}}
\end{array}
\]
As \form{d} is a fresh variable, the predicate \form{\seppred{\code{P'_1}}{d{,} B{,} u'{,}sc{,}tg}}
does not appear in the LHS of \form{\enode_{{2_1}_5}}. Hence,  \form{\enode_{{2_1}_5}}
is classified as {\invalid} and {\allEnt} returns {\invalid}. Hence, the Lemma holds.

\subsubsection{Case 2.2: \form{\code{Q}{\orderp}\code{P}}}
For a general case, we assume \form{\code{P} \equiv\code{P_2}} and \form{\code{Q} \equiv \code{P_1}}.
Then, \form{\enode_0} becomes:
\[\enode_{2_2}{:}~
\entailNCyc{\seppredF{\code{P_2}}{x{,}F{,}{B}{,}u{,}sc{,}tg}^0{\sep}\heap{\wedge}\pure{\wedge}x{\neq}F{\wedge}x{\neq}F_2}{\seppredF{\code{P_1}}{x{,}F_2{,}{B}{,}u} {\sep} \heap'}{\heap}
\]
The proof for this case is similar to {\bf Case 2.1}. After applying three rules \rulename{LInd}, \rulename{{\neq}\nil}
and \rulename{RInd} in sequence,
our system generates the following premise.
\[
\begin{array}{ll}
  \enode_{{2_2}_3}{:}&
  \sepnode{x}{c_1}{X{,}d{,}u{,}u'{,}sc'} {\sep}\seppred{\code{P_1}}{d{,} B{,} u'}  {\sep} 
  \seppredF{\code{P_2}}{X{,}F{,}{B}{,}u{,}sc'{,}tg}^0{\sep}\heap{\wedge}\pure{\wedge}x{\neq}F{\wedge}x{\neq}F_2 {\wedge}\a_{2_1} \\
  &\quad \ent~ \sepnode{x}{c_1}{X{,}d{,}u{,}u'{,}sc'} {\sep} \seppredF{\code{P_1}}{X{,}F_2{,}{B}{,}u} {\sep} \heap' {\wedge}\a_{1_1}
\end{array}
\]
where \form{X}, \form{d}, \form{sc'} and \form{u'} are two fresh variables.
Our system applies rule \rulename{ExM} for \form{d} and \form{B} to generate
the following two premises.
\[
\begin{array}{ll}
 \enode_{{2_2}_4}{:}&
  \sepnode{x}{c_1}{X{,}d{,}u{,}u'{,}sc'} {\sep}\seppred{\code{P_1}}{d{,} B{,} u'}  {\sep} 
  \seppredF{\code{P_2}}{X{,}F{,}{B}{,}u{,}sc'{,}tg}^0{\sep}\heap{\wedge}\pure{\wedge}x{\neq}F{\wedge}x{\neq}F_2 {\wedge}\a_{2_1} {\wedge}d{=}B\\
  &\quad \ent~ \sepnode{x}{c_1}{X{,}d{,}u{,}u'{,}sc'} {\sep} \seppredF{\code{P_1}}{X{,}F_2{,}{B}{,}u} {\sep} \heap' {\wedge}\a_{1_1} \\
  \enode_{{2_2}_5}{:}&
  \sepnode{x}{c_1}{X{,}d{,}u{,}u'{,}sc'} {\sep}\seppred{\code{P_1}}{d{,} B{,} u'}  {\sep} 
  \seppredF{\code{P_2}}{X{,}F{,}{B}{,}u{,}sc'{,}tg}^0{\sep}\heap{\wedge}\pure{\wedge}x{\neq}F{\wedge}x{\neq}F_2 {\wedge}\a_{2_1} {\wedge}d{\neq}B\\
  &\quad \ent~ \sepnode{x}{c_1}{X{,}d{,}u{,}u'{,}sc'} {\sep} \seppredF{\code{P_1}}{X{,}F_2{,}{B}{,}u} {\sep} \heap' {\wedge}\a_{1_1} \\
\end{array}
\]
As \form{d} is a fresh variable, the predicate \form{\seppred{\code{P_1}}{d{,} B{,} u'}}
does not appear in the RHS of \form{\enode_{{2_2}_5}}. Hence,  \form{\enode_{{2_2}_5}}
is classified as {\invalid} and {\allEnt} returns {\invalid}. Hence, the Lemma holds.

\subsection{Case 3: \code{P} and \code{Q} have different definitions and they are syntactically independent.}
We consider three  sub-cases based on the positions of inductive predicates
in the dependency hierarchies. In the first case, we assume \code{P} is ``bigger" than \code{Q}.
In the second case, we assume \code{P} is ``smaller" than \code{Q}.
And in the last case, we assume \code{P} is ``equal" to \code{Q}.

\subsubsection{Case 3.1:} We assume \form{\code{P} \equiv \code{P_2}} and \form{\code{Q} \equiv \code{Q_1}}.
Then, \form{\enode_0} becomes:
\[\enode_{3_1}{:}~
\entailNCyc{\seppredF{\code{P_2}}{x{,}F{,}{B}{,}u{,}sc{,}tg}^0{\sep}\heap{\wedge}\pure{\wedge}x{\neq}F{\wedge}x{\neq}F_2}{\seppredF{\code{Q_1}}{x{,}F_2{,}{B}{,}u} {\sep} \heap'}{\heap}
\]
If \form{c_1 {\not}{\equiv} c_2}, \code{is\_closed} returns
{\invalid} (Case 2d in Sect. \ref{ent.search.deci}).
Otherwise, the proof for this case is similar to {\bf Case 2.2}.

\subsubsection{Case 3.2:} We assume \form{\code{P} \equiv \code{P_1'}} and \form{\code{Q} \equiv \code{Q_2}}. Then, \form{\enode_0} becomes:
\[\enode_{3_2}{:}~
\entailNCyc{\seppredF{\code{P_1'}}{x{,}F{,}{B}{,}u{,}sc{,}tg}^0{\sep}\heap{\wedge}\pure{\wedge}x{\neq}F{\wedge}x{\neq}F_2}{\seppredF{\code{Q_2}}{x{,}F_2{,}{B}{,}u{,}sc{,}tg_2} {\sep} \heap'}{\heap}
\]
If \form{c_1 {\not}{\equiv} c_2}, the proof is straightforward.
Otherwise, the proof for this case is similar to {\bf Case 2.1}.

\subsubsection{Case 3.3:} We assume \form{\code{P} \equiv \code{P_1}} and \form{\code{Q} \equiv \code{Q_1}}. Then, \form{\enode_0} becomes:
\[\enode_{3_3}{:}~
\entailNCyc{\seppredF{\code{P_2}}{x{,}F{,}{B}{,}u}^0{\sep}\heap{\wedge}\pure{\wedge}x{\neq}F{\wedge}x{\neq}F_2}{\seppredF{\code{Q_2}}{x{,}F_2{,}{B}{,}u} {\sep} \heap'}{\heap}
\]
If \form{c_1 {\not}{\equiv} c_2}, the proof is straightforward.
Otherwise, the proof for this case is similar to {\bf Case 1.1}.

\subsubsection{Case 3.4:} We assume \form{\code{P} \equiv \code{P_2}} and \form{\code{Q} \equiv \code{Q_2}}. Then, \form{\enode_0} becomes:
\[\enode_{3_4}{:}~
\entailNCyc{\seppredF{\code{P_2}}{x{,}F{,}{B}{,}u{,}sc{,}tg}^0{\sep}\heap{\wedge}\pure{\wedge}x{\neq}F{\wedge}x{\neq}F_2}{\seppredF{\code{Q_2}}{x{,}F_2{,}{B}{,}u{,}sc{,}tg_2} {\sep} \heap'}{\heap}
\]
If \form{c_1 {\not}{\equiv} c_2}, the proof is straightforward.
Otherwise, the proof for this case is similar to {\bf Case 1.2}.

\hfill \qed.

\section{Complexity Analysis - Proposition \ref{prop.complex}}\label{app.thm.complex}

Suppose that n is the maximum number of predicates
(both inductive predicates and points-to predicates) among
the LHS of the input entailment and those definitions in \form{\PName}, and
\form{m} is the maximum number of fields of data structures.
Then, the complexity is defined as follows.

First, we analyze the number of computation when all inductive predicates
in the LHS are unfolded at most once.
Let \form{P(n,m)} be the time complexity function under this assumption.
Each pair of the root and segment parameters, say \form{r} and \form{F},
 of an inductive predicate
is applied with rule \rulename{ExM} at most one.
For the first premise where \form{r=F},
after applied with \rulename{LBase} the number of inductive predicates
is \form{n-1} and its running time is \form{P(n-1,m)}.

For the second premise, say \form{\enode_c}, where \form{r\neq F},
after applied with \rulename{LInd}, normalization rules
and \rulename{RInd}, it is applied with \form{\sep}
to create two premises. While one of them is linked back
to \form{\enode_c} the second is of the form: \form{\heap' \ent \heap''}
where \form{\heap'} (respective \form{\heap''})
is the matrix heap of the unfolded predicate in the LHS (respective RHS).
\rulename{ExM} is applied at most
 \form{(n-1) + (n-2) + ..+ 1 = \mathcal{O}(n^2)} times.
Moreover, as (i) all roots of inductive predicates in a matrix heap must
not be aliasing (ensured by the normalization rule \rulename{ExM})
and (ii) they are must be in the fields of the root points-to predicate
of the recursive definition rule,
the number of inductive predicates in both
\form{\heap'} and \form{\heap''} must be less than \form{m}.
Suppose that the running time of
such an entailment of \form{m} inductive predicates
of matrix heap is \form{q(m)},
then \form{P(n,m) = P(n-1,m) + q(m) + \mathcal{O}(n^2)}.

\[
\begin{array}{lcl}
P(n,m) &=& P(n-1,m) + q(m) + \mathcal{O}(n^2) \\
  &=& P(n-2,m) + 2q(m) + 2\mathcal{O}(n^2) \\
  &= & .... \\
  &=& P(1, m) + (n-1) \times q(m) + (n-1)\mathcal{O}(n^2) \\
  &=& 1 + n\times q(m) + \mathcal{O}(n^3)
   \end{array}
\]
(We presume that the running time of entailment without any inductive predicates is \form{1}.)

 We remark that
if a formula contains two inductive predicates
which has the same root parameters i.e., \form{\seppred{\code{P}}{r{,}F_1...} \sep \seppred{\code{Q}}{r{,}F_2...} \sep \D},
then at least one of them must be reduced into base case
with the empty heap.
As \form{m} is the maximum number of fields of data structures
and the roots parameters of \form{\heap'} must be
one of these variables of the fields,
the number of inductive predicates of
the LHS of any entailment that is derived from
\form{\heap' \ent \heap''}, 
is less than or equal to \form{m}.
Thus, under modular substitution the number of combination of such
 \form{m} inductive predicates is \form{\mathcal{O}(2^m)}.

Therefore, \form{P(n,m) = \mathcal{O}(n \times 2^m + n^3)}.

The unfolding is depth-first and
the steps for the second unfolding are similar.
As the proof is linear, then the number of computation
when all inductive predicates are unfolded at most two
times is at most as \form{2 \times P(n,m) = \mathcal{O}(n \times 2^m + n^3)}.

\section{Completeness of proof rules - Lemma \ref{rule.complete}}\label{app.local.complete}

The completeness of all rules except rule \rulename{\sep} is straightforward.
In the following, we prove the completeness of rule
\rulename{\sep}. 
 The proof is based on the following auxiliary Lemma.

\begin{lemma}\label{lemma.nf.subst}
If \form{\heap \wedge \phi \wedge \a} is in NF and \form{x{\neq}E {\not}{\in} \phi},
then \form{(\heap \wedge \phi) \subst{E}{x} \wedge \a} is in NF.
\end{lemma}
\begin{proof}
All but the fifth clause in the definition \ref{defn.nf} are invariant under substitution.
Moreover, \form{x{\neq}E {\not}{\in} \phi} exclude the violation of the fifth clause
under substitution as well
\hfill \qed.
\end{proof}

First, we provide proofs for pure part when
pure contraints in LHS
does not imply pure contraints in RHS.
\begin{proposition}\label{complete.pure}
If
\form{\entailNCyc{\heap \wedge\phi \wedge \a}{\heap_m \wedge \phi' \wedge \a'}{\heap'}}
is in NF and \form{\enode'{:}~\entailNCyc{\emp\wedge\phi \wedge \a}{\emp \wedge \phi' \wedge \a'}{\heap_m}} is not derivable, then \form{\entailNCyc{\heap \wedge \phi \wedge \a}{\heap' \wedge \phi' \wedge \a'}{\heap_m}} is invalid.
\end{proposition}
\begin{proof}
We show that there is a model of the LHS that satisfies either \form{\neg\phi'}
or \form{\neg\a'} holds.
We proceed cases for each predicate in the RHS.
\begin{enumerate}
\item Case \form{\phi' \equiv E_1{=}E_2}. As the LHS is in NF, any bad model of
\form{\basecomp{\heap_m} \wedge\phi \wedge \a} implies that \form{E_1{\neq}E_2}.
In other words, \form{\basecomp{\heap_m} \wedge\phi \wedge \a} implies that \form{\neg E_1{=}E_2}.
\item Case \form{\phi' \equiv E_1{\neq}E_2}. As \form{\enode'} is not derivable, then the side condition of rule \rulename{Hypothesis}
does not hold. This means \form{E_1{\neq}E_2 {\not}{\in} \phi}.

We note that if \form{{\heap_m} \wedge\phi \wedge \a} is in NF and
\form{E_1{\neq}E_2 {\not}{\in} \phi}, then \form{({\heap_m} \wedge\phi \wedge \a)[E_1/E_2]}
is also in NF (assuming that \form{E_1} is a variable - Lemma \ref{lemma.nf.subst}). Then suppose \form{\sstack, \sheaps}
be a bad model of \form{(\basecomp{\heap_m} \wedge\phi \wedge \a)[E_1/E_2]}, then
\form{\sstack[E_1 {\pto} \sstack(E_2)], \sheaps} is a model of \form{(\basecomp{\heap_m} \wedge\phi \wedge \a)}. \form{\sstack[E_1 {\pto} \sstack(E_2)], \sheaps}
implies that \form{\neg E_1{\neq}E_2}. Therefore, \form{(\basecomp{\heap_m} \wedge\phi \wedge \a)}
does not imply \form{E_1{\neq}E_2}. Neither is \form{({\heap_m} \wedge\phi \wedge \a)}.
\item Case \form{\phi' \equiv \true}.
As \form{\enode'} is not derivable, then the side condition of rule  \rulename{Hypothesis}
does not hold. We consider two cases.
\begin{enumerate}
\item \form{\a \wedge \a'} is unsatisfiable. Hence, any model of \form{\a}
implies \form{\neg\a'}.
\item \form{\a \wedge \a'} is satisfiable. Hence, \form{\a \wedge \neg \a'} is also
satisfiable and is in NF. Moreover,
 any model of \form{\a \wedge \neg \a'}
implies \form{\neg\a'}. As  \form{\a \wedge \neg \a'} is an under-approximation
of  \form{\a}, from any model \form{\a \wedge \neg \a'} we can construct
a model satisfying
 \form{\a} implies \form{\neg\a'}.
\end{enumerate}
\end{enumerate}
Therefore, any bad model of \form{\basecomp{\heap_m} \wedge\phi \wedge \a} is a counter-model.
\hfill \qed.
\end{proof}

Secondly,  we prove the completeness of rule
\rulename{\sep} when the LHS of the conclusion in NF is a base formula.
\[
\AxiomC{
$\begin{array}{c}
    \entailNCyc{\heap_1{\wedge}\pure}{\heap_2} {\heap{\sep}\heap'} \qquad
    \entailNCyc{\heap{\wedge}\pure}{\heap'} {\heap{\sep}\heap'}
\end{array}$
}
\LeftLabel{\scriptsize ${\sep}$}
\UnaryInfC{
$\entailNCyc{\heap_1 {\sep} \heap{\wedge}\pure}{\heap_2 {\sep} \heap'}{\heap}$
}
\DisplayProof
\]
\begin{proof}
We prove that if the rule's conclusion is derivable then the rule's premises are derivable.

We prove by induction on the number \form{n} of points-to predicates in the LHS of the conclusion.

\begin{enumerate}
\item Base case: \form{n =0} and \form{n =1}, the proof is trivial.

\item Inductive case: Assume that it is true for \form{n=k}.

Suppose \form{\heap_1 \sep \heap=  \sepnodeF{x_1}{c_1}{\setvars{v}_1}{\sep}...{\sep}\sepnodeF{x_{k+1}}{c_{k+1}}{\setvars{v}_{k+1}}}
and \form{\pure} contains enough disequalities for NF.

 We proceed by cases on \form{\heap_2}.
 \begin{enumerate}
 \item Case 1: \form{\heap_2} is a points-to predicate.
If \form{\heap_2 \equiv \sepnodeF{x_j}{\_}{\_}} where \form{x_j \not\in \{x_1,...x_{k+1}\}}.
Then procedure \form{is\_closed}
  has also returned {\invalid} already and the 
the conclusion is not derivable. Contradiction.
Therefore, \form{\heap_2} must be one of the points-to predicates
 in the LHS. Assume that \form{\heap_2 \equiv \sepnodeF{x_1}{c_1}{\setvars{v}_1}}.
   Then, \form{\heap_1 \force\heap_2} and by induction \form{\heap \wedge \pure ~\ent~ \heap'} is also derivable.
 
\item Case 2: \form{\heap_2} is an inductive predicate; assume \form{\heap_2 \equiv \seppred{\code{P}}{r{,} F}}.
  Similarly to the above case,  \form{r\in \{x_1; ...;x_{k+1}\}}. Otherwise, procedure \form{is\_closed}
  has returned {\invalid} already. Assume \form{r\equiv x_1}.
 Secondly, \form{x_1{\neq}F \in \pure}. Otherwise,
  {\allEnt} is stuck (it could not apply rule \rulename{RInd}) and
  procedure \form{is\_closed}
  has also returned {\invalid} already.
 Third, \form{ \sepnodeF{x_1}{c_1}{\_} \not\in \heap'}. Otherwise, the RHS of the conclusion is
\form{\false} the conclusion is not derivable.
  Now, the conclusion could be applied with rule \rulename{RInd} to
  generate \form{ \sepnodeF{x_1}{c_1}{\_} \sep \heap'' \sep  \seppred{\code{P}}{F_2{,} F}}.
   Now, it comes back to Case 1 above.

\end{enumerate}
\end{enumerate}

\hfill \qed.
\end{proof}

\section{Completeness of proof search - Proposition \ref{complete}}\label{app.thm.complete}

We prove the correctness of Proposition \ref{complete} through two steps:
\begin{enumerate}
\item proofs for the case where LHS is a base formula. Those entailments
are reduced without \rulename{LInd}.
\item proofs for the case where LHS is a general formula. Those entailments
are reduced with \rulename{LInd} prior to applying other rules.
\end{enumerate}
In the proofs,
we make use of the following auxiliary Lemmas.

\begin{lemma}\label{lemma.fr.big}
If \form{\basecomp{\heap} \wedge \pure ~\ent~ \heap_1'\sep \heap_2'} in NF is derivable,
then there exist \form{\heap_1, \heap_2} such that \form{\heap \equiv \heap_1 \sep \heap_2}
and
both \form{\basecomp{\heap_1} \wedge \pure ~\ent~ \heap_1'}
and \form{\basecomp{\heap_2} \wedge \pure ~\ent~ \heap_2'} are derivable.
\end{lemma}

\begin{lemma}\label{lemma.fr.gen}
If \form{\basecomp{\heap_1} \sep \basecomp{\heap_2} \wedge \pure} is in NF
and \form{\basecomp{\heap_1} \wedge \pure ~\ent~ \heap_1'} is valid, then
\form{\basecomp{\heap_1} \sep \basecomp{\heap_2} \wedge \pure ~\ent~ \heap_1'\sep \heap_2'} is valid
iff \form{\basecomp{\heap_2} \wedge \pure ~\ent~ \heap_2'} is valid.
\end{lemma}
Based on the fact that heaps of a normalized base formula is precise.
The proof is straightforward based on the semantics of the separating conjunction \form{\sep}.


\subsection{Base-Formula LHS}

First, we show the correctness of case 2a) of procedure \code{is\_closed} i.e., an entailment is stuck then
it is invalid. After that, we show the invalidity is preserved through proof search.

As the LHS is a base formula, rule  \rulename{LInd} (and rule  \rulename{LBase}) is never be applied.
We prove case 2a) by induction on the number of disequalities missing from the LHS
and generated by rule  \rulename{ExM}.
First, we prove the case where the RHS is an occurrence of
compositional predicate \code{P} assuming that the points-to predidcate in the definition
of \code{P} is \code{c}.
\begin{lemma}\label{lemma.complete.baselhs.predrhs}
If
\form{\enode_0{:}~\entailNCyc{\sepnode{x}{c}{F_2{,}\setvars{p}{,}u} {\sep} \basecomp{\heap} \wedge\phi \wedge \a}{\seppred{\code{P}}{x{,} F{,} \setvars{B}{,}u{,}sc{,}tg}}{\heap_m}}
is in NF and is stuck, then it is invalid.
\end{lemma}
\begin{proof}
Due to the stuckness, {\allEnt} could not applies rule \rulename{RInd}.
Hence, \form{x{\neq}F {\not}{\in} \phi}.
As the the entailment is in NF,  \form{\entailNCyc{(\sepnode{x}{c}{F_2{,}\setvars{p}{,}u} {\sep} \basecomp{\heap} \wedge\phi)\subst{F}{x} \wedge \a}{\seppred{\code{P}}{x{,} F{,} \setvars{B}{,}u{,}sc{,}tg}\subst{F}{x}}{\heap_m}} is in NF (by Lemma \ref{lemma.nf.subst}).
As all models satisfying the LHS \form{(\sepnode{x}{c}{F_2{,}\setvars{p}{,}u} {\sep} \basecomp{\heap} \wedge\phi)\subst{F}{x} \wedge \a} are non-empty heap and in NF, all models satisfying
the RHS \form{\seppred{\code{P}}{x{,} F{,} \setvars{B}{,}u{,}sc{,}tg}\subst{F}{x} \equiv \seppred{\code{P}}{F{,} F{,} \setvars{B}{,}u{,}sc{,}tg}}
are empty heap, this entailment is invalid.
As the substitution law is sound and complete, \form{\enode_0} is invalid. 
\hfill \qed.
\end{proof}

\begin{lemma}\label{lemma.complete.baselhs}
If
\form{\enode_0{:~}\entailNCyc{ \basecomp{\heap} \wedge\phi \wedge \a}{\heap'}{\heap_m}}
is in NF and is stuck, then it is invalid.
\end{lemma}
\begin{proof}
By induction on the number of disequalities missing from \form{\phi}.
We proceed by cases.
\begin{enumerate}
\item \form{\heap' \equiv \seppred{\code{P}}{x{,} F{,} \setvars{B}{,}u{,}sc{,}tg} \sep \heap''}  assuming that the points-to predidcate in the definition
of \code{P} is \code{c}.
\begin{enumerate}
\item \form{op(x) {\not}{\in} \basecomp{\heap}}. This case is the case 2c) of procedure
\code{is\_closed}. The bad model of the LHS is the counter-model.
\item \form{op(x) {\in} \basecomp{\heap}}.
{\allEnt} reduces the entailments by first applying rule \rulename{ExM} prior to
applying rule \rulename{LInd}. We are considering the case {\allEnt} could not apply
rule \rulename{LInd}.
We proceed cases for the LHS.
\begin{enumerate}
\item \form{\heap \equiv \sepnode{x}{c}{F_2{,}\setvars{v}}{\sep} \basecomp{\heap_0}} and
\form{x{\neq}F {\not}{\in} \phi}.

If \form{\enode_1{:~} \entailNCyc{\sepnode{x}{c}{F_2{,}\setvars{v}}{\sep} \basecomp{\heap_0} \wedge\phi \wedge \a {\wedge} {\bf x{\neq}F} }{\seppred{\code{P}}{x{,} F{,} \setvars{B}{,}u{,}sc{,}tg} \sep \heap''}{\heap_m}} is stuck, rule \rulename{LInd} could not be applied. Hence, \form{\sepnode{x}{c}{F_2{,}\setvars{v}} \in \heap''}.
Therefore, there is no model that satisfies the RHS. This entailment \form{\enode_1} is thus invalid.
As rule \rulename{ExM} is complete (Lemma \ref{rule.complete}),  \form{\enode_0} is invalid.

If \form{\enode_1} is deriable, following Lemma \ref{lemma.fr.big}, there exist
\form{\heap_1, \heap_2} such that
\form{\heap_0 \equiv \heap_1 \sep \heap_2} and both \form{\enode_2{:~} \entailNCyc{\sepnode{x}{c}{F_2{,}\setvars{v}}{\sep} \basecomp{\heap_1} \wedge\phi \wedge \a {\wedge} {\bf x{\neq}F} }{\seppred{\code{P}}{x{,} F{,} \setvars{B}{,}u{,}sc{,}tg}}{\heap_m}} and
\form{\enode_3{:~} \entailNCyc{\basecomp{\heap_2} \wedge\phi \wedge \a {\wedge} {\bf x{\neq}F} }{\heap''}{\heap_m}} are derivable. We proceed two sub-cases:
\begin{enumerate}
\item \form{\enode_3'{:~} \entailNCyc{\basecomp{\heap_2} \wedge\phi \wedge \a }{\heap''}{\heap_m}} is stuck. Hence either \form{op_1{x} \in \basecomp{\heap_2}} or \form{op_2(F) \in \basecomp{\heap_2}}. This implies either \form{op{x}\sep op_1{x}} in the LHS of \form{\enode_0}
or  \form{op{x}\sep op_3{F}} in the LHS of \form{\enode_0}. As \form{\enode_0} is in NF,
either \form{x{\neq}x \in \phi} or \form{x{\neq}F \in \phi}. Both can't not happen
as the first scenario contradicts with assumption that LHS is in LHS and the second one
contradicts with assumption \form{x{\neq}F {\not}{\in} \phi}.
\item \form{\enode_3'{:~} \entailNCyc{\basecomp{\heap_2} \wedge\phi \wedge \a }{\heap''}{\heap_m}} is derivable. Hence, by soundness (Lemma \ref{ent.global.sound}), it is valid. (2a)

As \form{\enode_0} is stuck and \form{\enode_3'} is derivable, we deduce
 that \form{\enode_2'{:~} \entailNCyc{\sepnode{x}{c}{F_2{,}\setvars{v}}{\sep} \basecomp{\heap_1} \wedge\phi \wedge \a  }{\seppred{\code{P}}{x{,} F{,} \setvars{B}{,}u{,}sc{,}tg}}{\heap_m}}
is stuck (Otherwise, \form{\enode_0} is derivable as well, contradition).
By Lemma \ref{lemma.complete.baselhs.predrhs}, \form{\enode_2'} is invalid. (2b)

By (2a), (2b) and Lemma \ref{lemma.fr.gen}, \form{\enode_0} is invalid.
\end{enumerate}

\item \form{\heap \equiv \sepnode{x}{c'}{F_2{,}\setvars{v}}{\sep} \basecomp{\heap_0}},
\form{x{\neq}F \in \phi} and
\form{c'{\not}{\equiv}c}. The proof is similar to Case 2d of procedure \code{is\_closed}.
The bad model of the LHS is the counter-model.
\end{enumerate}
\end{enumerate}
\item \form{\heap' \equiv \sepnode{x}{c}{\setvars{v}} \sep \heap''}. Straightforward.
\item \form{\heap' \equiv \emp}. Straightforward.
\end{enumerate}
\hfill \qed.
\end{proof}

\begin{proposition}\label{complete.baselhs}
If
\form{\entailNCyc{\basecomp{\heap} \wedge\phi \wedge \a}{\heap'}{\heap_m}}
is in NF and is not derivable, then it is invalid.
\end{proposition}
\begin{proof}
In a incomplete proof, if a leaf node in NF is stuck then it is invalid
(Lemma \ref{lemma.complete.baselhs} and Lemma \ref{complete.pure}).
The invalidity is preserved up to the root based on  Lemma \ref{rule.complete}.
\hfill \qed.
\end{proof}

\subsection{General LHS}
By induction on the RHS. We proceed cases on the RHS.
In the proofs, for convenient, we write \form{\entailCyc{\heap_a\wedge\pure_a}{\heap_c\wedge \pure_c}{\heap}}
as a shorthand of \form{\entailNCyc{\heap_a\sep \heap \wedge\pure_a}{\heap_c\sep \heap \wedge\pure_c}{\heap}}
and and no any matching heaps between \form{\heap_a} and \form{\heap_c} could be found through
 the application of rule \rulename{\sep}.

\hide{
First, we consider the case where  an occurrence of compositional predicates is in the RHS and {\allEnt} is stuck at \rulename{LInd}
due to the unsatisfied side conditions.

\begin{lemma}\label{lemma.complete.rhs.pred}
If
\form{\enode{:}~\entailCyc{\heap\wedge\pure}{\seppredF{\code{Q}}{x{,} F_3{,} \setvars{B}{,}u{,}sc{,}tg_3} {\sep} \heap'}{\heap_m}} in NF where
\[
\form{\enode_0{:}~\entailNCyc{\heap\wedge\pure}{\seppredF{\code{Q}}{x{,} F_3{,} \setvars{B}{,}u{,}sc{,}tg_3} {\sep} \heap'}{\heap_m}}
\] is stuck
and \form{\seppredF{\code{Q}}{x{,} F_3{,} \setvars{B}{,}u{,}sc{,}tg_3}  {\not}{\in} \heap},  then \form{\enode_0}  is invalid.
\end{lemma}

\begin{proof}
We first show \form{\enode_0} is invalid. After that we apply Lemma \ref{lemma.fr.gen}
to show that \form{\enode} is invalid. To show invalidity of \form{\enode_0},
we proceed cases of the possible base formula of the LHS.
\begin{enumerate}
\item \form{\enode_1{:}~\entailNCyc{\basecomp{\heap}\wedge\pure}{\seppredF{\code{Q}}{x{,} F_3{,} \setvars{B}{,}u{,}sc{,}tg_3} {\sep} \heap'}{\heap_m}} is stuck.
By Proposition \ref{complete.baselhs}, \form{\enode_1} is invalid.
As
\form{\basecomp{\heap}\wedge\pure} is an approximation of \form{{\heap}\wedge\pure},
\form{\enode_0} is invalid.
\item \form{\enode_2{:}~\entailNCyc{\basecomp{\heap}\wedge\pure}{\seppredF{\code{Q}}{x{,} F_3{,} \setvars{B}{,}u{,}sc{,}tg_3} {\sep} \heap'}{\heap_m}} is derivable.
As the LHS is in NF, \form{\enode_2} could be reduced by \rulename{RInd}.
This implies that \form{\sepnodeF{x}{c}{F{,}\setvars{d},tg{,}u}\subst{\setvars{v}}{\setvars{d}}{\wedge} \pure_0\subst{tg}{scd}  \in \basecomp{\heap}} and \form{x{\neq}F_3 \in \pure}.
As \form{\enode_0} is stuck, \form{\sepnodeF{x}{c}{F{,}\setvars{d},tg{,}u}\subst{\setvars{v}}{\setvars{d}}  \not\in {\heap}}
(otherwise, \rulename{LInd} could be applied). Instead, \form{\seppredF{\code{P}}{x{,} F{,} \setvars{B}{,}u{,}sc{,}tg} \in \heap} and the side conditions
of \rulename{LInd} does not hold. Particularly, we assume that
 \form{\heap \equiv \seppredF{\code{P}}{x{,} F{,} \setvars{B}{,}u{,}sc{,}tg} \sep \heap_0},
 and \form{op(F_3) \not\in\heap_0}.

In intuition, if \form{\enode_0} had been valid (and \form{op(F_3) \in\heap_0}) \form{F} would have been a points-to predicate in the chain
between \form{x} and \form{F_3} which would have been a dangling pointer.
Here, \form{\enode_0} is invalid (and \form{op(F_3) \not\in\heap_0}),
there are two ways to show that \form{\enode_0} is invalid.
\begin{enumerate}
\item \form{x{=}F_3}. \form{\seppredF{\code{Q}}{x{,} F_3{,} \setvars{B}{,}u{,}sc{,}tg_3}} is reduced to the base case while
\form{\seppredF{\code{P}}{x{,} F{,} \setvars{B}{,}u{,}sc{,}tg}} is not empty.
\item \form{F_3} is a points-to predicate in the chain
between \form{x} and \form{F}.
\end{enumerate}
As LHS of \form{\enode_3} is in NF, we show the invalidity following the second way.

Replacing the above assumption, we obtain
\[\form{\enode_2{:}~\entailNCyc{\basecomp{\seppredF{\code{P}}{x{,} F{,} \setvars{B}{,}u{,}sc{,}tg}} \sep \basecomp{\heap_0}\wedge\pure}{\seppredF{\code{Q}}{x{,} F_3{,} \setvars{B}{,}u{,}sc{,}tg_3} {\sep} \heap'}{\heap_m}}
\]

As \form{\enode_2} is derivable, by Lemma \ref{lemma.fr.big}, there exist
\form{\heap_1, \heap_2} s.t. \form{\heap_0 \equiv \heap_1\sep \heap_2} and
both 
  \form{\enode_3{:}~\entailNCyc{\basecomp{\seppredF{\code{P}}{x{,} F{,} \setvars{B}{,}u{,}sc{,}tg}}\sep \basecomp{\heap_1} \wedge\pure}{\seppredF{\code{Q}}{x{,} F_3{,} \setvars{B}{,}u{,}sc{,}tg_3} }{\heap_m}}
and  \form{\enode_4{:}~\entailNCyc{ \basecomp{\heap_2}\wedge\pure}{\heap'}{\heap_m}}
are derivable.
By soundness (Lemma \ref{ent.global.sound}),
\form{\enode_3} is valid (3a) and
\form{\enode_4} is valid (3b).

From (3a), for any \form{\sstack, \sheaps} such that \form{\sstack, \sheaps \force \basecomp{\seppredF{\code{P}}{x{,} F{,} \setvars{B}{,}u{,}sc{,}tg}}\sep \basecomp{\heap_1}\wedge\pure},
\form{\sstack, \sheaps \force \seppredF{\code{Q}}{x{,} F_3{,} \setvars{B}{,}u{,}sc{,}tg_3}}.
Furthermore, \form{op(F_3) \not\in\heap_0} together
with semantics of compositional predicate imply
\[
\begin{array}{l}
\form{\enode_3'{:}~{\sepnodeF{x}{c}{F_3{,}\setvars{d},tg{,}u}\subst{\setvars{v}}{\setvars{d}} \sep \basecomp{\heap'\subst{\setvars{v}}{\setvars{d}}} \sep \basecomp{\heap_1} \sep \sepnodeF{F_3}{c}{F{,}\setvars{d}',tg_3{,}u}\wedge\pure{\wedge} \pure' \wedge \pure_0\subst{tg}{scd}} \\
\qquad~\ent_{{\heap_m}}~ {\seppredF{\code{Q}}{x{,} F_3{,} \setvars{B}{,}u{,}sc{,}tg_3}}}
\end{array}
\] is invalid (4). Here \form{\pure'} captures all disequalities to form NF.

Applying Lemma \ref{lemma.fr.gen} into (3b) and (4), we infer that
\form{\enode_{3b4}{:}~\entailNCyc{\sepnodeF{x}{c}{F_3{,}\setvars{d},tg{,}u}\subst{\setvars{v}}{\setvars{d}} \sep \basecomp{\heap'\subst{\setvars{v}}{\setvars{d}}} \sep \basecomp{\heap_1} \sep \basecomp{\heap_2} \sep \sepnodeF{F_3}{c}{F{,}\setvars{d}',tg_3{,}u}\wedge\pure{\wedge} \pure' \wedge \pure_0\subst{tg}{scd}}{\seppredF{\code{Q}}{x{,} F_3{,} \setvars{B}{,}u{,}sc{,}tg_3} \sep \heap'}{\heap_m}}
is invalid (5).

Since \form{\overline{\heap_i}} is an under-approximation of \form{\heap_i} (\form{i\in\{1,2\}}),
\form{\enode_{5}{:}~\entailNCyc{\sepnodeF{x}{c}{F_3{,}\setvars{d},tg{,}u}\subst{\setvars{v}}{\setvars{d}} \sep \basecomp{\heap'\subst{\setvars{v}}{\setvars{d}}} \sep \basecomp{\heap_1} \sep \basecomp{\heap_2} \sep \sepnodeF{F_3}{c}{F{,}\setvars{d}',tg{,}u}\wedge\pure{\wedge} \pure' \wedge \pure_0\subst{tg}{scd}}{\seppredF{\code{P}}{x{,} F{,} \setvars{B}{,}u{,}sc{,}tg} \sep \heap_1 \sep \heap_2}{\heap_m}}
is valid (6).

Both LHS of \form{\enode_{3b4}} and \form{\enode_5} are in NF, there exist
\form{\sstack, \sheaps} such that
\[\form{\sstack, \sheaps \force \sepnodeF{x}{c}{F_3{,}\setvars{d},tg{,}u}\subst{\setvars{v}}{\setvars{d}} \sep \basecomp{\heap'\subst{\setvars{v}}{\setvars{d}}} \sep \basecomp{\heap_1} \sep \basecomp{\heap_2} \sep \sepnodeF{F_3}{c}{F{,}\setvars{d}',tg{,}u}\wedge\pure{\wedge} \pure' \wedge \pure_0\subst{tg}{scd}}\]

From (6), \form{\sstack, \sheaps \force \seppredF{\code{P}}{x{,} F{,} \setvars{B}{,}u{,}sc{,}tg} \sep \heap_1 \sep \heap_2}.

From (5), \form{\sstack, \sheaps ~{\not\force}~ \seppredF{\code{Q}}{x{,} F_3{,} \setvars{B}{,}u{,}sc{,}tg_3} \sep \heap'}.

Hence, \form{\sstack, \sheaps} is a counter-model of the entailment:
\[\form{\seppredF{\code{P}}{x{,} F{,} \setvars{B}{,}u{,}sc{,}tg} \sep \heap_1 \sep \heap_2 \ent  \seppredF{\code{Q}}{x{,} F_3{,} \setvars{B}{,}u{,}sc{,}tg_3} \sep \heap'}
\]
Thus, together with Lemma \ref{lemma.fr.gen},
\form{\enode_0} is invalid.

\end{enumerate}
\hfill \qed.
\end{proof}}

\begin{lemma}\label{lemma.complete.rhs}
If
\form{\enode{:}~\entailCyc{\heap\wedge\pure}{ \heap'}{\heap_m}} in NF where
\form{\enode_0{:}~\entailNCyc{\heap\wedge\pure}{ \heap'}{\heap_m}} is stuck,
then \form{\enode_0}  is invalid.
\end{lemma}
\begin{proof}
We first show \form{\enode_0} is invalid. After that by using Lemma \ref{lemma.fr.gen},
we could deduce the invalidity of \form{\enode}. To show invalidity of \form{\enode_0},
we proceed cases on the possible base formula of the LHS.
\begin{enumerate}
\item \form{\enode_1{:}~\entailNCyc{\basecomp{\heap}\wedge\pure}{ \heap'}{\heap_m}} is stuck.
By Proposition \ref{complete.baselhs}, \form{\enode_1} is invalid. As
\form{\basecomp{\heap}\wedge\pure} is an approximation of \form{{\heap}\wedge\pure},
\form{\enode_0} is invalid.
\item \form{\enode_2{:}~\entailNCyc{\basecomp{\heap}\wedge\pure}{\heap'}{\heap_m}} is derivable.
And \form{\heap' \equiv op(E) \sep \heap''}. We proceed cases on \form{op(E)}.
\begin{itemize}
\item \form{op(E) \equiv \seppredF{\code{Q}}{x{,} F_3{,} \setvars{B}{,}u{,}sc{,}tg_3}}.
As the LHS is in NF, \form{\enode_2} could be reduced by \rulename{RInd}.
This implies that \form{\sepnodeF{x}{c}{F{,}\setvars{d},tg{,}u}\subst{\setvars{v}}{\setvars{d}}{\wedge} \pure_0\subst{tg}{scd}  \in \basecomp{\heap}} and \form{x{\neq}F_3 \in \pure}.
This implies that there are two possible sub-cases.
\begin{enumerate}
\item Sub-case 1: \form{\sepnodeF{x}{c}{F{,}\setvars{d},tg{,}u}\subst{\setvars{v}}{\setvars{d}}  \in {\heap}}.
As \form{x{\neq}F_3 \in \pure}, \form{\enode_0} could be applied with \rulename{RInd}. It is impossible as it contradicts with the assumption that \form{\enode_0} is stuck.
\item Sub-case 2: \form{\seppredF{\code{P}}{x{,} F{,} \setvars{B}{,}u{,}sc{,}tg} \in \heap}.
As \form{x{\neq}F_3 \in \pure}, \form{\enode_0} could be applied with \rulename{LInd}. As \form{x{\neq}F_3 \in \pure}, \form{\enode_0} could be applied with \rulename{LInd}.
\end{enumerate}
\item  \form{op(E) \equiv \sepnodeF{x}{c}{next:F,\setvars{v}}}.
Based on \form{\sepnodeF{x}{c}{next:F,\setvars{v}} \in \basecomp{\heap}},
there are two cases.
\begin{enumerate}
\item \form{\sepnodeF{x}{c}{next:F, \setvars{v}} \in \heap}. This contradicts
with the assumption that \form{\sepnodeF{x}{c}{next:F, \setvars{v}}} could not be matched
with any predicate in \form{\heap}. This case is impossible.
\item \form{\seppredF{\code{P}}{x{,} E, ..}  \in \heap}. Any model satisfying
the LHS when replacing  \form{\seppredF{\code{P}}{x{,} E, ..}} by three-time unfolding (with two
points-to predicates e.g.,  \form{\sepnodeF{x}{c}{next:F_1, ..} \sep \sepnodeF{F_1}{c}{next:F,..}} is a counter-model.
\end{enumerate}
\end{itemize}
\end{enumerate}
\hfill \qed.
\end{proof}

\begin{proposition}[Incompleteness Preservation]\label{complete.gen}
Given an input entailment
\form{\enode_0{:}~\entailNCyc{\D}{\D'}{\heap_m}},
and there is an leaf node \form{\enode_i{:}~\entailCyc{{\D_l}}{\D'_l}{\heap_m}}
in its incomplete proof tree 
where
\begin{itemize}
\item the leaf node \form{\enode_i} is in NF; and
\item none of application of rule \code{\scriptsize FR} from the root
\form{\enode_0} to the leaf node \form{\enode_i}; and
\item \form{\entailNCyc{{\D_l}}{\D'_l}{\heap_m}} is not derivable.
\end{itemize}
then \form{\enode_0} is invalid.
\end{proposition}
\begin{proof}
By Lemma \ref{lemma.complete.rhs}, \form{\enode_i} is invalid.
By Lemma \ref{rule.complete}, \form{\enode_0} is invalid.
\hfill \qed.
\end{proof}

 }{}

\end{document}